\documentclass[12pt,onecolumn,twoside]{IEEEtran}

\usepackage{calc}

\usepackage{multirow}
\usepackage{cite}
\usepackage{graphicx,subfigure}
\usepackage{psfrag}
\usepackage{amsmath,amssymb,amsthm}
\usepackage{color}
\usepackage{tikz} 

\usepackage{bbm}

\usepackage{cases}

\DeclareMathOperator*{\defeq}{\triangleq}

\newtheorem{theorem}{Theorem}
\newtheorem{corollary}{Corollary}

\newtheorem{lemma}{Lemma}

\newcommand{\bit}{\begin{itemize}}
\newcommand{\eit}{\end{itemize}}

\newcommand{\bc}{\begin{center}}
\newcommand{\ec}{\end{center}}

\newcommand{\ba}{\begin{array}}
\newcommand{\ea}{\end{array}}

\newcommand{\beq}{\begin{equation}}
\newcommand{\eeq}{\end{equation}}

\newcommand{\beqn}{\begin{equation*}}
\newcommand{\eeqn}{\end{equation*}}

\newcommand{\bean}{\begin{eqnarray*}}
\newcommand{\eean}{\end{eqnarray*}}
\newcommand{\bea}{\begin{eqnarray}}
\newcommand{\eea}{\end{eqnarray}}

\def\E{\mathbb{E}}

\def\F{\mathbb{F}}

\def\sv{\boldsymbol{s}}

\newcommand{\Lc}{{\mathcal L}}

\newcommand{\Nc}{{\mathcal N}}

\newcommand{\Rc}{{\mathcal R}}
\newcommand{\Sc}{{\mathcal S}}

\newcommand{\Zc}{{\mathcal Z}}
\newcommand{\T}{{\scriptscriptstyle\mathsf{T}}}

\newcommand{\non}{\nonumber}

\newcommand{\Hen}{\mathbb{H}}
\newcommand{\hen}{\mathrm{h}}
\newcommand{\Imu}{\mathbb{I}}

\newcommand{\bln}{n}

\newcommand{\Ho}{\mathcal{H}_{\text{out}}}
\newcommand{\Hob}{\bar{\mathcal{H}}_{\text{out}}}

\IEEEoverridecommandlockouts

\begin{document}
\sloppy

\title{ On the Optimality of Secure Communication Without Using Cooperative Jamming}
\title{Secure Communication over Interference Channel: To Jam or Not to Jam?}
\author{Jinyuan Chen 
\thanks{Jinyuan Chen is with Louisiana Tech University, Department of Electrical Engineering, Ruston, LA 71272, USA (email: jinyuan@latech.edu).  The work was partly supported by Louisiana Board of Regents Support Fund (BoRSF) Research Competitiveness Subprogram (RCS) under grant 32-4121-40336.
This work was presented in part at the 54th Annual Allerton Conference on Communication, Control, and Computing, 2016 and the 56th Annual Allerton Conference on Communication, Control, and Computing, 2018.
} 
}

\maketitle
\pagestyle{headings}

\begin{abstract}
We consider a secure communication over a two-user Gaussian interference channel, where each transmitter sends a \emph{confidential} message  to its legitimate receiver. 
For this setting,  we identify a regime where  the simple scheme of using Gaussian wiretap codebook at each transmitter (without cooperative jamming) and treating interference as noise at each intended receiver (in short, GWC-TIN scheme) achieves the optimal secure sum capacity to within a constant gap. 
For the symmetric case, this simple scheme is optimal when the interference-to-signal ratio (all link strengths in decibel scale) is no more than $2/3$. However, when the ratio is more than $2/3$, we show that this simple scheme is not optimal anymore and a scheme with cooperative jamming is proposed to achieve the optimal secure sum generalized degrees-of-freedom (GDoF).  
Specifically, for the symmetric case, we complete the optimal secure sum GDoF characterization for all the interference regimes.

\end{abstract}

\section{Introduction}

The notion of information-theoretic secrecy was first introduced by Shannon  in his seminal work \cite{Shannon:49}, which studied a secure communication in the presence of a private key that is revealed to both transmitter and legitimate receiver  but not to the eavesdropper. 
Later,  Wyner introduced the notion of secure capacity via  a degraded  wiretap channel, in which a  transmitter intends to send a confidential message to a legitimate receiver by hiding it from a degraded eavesdropper \cite{Wyner:75}. 
The secure capacity is the maximum rate at which the confidential message can be transmitted reliably and securely to the legitimate receiver. 
Wyner's result  was subsequently generalized to the non-degraded wiretap channel by Csisz\`ar and  K{\"o}rner \cite{CsiszarKorner:78}, and the Gaussian wiretap channel by Leung-Yan-Cheong and Hellman \cite{CH:78}.
This line of  secure capacity research has been extended to many multiuser channels, most notably,  broadcast  channels \cite{LMSY:08, LLPS:10, XCC:09,ChiaGamal:12,KTW:08}, multiple access channels  \cite{TY:08cj, TekinYener:08d, LP:08, LLP:11, KG:15, HKY:13},  and interference channels \cite{LMSY:08,LBPSV:09,LYT:08,HY:09,YTL:08, PDT:09, KGLP:11, XU:14, XU:15, MDHS:14, MM:14o, MM:16, GTJ:15, GJ:15, MXU:17}.

In the line of secure capacity research, cooperative jamming has been proposed extensively to improve the achievable secure rates in many channels (see \cite{TY:08cj,LMSY:08, XU:14, XU:15} and references therein). 
In particular,  cooperative jamming has been proposed in \cite{XU:14} and \cite{XU:15} to achieve the optimal secure sum degrees-of-freedom (DoF) in the interference channel with confidential messages,  wiretap channel with helpers, multiple access wiretap channel, and broadcast channel with confidential messages. 
The basic idea of the cooperative jamming scheme is to send jamming signals to confuse the potential eavesdroppers, while keeping legitimate receivers' abilities to decode the desired messages. This might involve a cooperation between the transmitters, and a careful design on the \emph{direction} and/or \emph{power}  of the cooperative jamming signals (see \cite{TY:08cj, LMSY:08, XU:14, XU:15}).  It is therefore implicit that the cooperative jamming schemes might  incur  some extra overhead, e.g., due to network coordination,  channel state information (CSI) acquisition, and power consumption.
In this work we seek to understand when it is  necessary to use cooperative jamming and when it is not,  for the secure communication over the interference channel.

Specifically, we focus on a secure communication over a two-user Gaussian interference channel, where each transmitter sends a \emph{confidential} message  to its legitimate receiver.  
For this setting, we identify a regime in which  the simple scheme of using Gaussian wiretap codebook at each transmitter, without cooperative jamming, and treating interference as noise at each intended receiver (in short, GWC-TIN scheme) achieves the optimal secure sum capacity to within a constant gap. 
The secrecy offered by this GWC-TIN scheme is information-theoretic secrecy, which holds for any decoding methods at any unintended receiver (eavesdropper). 
In this simple scheme, the transmitters do \emph{not} need to know all the information of channel realizations. Therefore, the overhead associated with acquiring  channel state information at the transmitters (CSIT) is minimal for the  GWC-TIN scheme.
For the symmetric case, this simple scheme is optimal when the interference-to-signal ratio (all link strengths in decibel scale), denoted by $\alpha$,  is no more than $2/3$, i.e., $0 \leq \alpha \leq 2/3$. However, when the ratio is more than $2/3$, we show that this simple scheme is not optimal anymore and a scheme with cooperative jamming is proposed to achieve the optimal secure sum generalized degrees-of-freedom (GDoF).

{\bf Some related works:}  For a two-user Gaussian interference channel \emph{without} secrecy constraints, the capacity has been characterized to within one bit in \cite{ETW:08} for a decade.  Our work focuses on the secure capacity approximation for the case with secrecy constraints. Specifically, for the symmetric case this work completes the full characterization on secure sum GDoF for all the interference regimes. Furthermore, this work identifies a regime in which the low-complexity scheme, i.e., the GWC-TIN scheme, achieves the optimal secure sum capacity to within a constant gap.  
The low-complexity communication schemes have received significant attention in the literature (see, e.g., \cite{GNAJ:13}). Specifically, the  work in \cite{GNAJ:13} considered the Gaussian interference channel \emph{without} secrecy constraints and investigated the optimality of the low-complexity scheme that treats interference as noise (TIN scheme), in terms of GDoF region and constant-gap capacity region.
For the Gaussian interference channel with secrecy constraints, the previous work in \cite{GJ:15} showed that secrecy constraints incur no penalty in GDoF, i.e., the secure GDoF region and the GDoF region remain the same  under the condition in which the TIN scheme was proved to be GDoF-optimal (cf.~\cite{GNAJ:13}).  
Note that for the two-user symmetric Gaussian interference channel, the condition identified in \cite{GNAJ:13} for the TIN scheme to be GDoF-optimal---the same condition in which secrecy constraints incur no penalty in GDoF (cf.~\cite{GJ:15})---is simplified to  $0 \leq \alpha \leq 1/2$.  Interestingly,  for the two-user symmetric Gaussian interference channel \emph{with secrecy constraints},  our work shows that the proposed low-complexity scheme  (GWC-TIN scheme) is optimal in terms of secure sum GDoF \emph{if and only if} $0 \leq \alpha \leq 2/3$.  When $ \alpha > 2/3$, our work also characterizes  the optimal secure sum GDoF, which is achievable by  the proposed scheme using interference alignment and  cooperative jamming. In addition to the new achievability scheme, our work also provides a new converse to prove the results. 
In a related work in \cite{GTJ:15}, the authors considered the two-user symmetric  deterministic interference channel  and provided the upper bound and lower bound on the secure symmetric capacity. Note that, there is still a gap between the upper bound and lower bound derived in  \cite{GTJ:15},  when the interference-to-signal ratio $\alpha$ (all link strengths now measured in bits) is  $1/2 < \alpha < 3/4$.  In our work in \cite{allerton:16} (conference version), we closed the gap by providing a new capacity upper bound. 
In a related work in \cite{MM:16}, the authors considered the secure communication over a two-user symmetric interference channel with \emph{transmitter cooperation}, and provided the secure capacity upper bound and lower bound.  Note that, for the deterministic channel model and removing the transmitter cooperation, the  upper bound and lower bound derived in \cite{MM:16} are not matched for some channel parameters. We believe that our optimal secure GDoF results derived in our setting can be extended to solve the open problems in the other communication channels, including the channel considered in \cite{MM:16}.

The remainder of this work is organized as follows. 
Section~\ref{sec:system} describes the  system model and the simple scheme without cooperative jamming.  
Section~\ref{sec:mainresult} provides the  main results of this work. 
A scheme with cooperative jamming is proposed in Sections~\ref{sec:CJGau}. Some analysis on the cooperative jamming scheme  is provided in  Section~\ref{sec:rateerror341} and the appendices. 
The converse proof is provided in  Section~\ref{sec:converse} and the appendices.
The work is concluded in Section~\ref{sec:conclusion}.  
Throughout this work, $\Imu(\bullet)$, $\Hen(\bullet)$ and $\hen(\bullet)$ denote the mutual information, entropy and differential entropy,  respectively.  
$(\bullet)^\T$ and $(\bullet)^{-1}$ denote the transpose  and inverse operations, respectively.  
$\F^{q}_{2}$ denotes a set of $q$-tuples of binary numbers.  
     $\Zc$,  $\Zc^+$ and $\Rc$  denote the sets of integers, positive integers, and real numbers, respectively.   
        $o(\bullet)$ comes from the standard Landau notation, where  $f(x)=o(g(x))$ implies that $\lim_{x \to \infty} f(x)/g(x) =0$. 
$(\bullet)^+= \max\{0, \bullet\}$. Logarithms are in base~$2$.  
Unless for some specific parameters,  matrix,  scalar, and  vector are usually denoted by the  italic uppercase symbol (e.g., $S$),  italic lowercase symbol (e.g., $s$), and  the bold italic lowercase symbol  (e.g., $\sv$), respectively.
$s \sim \mathcal{N}(0, \sigma^2)$  denotes that the random variable $s$ has a normal distribution with zero mean and $\sigma^2$ variance.

\section{System model and preliminaries  \label{sec:system} }

This section provides the system model and discusses the  simple scheme without using cooperative jamming.

\begin{figure}[t!]
\centering
\includegraphics[width=6cm]{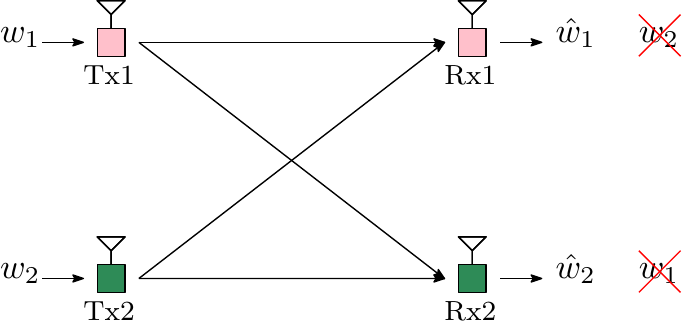}
 \vspace{-.05 in}
\caption{Two-user interference channel with confidential messages.}
\label{fig:IC_secrecy}
\end{figure}

\subsection{Gaussian interference channel  \label{sec:systemGaussian} }

We begin with a two-user $(K=2)$ Gaussian interference channel (see Fig.~\ref{fig:IC_secrecy}).  By following the convention in \cite{NM:13},  the channel output at receiver~$k$ at time~$t$ is given by  
\begin{align}
y_{k}(t) &= \sum_{\ell=1}^{K}  2^{m_{k\ell}} h_{k\ell} x_{\ell}(t) +z_{k}(t), \quad \quad k=1,\cdots,K  \label{eq:channelG101} 
\end{align}
$t=1,2, \cdots, \bln $, where  $x_{\ell}(t)$ is the channel input at transmitter~$\ell$  subject to a normalized power constraint  $\E |x_{\ell}(t)|^2 \leq 1$, $z_k(t) \sim \mathcal{N}(0, 1)$ is additive white Gaussian noise at receiver~$k$,  $m_{k\ell}$ is a nonnegative integer, and  $h_{k\ell} \in (1, 2]$,  for $k, \ell \in\{1,\cdots, K\}$.  Note that all the channel gains greater than or equal to one can be represented in the form of $2^{m_{k\ell}} h_{k\ell}$. Thus, the model in \eqref{eq:channelG101} can represent the general cases of the channels  in terms of capacity approximations.
Let us define  
\begin{align}
\alpha_{k\ell} &\defeq  \frac{\log 2^{m_{k\ell}}}{ \frac{1}{2}\log P}  \quad \quad k, \ell \in\{1,\cdots, K\}  \label{eq:alpha11} 
\end{align}
where $P \defeq \max_{k}\{ 2^{2 m_{kk}} \}$. 
Then, the original channel model  can be rewritten by 
\begin{align}
y_{k}(t) &= \sum_{\ell =1}^{K}  \sqrt{P^{\alpha_{k\ell}}} h_{k\ell} x_{\ell}(t) +z_{k}(t), \quad \quad k=1,\cdots,K  \label{eq:channelG} 
\end{align}
where the exponent  $\alpha_{k\ell} \geq 0$ represents the \emph{channel strength} of the link between transmitter~$\ell$ and receiver~$k$, and 
$h_{k\ell} \in (1, 2]$ represents  the normalized channel coefficient (we call it as \emph{channel phase}).  
In what follows, we will consider the channel model in \eqref{eq:channelG}. 
It is assumed that  each node  knows all the channel strengths and phases.  However,  in the  scheme without using cooperative jamming  that will be discussed later on, the transmitters do not need to know  the channel phases $\{h_{k\ell} \}_{k, \ell}$. 
For the \emph{symmetric} case, it is assumed that 
\[ \alpha_{11}= \alpha_{22}=1, \quad  \alpha_{21}= \alpha_{12}=\alpha,  \quad   \alpha \geq 0.\]

For this interference channel, transmitter~$k$ wishes to send to  receiver~$k$ a confidential message $w_k$ that is  uniformly chosen from a set $\mathcal{W}_k \defeq \{1,2,\cdots, 2^{\bln R_k}\}$, where $R_k$ is  the rate  (bits/channel use) of this message and $\bln$ is the total number of channel uses. 
At each transmitter, a corresponding stochastic function is employed to map  the message  to a transmitted codeword (cf.~\cite{TY:08cj,LMSY:08, XU:14, XU:15, YTL:08} and references therein). Specifically, at transmitter~$k$, the following function \[f_k: \mathcal{W}_k  \times   \mathcal{W}'_{k} \times   \mathcal{W}''_{k}  \to   \mathcal{R} ^{\bln}, \quad  \quad k=1,2\] maps the message $w_k \in \mathcal{W}_k$  to a transmitted codeword  $ x_k^{\bln}  = f_k(w_k, w'_k , w''_k)   \in  \mathcal{R} ^{\bln}$, where $w'_k \in  \mathcal{W}'_k$ and $w''_k \in  \mathcal{W}''_k$ represent the randomness in this mapping, and $\{w'_k, w''_k \}$ are available at transmitter~$k$ only (cf.~\cite{YTL:08}). 
The random variables $\{w_1, w'_1 , w''_1, w_2, w'_2 , w''_2\}$ are assumed to be mutually independent.  In our setting,  the randomness coming from $w'_k$ at transmitter~$k$ is used to guarantee the security (see \eqref{eq:defsecrecy1} and \eqref{eq:defsecrecy2} below) of its own confidential message $w_k$; while the randomness coming from $w''_k$ at transmitter~$k$ is used to potentially improve the secure rate of the \emph{other} transmitter's confidential message. Thus, $w''_k$ can be treated as the \emph{helping} randomness and the signal mapped from $w''_k$ can be treated as a \emph{cooperative jamming signal}.
Hereafter, when any of $w''_1$ and $w''_2$ are used in a communication scheme, we will call it as a scheme with  cooperative jamming.       
When $w''_1$ and $w''_2$ are not used in a communication scheme, i.e., $\mathcal{W}''_1=\mathcal{W}''_2= \phi$, we will call it as a scheme \emph{without} cooperative jamming.       
A secure rate pair $(R_1, R_2)$ is said to be achievable  if for any $\epsilon >0$ there exists a sequence of $\bln$-length codes such that each receiver can decode its own message reliably, i.e., 
 \begin{align}
 \text{Pr}[w_k  \neq \hat{w}_k  ]  \leq \epsilon, \quad \forall k   \label{eq:Pedef}
  \end{align}
and the messages are kept secret such that  
 \begin{align}
 \Imu(w_1; y_{2}^{\bln})  &\leq  \bln \epsilon   \label{eq:defsecrecy1}  \\
\Imu(w_2; y_{1}^{\bln})  &\leq  \bln \epsilon   \label{eq:defsecrecy2}
 \end{align}
where $y_{k}^{\bln}$ represents the $\bln$-length channel output  of receiver~$k$, $k=1,2$. 
The secure capacity region $C$ is the closure of the set of all achievable secure rate pairs.
The secure sum capacity is defined as: 
 \begin{align}
 C_{\text{sum}} \defeq \sup \big\{ R_1 + R_2 |  \  (R_1, R_2) \in C  \big\} .  \label{eq:defGDoFsum}
  \end{align}
The secure sum GDoF is defined as   
 \begin{align}
 d_{\text{sum}}  \defeq   \lim_{P \to \infty}   \frac{C_{\text{sum}}}{ \frac{1}{2} \log P}.  \label{eq:defGDoF}
 \end{align}

\subsection{A scheme without cooperative jamming   \label{sec:noCJGau} }

This subsection  discusses a scheme without cooperative jamming, for the two-user Gaussian interference channel defined in Section~\ref{sec:systemGaussian}.
In the proposed scheme without cooperative jamming, each transmitter simply employs a \emph{Gaussian wiretap codebook} (GWC) to guarantee the secrecy \emph{without using cooperative jamming}, while each receiver simply \emph{treats interference as noise} (TIN) when decoding its desired message. It is called as a  GWC-TIN scheme hereafter.  Note that the secrecy offered by the GWC-TIN scheme is information-theoretic secrecy, which holds for any decoding methods at any eavesdropper. 
Some details of the scheme are discussed as follows.

\subsubsection{Gaussian wiretap codebook}
To build the  codebook, transmitter~$k$  generates a total of $2^{\bln (R_k + R_k')}$  independent  codewords  $v^{\bln}_k$  with each element independent and identically distributed (i.i.d.) according to a  Gaussian distribution with zero mean and variance $P^{- \beta_k}$, $k=1,2$, for some $R_k, R_k'$ and $\beta_k \geq 0$ that will be designed specifically later on. 
The codebook $\mathcal{B}_{k}$ is defined as a set of the labeled codewords:
   \begin{align}
     \mathcal{B}_{k} \defeq \bigl\{ v^{\bln}_k (w_k,  w_k'): \  w_k \in \{1,2,\cdots, 2^{\bln R_k}\}, \  w_k' \in \{1,2,\cdots, 2^{\bln R_k'}\}   \bigr\},  \quad k=1,2.      \label{eq:code2341}
     \end{align}
 To transmit the message $w_k$,  transmitter~$k$ at first selects a bin (sub-codebook)  $\mathcal{B}_{k}( w_k) $ that is defined as 
\[   \mathcal{B}_{k} (w_k)  \defeq \bigl\{ v^{\bln}_k (w_k,  w_k'): \  w_k' \in \{1,2,\cdots, 2^{\bln R_k'}\}   \bigr\},  \quad k=1,2  \]
and then \emph{randomly} chooses a codeword $v^{\bln}_k$ from the selected bin according to a uniform distribution.
Since this scheme will not use cooperative jamming, the chosen codeword $v^{\bln}_k$ will be mapped exactly as a channel input sequence by transmitter~$k$, that is,  $x_k (t) =  v_k (t), \ t=1,2, \cdots, \bln$,
where $v_k (t)$ is the   $t$th element of  the codeword $v^{\bln}_k$,  $k=1,2$. Based on this one-to-one mapping and Gaussian codebook, it implies that  
\begin{align}
x_k (t)= v_k (t) \sim \Nc (0, P^{- \beta_k}), \quad \forall t,  \quad   k=1,2.  \label{eq:map888}
\end{align}
Then,  the received signals take the following forms (removing the time index for simplicity):
\begin{align}
y_{1} &=  \underbrace{\sqrt{P^{\alpha_{11}}} h_{11} v_1}_{| h_{11}|^2 P^{\alpha_{11}  - \beta_1}}+  \underbrace{\sqrt{P^{\alpha_{12}}} h_{12} v_2}_{| h_{12}|^2 P^{\alpha_{12} - \beta_2}}  +  \underbrace{z_{1}}_{P^{0}},  \label{eq:y441}  \\
   y_{2} &= \underbrace{\sqrt{P^{\alpha_{22}}} h_{22} v_{2} }_{| h_{22}|^2P^{\alpha_{22}-\beta_2}} + \underbrace{ \sqrt{P^{\alpha_{21}}} h_{21} v_{1}}_{| h_{21}|^2 P^{\alpha_{21}-\beta_1}}    + \underbrace{z_{2}}_{P^{0}} \label{eq:y442}
\end{align}
(cf.~\eqref{eq:channelG}). In the above equations, the  average power is noted under each  summand term.

\subsubsection{Treating interference as noise}
In terms of decoding, each intended receiver simply treats interference as noise. This implies that receiver~$k$ can  decode the codeword $v^{\bln}_k (w_k,  w_k')$  with arbitrarily small error probability when  $\bln$ gets large and  the rate of the codeword (i.e., $R_k + R_k'$) satisfies the following condition: 
\begin{align}
R_k + R_k' <  \Imu(v_k; y_k ) ,  \quad k=1,2      \label{eq:Rk876}
\end{align}
(cf.~\cite{CT:06}). 
Note that $R_k$ and $R_k'$ represent the rates of the secure message $w_k$ and the confusion message $w_k'$, respectively (cf.~\eqref{eq:code2341}).
Once  the codeword  $v^{\bln}_k (w_k,  w_k')$ is decoded, the message $w_k$ can be decoded directly from the codebook mapping.
At this point, let us set
\begin{align}
R_k &\defeq    \Imu(v_k; y_k) -  \Imu ( v_k; y_{\ell} | v_{\ell} ) - \epsilon,  \label{eq:Rk623} \\
R_k'   &\defeq   \Imu ( v_k; y_{\ell} | v_{\ell}) - \epsilon  \label{eq:Rk623b}  
\end{align}
for some $\epsilon >0$ and $k,\ell   \in \{1,2\}, k \neq \ell$. Obviously,  $R_k$ and $R_k'$ designed in \eqref{eq:Rk623} and \eqref{eq:Rk623b} satisfy the condition in \eqref{eq:Rk876}.

\subsubsection{Secure rate}
 From the  proof of \cite[Theorem~2]{XU:15}   (or \cite[Theorem~2]{LMSY:08})  it implies that, given the above wiretap codebook and the rates designed in \eqref{eq:Rk623} and \eqref{eq:Rk623b},  the messages $w_1$ and $w_2$ are secure from their eavesdroppers, that is,  
$\Imu(w_1; y_{2}^{\bln})  \leq  \bln \epsilon$ and $\Imu(w_2; y_{1}^{\bln})  \leq  \bln \epsilon$. 
Therefore, by letting $\epsilon \to 0$, the scheme 
achieves the secure rate pair  $R_1  =  \Imu(v_1; y_1) -  \Imu ( v_1; y_2 | v_2 ) $ and $R_2  = \Imu(v_2; y_2) -  \Imu ( v_2; y_1 | v_1 )$.   Due to the Gaussian inputs  and  outputs, this achievable secure rate pair is expressed as  
 \begin{align}
R_1  &=\underbrace{\frac{1}{2} \log \bigl(  1+    \frac{|h_{11}|^2 P^{\alpha_{11} - \beta_1}}{ 1+ |h_{12}|^2 P^{\alpha_{1 2} - \beta_{2}}}  \bigr)}_{= \Imu(v_1; y_1)}   -  \underbrace{\frac{1}{2} \log (1+ |h_{21}|^2P^{\alpha_{21} - \beta_{1}})}_{  = \Imu ( v_1; y_2 | v_2 ) }   \non\\
R_2    &=  \underbrace{\frac{1}{2} \log \bigl(  1+     \frac{|h_{22}|^2 P^{\alpha_{22} - \beta_2}}{ 1+ |h_{21}|^2P^{\alpha_{21} - \beta_{1}}}  \bigr) }_{ = \Imu(v_2; y_2)  } -   \underbrace{\frac{1}{2} \log (1+ |h_{12}|^2 P^{\alpha_{12} - \beta_{2}}) }_{ = \Imu ( v_2; y_1 | v_1 ) } \non
\end{align}
for some $\beta_1, \beta_2 \geq 0$.  By setting $\beta_1 =  \alpha_{21}$ and  $\beta_2 =  \alpha_{12}$, then the interference at each receiver is scaled down to the noise level (see~\eqref{eq:y441} and \eqref{eq:y442}) and the  achievable secure rate pair becomes 
 \begin{align}
R_1 &=  \frac{1}{2} \log \bigl(  1+    \frac{|h_{11}|^2 P^{\alpha_{11} - \alpha_{21} }}{1+ |h_{12}|^2 }  \bigr)  -  \frac{1}{2} \log (1+ |h_{21}|^2)     \label{eq:Rk111}    \\
R_2  &=  \frac{1}{2} \log \bigl(  1+     \frac{|h_{22}|^2 P^{\alpha_{22} -  \alpha_{12} }}{ 1+ |h_{21}|^2 }\bigr)  -   \frac{1}{2} \log (1+ |h_{12}|^2).   \label{eq:Rk432}  
\end{align}
Note that the above secure rate pair is achieved by using a simple choice of $\beta_1 $ and  $\beta_2$. One can improve the secure rate pair by selecting the optimized parameters  of $\beta_1 $ and  $\beta_2$.
From the achievable secure rate pair expressed in \eqref{eq:Rk111} and \eqref{eq:Rk432}, it implies that the GWC-TIN scheme achieves the following  secure sum GDoF
 \begin{align}
 d_{\text{sum}}^{GT}  &= (\alpha_{11} - \alpha_{21})^+ + (\alpha_{22} -  \alpha_{12})^+   .  \label{eq:GWC-TINGDoF}    
\end{align}
Note that in this GWC-TIN scheme, the transmitters do \emph{not} need to know  the realizations of $\{h_{k\ell} \}_{k, \ell}$.

\section{Main results  \label{sec:mainresult}}

In this section we provide the main results of this work.   At first we provide the secure sum GDoF characterization for the  two-user symmetric  Gaussian interference channel defined in Section~\ref{sec:systemGaussian}.

\begin{theorem}  \label{thm:GDoF}
Considering the two-user symmetric  Gaussian interference channel defined in Section~\ref{sec:systemGaussian}, for almost all channel coefficients  $\{h_{k\ell}\} \in (1, 2]^{2\times 2}$,  the optimal secure sum GDoF  is characterized as 
\begin{subnumcases}
{ d_{\text{sum}}  =} 
     2(1- \alpha)    &    for   \ $ 0 \leq \alpha \leq  \frac{2}{3}$        \label{thm:capacitydet1} \\
    2(2\alpha- 1)  &  for \ $\frac{2}{3}  \leq \alpha \leq  \frac{3}{4}$    \label{thm:capacitydet2} \\ 
        2(1 -  2\alpha / 3)  &  for  \   $\frac{3}{4}  \leq \alpha \leq  1$    \label{thm:capacitydet3} \\ 
            2\alpha/3  &  for  \  $1  \leq  \alpha \leq  \frac{3}{2}$    \label{thm:capacitydet5} \\ 
                        2(2- \alpha)  &  for \   $\frac{3}{2} \leq  \alpha \leq 2$   \label{thm:capacitydet6} \\ 
                                               0  &  for  \   $2\leq  \alpha$ .  \label{thm:capacitydet7}
\end{subnumcases}
Moreover,  a simple scheme without using cooperative jamming, that is, GWC-TIN  scheme, achieves the optimal GDoF if and only if  $\alpha \in [0,  \frac{2}{3}]$. 
\end{theorem}
\begin{proof}
The converse follows from Lemma~\ref{lm:gupper} and Corollary~\ref{cor:symSGDoF} provided in Section~\ref{sec:converse}.
When $\alpha \in [0,  \frac{2}{3}]$, the optimal secure sum GDoF is achievable by the  proposed GWC-TIN  scheme (see \eqref{eq:GWC-TINGDoF}). When $\alpha \in (\frac{2}{3}, 2)$, the optimal secure sum GDoF is achievable by the  proposed   scheme with cooperative jamming.   Section~\ref{sec:CJGau}  provides  the cooperative jamming scheme for achieving the optimal secure sum GDoF when $\alpha \in (\frac{2}{3}, 2)$. 
\end{proof}

Fig.~\ref{fig:Scapacity}  depicts the sum GDoF  with  secrecy constraint (cf.~Theorem~\ref{thm:GDoF}), as well as the sum GDoF without  secrecy constraint (cf.~\cite{ETW:08}), for the two-user \emph{symmetric} Gaussian interference channel.  
Note that the secrecy constraint  incurs no  penalty in sum GDoF \emph{if and only if} $\alpha \in [0,  \frac{1}{2}]$. 
In the following we focus on the general two-user  Gaussian interference channel defined in Section~\ref{sec:systemGaussian}, and provide the optimality conditions  in which the GWC-TIN scheme is optimal in terms of  secure  sum capacity to within a constant gap.

\begin{figure}[t!]
\centering
\includegraphics[width=8.9cm]{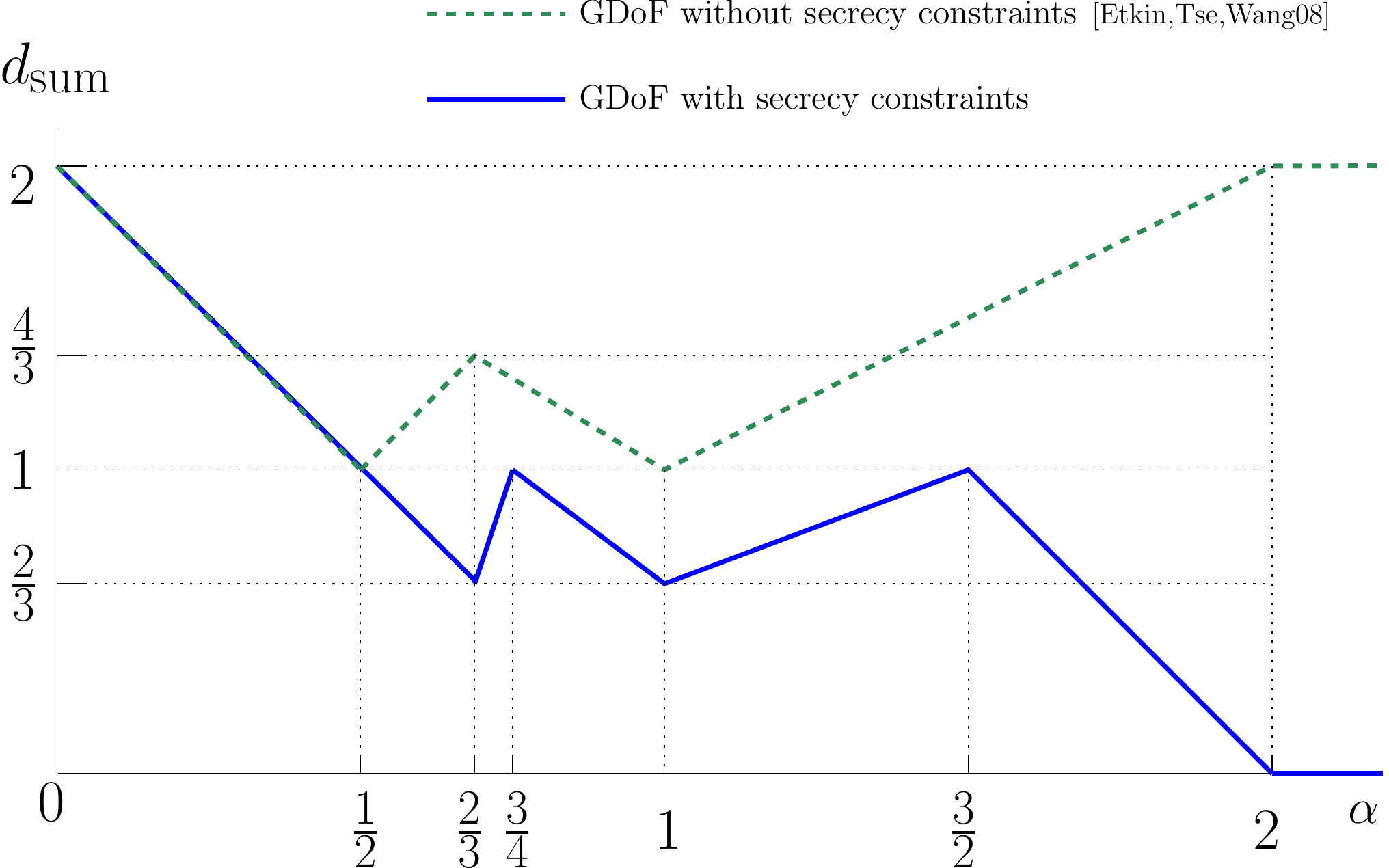}
\caption{GDoF vs. $\alpha$ for the two-user \emph{symmetric} Gaussian interference channel with and without secrecy constraints.  Note that a simple scheme  without using cooperative jamming achieves the optimal secure sum GDoF if and only if  $\alpha \in [0,  \frac{2}{3}]$.}
\label{fig:Scapacity}
\end{figure}

\begin{theorem}  \label{thm:GaussianNCJ}
For the two-user  Gaussian interference channel defined in Section~\ref{sec:systemGaussian}, where $\alpha_{k\ell} $ denotes the channel strength from transmitter~$\ell$ to receiver~$k$, $\forall k, \ell \in \{1,2\}$,  if the following conditions are satisfied,  
\begin{align}
    \alpha_{22} + (\alpha_{11 } - \alpha_{12})^+    & \geq   \alpha_{21}+  \alpha_{12},   \label{eq:capGaussian1}    \\
    \alpha_{11} + (\alpha_{22 } - \alpha_{21})^+     & \geq   \alpha_{21}+  \alpha_{12}       \label{eq:capGaussian2} 
\end{align}
then the simple scheme of using Gaussian wiretap codebook at each transmitter (without using cooperative jamming) and treating interference as noise at each intended receiver (that is, GWC-TIN scheme) achieves the optimal secure sum capacity $C_{\text{sum}}$ to within a constant gap of no larger than $11$ bits.   
More specifically, given the conditions of \eqref{eq:capGaussian1} and \eqref{eq:capGaussian2}, the optimal secure sum capacity $C_{\text{sum}}$ satisfies  
 \begin{align}
C^{lb}_{\text{sum}}   \leq  C_{\text{sum}}    \leq   C^{lb}_{\text{sum}}  +  11  \non 
  \end{align}
  where the lower bound $C^{lb}_{\text{sum}}$ is defined by
 \begin{align*}
  &C^{lb}_{\text{sum}}  \defeq \frac{1}{2} \log \bigl(  1+    \frac{|h_{11}|^2 P^{\alpha_{11} - \alpha_{21} }}{1+ |h_{12}|^2 }  \bigr)  -  \frac{1}{2} \log (1+ |h_{21}|^2)  +  \frac{1}{2} \log \bigl(  1+     \frac{|h_{22}|^2 P^{\alpha_{22} -  \alpha_{12} }}{ 1+ |h_{21}|^2 }\bigr)  -   \frac{1}{2} \log (1+ |h_{12}|^2) .
  \end{align*}  
\end{theorem}

\begin{proof}
As discussed in Section~\ref{sec:noCJGau}, the GWC-TIN scheme achieves the  secure sum capacity lower bound $C^{lb}_{\text{sum}}$  (see \eqref{eq:Rk111} and \eqref{eq:Rk432}).
To prove the optimality of GWC-TIN scheme, we provide a secure sum capacity upper bound  in Lemma~\ref{lm:gupper} (see \eqref{eq:gaussianupbound} in Section~\ref{sec:converse}). 
 The derived  upper bound  reveals that, if the conditions in \eqref{eq:capGaussian1} and \eqref{eq:capGaussian2} are satisfied,  then the achievable secure sum rate of GWC-TIN scheme  indeed approaches the secure sum capacity  to within a constant gap of no larger than $11$ bits (see Appendix~\ref{sec:gap} for the details on bounding the gap).
\end{proof}

\section{Cooperative Jamming Scheme \label{sec:CJGau} }

This section  provides  the cooperative jamming scheme for the two-user \emph{symmetric} Gaussian  interference channel defined in Section~\ref{sec:systemGaussian}, focusing on the regime of  $\frac{2}{3} < \alpha \leq 2$. Note that when $0 \leq  \alpha \leq \frac{2}{3}$, a scheme without cooperative jamming has been proposed  in Section~\ref{sec:noCJGau} to achieve the optimal secure sum capacity to within a constant gap (see Theorem~\ref{thm:GaussianNCJ}).  In this cooperative jamming scheme,  pulse amplitude modulation (PAM) and interference alignment will be used.

At first let us review the  PAM modulation, which will be used in our cooperative jamming scheme described later. 
Assume that a random variable $x$ is uniformly drawn from 
a PAM constellation set, defined as
\begin{align}
   \Omega (\xi,  Q)  \defeq   \{ \xi \cdot a :   \    a \in  \Zc  \cap [-Q,   Q]   \}      \label{eq:constellation}  
 \end{align}
for some $Q \in \Zc^+ $ and $\xi \in \Rc$.  The role of $\xi$ is to regularize  the  average power of $x$.  
Given that $x$ is uniformly drawn from $ \Omega (\xi,  Q)$, the average power of $x$ is 
\begin{align}
 \E |x|^2  =  \frac{2  \xi^2 }{ 2Q +1}  \sum_{i=1}^{Q} i^2  =  \frac{  \xi^2  Q(Q+1)}{3} \label{eq:xpower}  
 \end{align} 
which implies that  
\begin{align}
 \E |x|^2  \leq 1/3 <1,   \quad  \text{if} \quad  \xi \leq  \frac{  1}{ \sqrt{2}Q} .    \label{eq:xpowerbeta}  
\end{align} 
Let us consider the communication of $x$  over  a  channel model given as
  \begin{align}
   y=  \sqrt{P^{\alpha_1}}  h x + \sqrt{P^{\alpha_2}} g + z .     \label{eq:aqmy1}  
 \end{align}
 In the above channel model,   $x \in \Omega (\xi,  Q)$  is the input with a normalized power constraint  $\E |x|^2 \leq 1$. $z\sim \mathcal{N}(0, \sigma^2)$ is additive white Gaussian noise. $g  \in  \Sc_{g}$ is a discrete random variable such that \[|g | \leq  g_{\max}, \quad \forall g \in  \Sc_{g}\] for a given set $\Sc_{g} \subset \Rc$,  where $g_{\max}$ is a positive and finite constant  independent of $P$.    In this setting $\alpha_1$, $\alpha_2$, $\sigma$ and $h$ are  four given positive and finite parameters independent of $P$.  
Note that,  the  minimum distance of the constellation  for $\sqrt{P^{\alpha_1}}  h x$ is  
 \begin{align}
  d_{\text{min}} (\sqrt{P^{\alpha_1}} hx) =\sqrt{P^{\alpha_1}}h \cdot \xi.  \label{eq:distance1}  
 \end{align}
 For this setting, we have the following result regarding the probability of error for decoding $x$   from $y$.

 \begin{lemma}  \label{lm:AWGNic}
Consider the channel model in \eqref{eq:aqmy1} and consider the case of $\alpha_1 - \alpha_2 >0$. 
By setting  $Q$ and $\xi$ such that 
  \begin{align}
   Q =  \frac{P^{\frac{\bar{\alpha}}{2}} \cdot h \gamma }{2 g_{\max} },    \quad  \quad    \xi =  \gamma \cdot \frac{ 1}{Q},   \quad   \quad   \forall \bar{\alpha} \in (0, \alpha_1 -\alpha_2)  \label{eq:settingAWGNIC}                           
 \end{align}
where $\gamma > 0$ is a finite constant independent of $P$, then  
the probability of error for decoding a symbol $x$ from $y$ is \[ \text{Pr} (e) \to 0  \quad  \text{as} \quad   P\to \infty.\] 
\end{lemma}
\begin{proof}
See Appendix~\ref{sec:AWGNic}. 
\end{proof}

Now let us go back to the  two-user  \emph{symmetric}  Gaussian interference channel defined in Section~\ref{sec:systemGaussian}. For this setting, we will provide a cooperative jamming scheme by using PAM modulation and interference alignment.  For the proposed scheme, the details of codebook generation, signal mapping, secure rate analysis,  PAM constellation, and interference alignment are provided as follows.

\subsubsection{Codebook generation and signal mapping}
To build the  codebook, transmitter~$k$, $k=1,2$,  generates a total of $2^{\bln (R_k + R_k')}$  independent  sequences (codewords)  $v^{\bln}_k $, where all the elements of all the sequences are  independent and identically distributed  according to a distribution  that will be designed specifically in the scheme,
and  $R_k, R_k'$ will be  defined  later on.
The codebook $\mathcal{B}_{k}$ is defined as a set of the labeled codewords:
   \begin{align}
     \mathcal{B}_{k} \defeq \Bigl\{ & v^{\bln}_k (w_k,  w_k'):  \  w_k \in \{1,2,\cdots, 2^{\bln R_k}\},   \  w_k' \in \{1,2,\cdots, 2^{\bln R_k'}\}   \Bigr\},  \quad k=1,2.      \label{eq:code2341J}
     \end{align}
 To transmit the message $w_k$,  transmitter~$k$ at first selects a bin (sub-codebook)  $\mathcal{B}_{k}( w_k) $ that is defined as 
\[   \mathcal{B}_{k} (w_k)  \defeq \bigl\{ v^{\bln}_k (w_k,  w_k'): \  w_k' \in \{1,2,\cdots, 2^{\bln R_k'}\}   \bigr\},  \quad k=1,2  \]
and then \emph{randomly} chooses a codeword $v^{\bln}_k$ from the selected bin according to a uniform distribution.
Since this scheme uses \emph{cooperative jamming}, the chosen codeword $v^{\bln}_k$ will be mapped to the channel input using a stochastic function given as
 \begin{align}
  x_k (t) = h_{\ell\ell} \cdot v_k (t)  +    \sqrt{P^{ - \beta_{u_k}}}  \cdot   h_{k\ell}  \cdot   u_{k}(t)   \label{eq:xvkkk}      
   \end{align}
  for $\ell \neq k, \ k, \ell =1,2$, and  $t=1,2, \cdots, \bln$, where $\beta_{u_k}$ is a parameter that will be defined later, and
 $\{u_{k}(t)\}_{t}$ are i.i.d. random variables (jamming signals) \emph{uniformly} and \emph{independently} drawn from a constellation set  that will be designed specifically in the scheme.

\subsubsection{PAM constellation and interference alignment}
Specifically,   the codebook at transmitter $k$ is generated such that  each element takes the following form
 \begin{align}
  v_{k}  =     \sqrt{P^{ - \beta_{v_{k,c}}}}  \cdot  v_{k,c} +   \sqrt{P^{ - \beta_{v_{k,p}}}}    \cdot  v_{k,p}    \label{eq:xvk}  
 \end{align}
 and then the channel input in~\eqref{eq:xvkkk} becomes 
 \begin{align}
  x_k  =   \sqrt{P^{ - \beta_{v_{k,c}}}}  \cdot h_{\ell\ell} \cdot v_{k,c} +   \sqrt{P^{ - \beta_{v_{k,p}}}}    \cdot h_{\ell\ell} \cdot v_{k,p} +    \cdot \sqrt{P^{ - \beta_{u_k}}}  \cdot  h_{k\ell}  \cdot u_{k}  \label{eq:xvkkk1}  
 \end{align}
(removing the time index for simplicity), for $\ell \neq k$,  $k, \ell =1,2$,   
where $v_{k,c}$, $v_{k,p}$ and $u_{k}$ are \emph{independent} random variables   \emph{uniformly} drawn from their PAM constellation sets, 
 \begin{align}
   v_{k,c}      &  \in    \Omega ( \xi =  \gamma_{v_{k,c}} \cdot \frac{ 1}{Q} ,   \   Q =  P^{ \frac{ \lambda_{v_{k,c}} }{2}} )  \label{eq:constellationGsym1}   \\ 
   v_{k,p}      &  \in    \Omega ( \xi =  \gamma_{v_{k,p}} \cdot \frac{ 1}{Q} ,   \   Q = P^{ \frac{  \lambda_{v_{k,p}} }{2}} )    \label{eq:constellationGsym2}    \\
      u_{k}      &  \in    \Omega ( \xi =  \gamma_{u_{k}} \cdot \frac{ 1}{Q} ,      \quad  Q =  P^{ \frac{ \lambda_{u_{k}} }{2}} )  \label{eq:constellationGsym1u} 
 \end{align}
 respectively, for $k=1,2$, where  $\Omega (\xi,    Q )$ is defined in  \eqref{eq:constellation};  and  $\{\gamma_{v_{k,c}}, \gamma_{v_{k,p}}, \gamma_{u_k}\}$ are some finite constants independent of $P$, such that
  \begin{align}
  \gamma_{v_{1,c}} =  \gamma_{v_{2,c}} =\gamma_{u_1} = \gamma_{u_2} =  2\gamma_{v_{1,p}}=  2\gamma_{v_{2,p}} = 2\gamma  \in \bigl(0, \frac{1}{2\sqrt{2}}\bigr].  \label{eq:gammadef} 
 \end{align}
In our scheme,  $\{\beta_{v_{k,c}}, \beta_{v_{k,p}}, \beta_{u_k}, \lambda_{v_{k,c}},  \lambda_{v_{p,c}}, \lambda_{u_{k}}\}_{k=1,2}$ are given parameters designed  in Table~\ref{tab:para} for different cases of $\alpha$.  Note that for this \emph{symmetric} Gaussian interference channel, the parameters are designed symmetrically for the two transmitters, for example, $\beta_{v_{1,c}}= \beta_{v_{2,c}}$ and $\beta_{v_{1,p}}= \beta_{v_{2,p}}$.     
For the parameters designed as $\beta_{v_{1,p}} =\infty$ and $\lambda_{v_{1,p}}=0$, it implies that the random variable $v_{1,p}$ can be treated as an empty term  in the transmitted signal $x_1$. This implication holds for the other parameters.   
Based on our design on $x_k$ (see \eqref{eq:xvkkk1}-\eqref{eq:constellationGsym1u}) and the parameters, one can easily check that the power constraint  $\E |x_k|^2 \leq 1$ is satisfied (cf.~\eqref{eq:xpowerbeta}) for $k=1,2$.

Given the above design on the transmitted signals,  the  signals received at the receivers then take the following forms (without the time index)
\begin{align}
y_{1} &=    \sqrt{P^{ 1 - \beta_{v_{1,c}}}} h_{11}h_{22} v_{1,c}  +    \sqrt{P^{ 1 - \beta_{v_{1,p}}}} h_{11}h_{22} v_{1,p}    +     h_{12} h_{11}(   \sqrt{P^{ \alpha - \beta_{v_{2,c}}}} v_{2,c}  +   \sqrt{P^{ 1 - \beta_{u_1}}} u_1)    \non\\
&\quad +   \sqrt{P^{ \alpha - \beta_{u_{2}}}}  h_{12} h_{21} u_{2}  +    \sqrt{P^{ \alpha - \beta_{v_{2,p}}}}  h_{12} h_{11}  v_{2,p}   +  z_{1}    \label{eq:yvk1}  \\
y_{2} &=     \sqrt{P^{ 1 - \beta_{v_{2,c}}}} h_{22}h_{11} v_{2,c}    +    \sqrt{P^{ 1 - \beta_{v_{2,p}}}} h_{22}h_{11} v_{2,p}    +     h_{21} h_{22}(   \sqrt{P^{ \alpha - \beta_{v_{1,c}}}} v_{1,c}  +   \sqrt{P^{ 1 - \beta_{u_2}}} u_2)       \non\\
 &\quad +   \sqrt{P^{ \alpha - \beta_{u_{1}}}}  h_{21} h_{12} u_{1}  +   \sqrt{P^{ \alpha - \beta_{v_{1,p}}}}  h_{21} h_{22}  v_{1,p}   +  z_{2} .   \label{eq:yvk2}  
\end{align}
 With the signal design in \eqref{eq:xvkkk1}, the signal $v_{2,c}$ is aligned with the jamming signal $u_1$  at receiver~1, while the signal $v_{1,c}$ is aligned with the jamming signal $u_2$ at receiver~2.

\subsubsection{Secure rate analysis} Let us now define   $R_k$ and $R_k'$ as 
\begin{align}
R_k &\defeq   \Imu(v_k; y_k) -  \Imu ( v_k; y_{\ell} | v_{\ell} ) - \epsilon   \label{eq:Rk623J} \\  
R_k'  &\defeq  \Imu ( v_k; y_{\ell} | v_{\ell}) - \epsilon  \label{eq:Rk623bJ}  
\end{align}
for some $\epsilon >0$, $\ell \neq k, \  k,\ell   \in \{1,2\}$. From the  proof of \cite[Theorem~2]{XU:15}   (or \cite[Theorem~2]{LMSY:08})  it implies that, given the above codebook and signal mapping, the  rate pair $(R_1, R_2)$ defined in \eqref{eq:Rk623J} and \eqref{eq:Rk623bJ} is achievable  and  the messages $w_1$ and $w_2$ are secure from their eavesdroppers, that is,  
$\Imu(w_1; y_{2}^{\bln})  \leq  \bln \epsilon$ and $\Imu(w_2; y_{1}^{\bln})  \leq  \bln \epsilon$. 
In what follows we will analyze the secure rate performance of the proposed scheme, focusing on the regime of  $2/3 < \alpha \leq 2$.

\begin{table}
\caption{Parameter design for the symmetric  channel, for some $\epsilon >0$.}
\begin{center}
{\renewcommand{\arraystretch}{1.7}
\begin{tabular}{|c|c|c|c|c|}
  \hline
                     &   $\frac{2}{3} < \alpha \leq \frac{3}{4}$  &  $\frac{3}{4} \leq  \alpha \leq 1$  & $1 \leq \alpha \leq \frac{3}{2}$   & $\frac{3}{2} \leq \alpha \leq 2$  \\
   \hline
   $\beta_{v_{1,c}}, \ \beta_{v_{2,c}}$    		&   $0$     			&   $ 0$    &    $\alpha -1 $    &   $\alpha -1 $  \\
    \hline
   $\beta_{u_1}, \  \beta_{u_2}$     		     		&   $1- \alpha$      	&    $1- \alpha$   &   0    &    0  \\
    \hline
   $\beta_{v_{1,p}}, \ \beta_{v_{2,p}}$ 			&   $\alpha$    		&    $\alpha$    &   $\infty$    &    $\infty$  \\
    \hline
   $\lambda_{v_{1,c}}, \ \lambda_{v_{2,c}}$ 		&   $3\alpha -2 - \epsilon$ 	&  $\alpha/3 - \epsilon$     &   $\alpha/3 - \epsilon$     &  $2- \alpha - \epsilon$   \\
    \hline
   $\lambda_{u_1}, \  \lambda_{u_2}$ 					&    $3\alpha -2 - \epsilon$ &   $\alpha/3 - \epsilon$    &   $\alpha/3 - \epsilon$     &   $2- \alpha - \epsilon$  \\
  \hline
   $\lambda_{v_{1,p}}, \ \lambda_{v_{2,p}}$  &   $1 - \alpha - \epsilon$  	&  $1- \alpha - \epsilon$     &     $0 $  &    0  \\
    \hline
    \end{tabular}
}
\end{center}
\label{tab:para}
\end{table}

 \subsection{Rate analysis when $2/3 <  \alpha \leq  3/4$   \label{sec:CJscheme2334}}

 For the case with $2/3 <  \alpha \leq  3/4$, we design the parameters  such that
 \begin{align}
\beta_{v_{1,c}}&= \beta_{v_{2,c}} =0,        \quad\quad \quad    \quad\quad       \beta_{u_1}=  \beta_{u_2}= 1- \alpha,   \quad  \quad  \quad \quad\quad \beta_{v_{1,p}}=\beta_{v_{2,p}}= \alpha  \label{eq:para111} \\
 \lambda_{v_{1,c}} &= \lambda_{v_{2,c}}= 3\alpha -2 - \epsilon,\quad     \lambda_{u_1}=  \lambda_{u_2}= 3\alpha -2 - \epsilon,   \quad \quad  \quad  \lambda_{v_{1,p}}= \lambda_{v_{2,p}}  = 1 - \alpha - \epsilon  \label{eq:para222}
 \end{align}
 where  $\epsilon>0$ can be set arbitrarily small.
 In this case, the transmitted signal at transmitter~$k$ takes the following form
 \begin{align}
  x_k  =    h_{\ell\ell} \cdot v_{k,c} +   \sqrt{P^{ -\alpha}}    \cdot h_{\ell\ell} \cdot v_{k,p} +    \cdot \sqrt{P^{ -(1- \alpha)}}  \cdot  h_{k\ell}  \cdot u_{k}    \label{eq:xvkkk12334}
 \end{align}
  for $\ell \neq k$,  $k, \ell =1,2$.   
Then the received signals at the receivers take the following forms 
\begin{align}
y_{1} &=    \sqrt{P} h_{11}h_{22} v_{1,c}  +    \sqrt{P^{ 1 - \alpha}} h_{11}h_{22} v_{1,p}    +    \sqrt{P^{ \alpha}}  h_{12} h_{11}(   v_{2,c}  +  u_1)    +   \sqrt{P^{ 2\alpha - 1 }}  h_{12} h_{21} u_{2}  +    h_{12} h_{11}  v_{2,p}   +  z_{1}    \label{eq:yvk12334}  \\
y_{2} &=     \sqrt{P} h_{22}h_{11} v_{2,c}    +    \sqrt{P^{ 1 - \alpha}} h_{22}h_{11} v_{2,p}    +     \sqrt{P^{ \alpha}} h_{21} h_{22}(  v_{1,c}  +    u_2)       +  \sqrt{P^{ 2\alpha - 1 }}   h_{21} h_{12} u_{1}  +    h_{21} h_{22}  v_{1,p}   +  z_{2} .   \label{eq:yvk22334}  
\end{align}

For the proposed jamming scheme, the secure rate pair $(R_1, R_2)$ defined in \eqref{eq:Rk623J} and \eqref{eq:Rk623bJ} is achievable. By setting $\epsilon \to 0$, this secure rate pair is expressed as    
\begin{align}
R_1  & =    \Imu(v_1; y_1) - \Imu(v_1; y_2 | v_2 )      \label{eq:lboundit1} \\
R_2  & =    \Imu(v_2; y_2) - \Imu(v_2; y_1 | v_1 ) .    \label{eq:lboundit2}
\end{align}

Let us first focus on the lower bound of $R_1$ expressed  in \eqref{eq:lboundit1}. We will  begin with the lower bound of $ \Imu(v_1; y_1)$.
For this case, $v_1$ takes the form \[v_1 = v_{1,c} +   \sqrt{P^{ -\alpha}}   \cdot v_{1,p}\] (see \eqref{eq:xvk} and \eqref{eq:para111}).   
From $y_1$ expressed in \eqref{eq:yvk12334},  $\{v_{1,c}, v_{1,p}\}$  can be  estimated by using a successive decoding method, given the  design in \eqref{eq:constellationGsym1}-\eqref{eq:constellationGsym1u} and \eqref{eq:para111}-\eqref{eq:xvkkk12334}. 
The following lemma provides a result on the error probability for this estimation.  Later on we will use it to  derive the lower bound on $ \Imu(v_1; y_1)$.

 \begin{lemma}  \label{lm:rateerror2334}
When $2/3 < \alpha \leq 3/4$, given the signal design in \eqref{eq:constellationGsym1}-\eqref{eq:constellationGsym1u} and \eqref{eq:para111}-\eqref{eq:xvkkk12334},  the error probability of estimating  $\{v_{k,c}, v_{k,p} \}$ from $y_k$  is
 \begin{align}
 \text{Pr} [  \{ v_{k,c} \neq \hat{v}_{k,c} \} \cup  \{ v_{k,p} \neq \hat{v}_{k,p} \}  ]  \to 0         \quad \text {as}\quad  P\to \infty    \label{eq:error1c1p}
 \end{align}
for $k=1,2$, where  $ \hat{v}_{k,c}$ and  $\hat{v}_{k,p}  $ are the corresponding estimates for  $v_{k,c}$ and $v_{k,p}$, respectively, based on the observation $y_k$ expressed in \eqref{eq:yvk12334} and \eqref{eq:yvk22334}. 
 \end{lemma}

\begin{proof}
See Appendix~\ref{sec:rateerror2334}.
\end{proof}

Since   $v_{1,c} \in    \Omega (\xi   =  \gamma_{v_{1,c}} \cdot \frac{ 1}{Q},   \   Q =  P^{ \frac{ 3\alpha -2 - \epsilon}{2}} ) $ and $v_{1,p}  \in    \Omega (\xi   =\gamma_{v_{1,p}} \cdot \frac{ 1}{Q},   \   Q = P^{ \frac{ 1 - \alpha - \epsilon}{2}} ) $  are  \emph{uniformly} and \emph{independently}  drawn from their PAM constellation sets, the rates of $v_{1,c} $ and $v_{1,p}$ are given as 
  \begin{align}
  \Hen(v_{1,c}) &=  \log (2 \cdot P^{ \frac{ 3\alpha -2 - \epsilon}{2}} +1)   \label{eq:rateAWGNIC111}     \\
  \Hen (v_{1,p}) &=  \log (2 \cdot P^{ \frac{ 1 - \alpha - \epsilon}{2}} +1)  \label{eq:rateAWGNIC222}                              
 \end{align}
 where $\epsilon$ can be set as $\epsilon \to 0$. 
 Due to the signal design $v_{1} \defeq   v_{1,c} +   \sqrt{P^{ - \alpha}}    \cdot  v_{1,p} $ (see \eqref{eq:xvk}), it is true that  $\{ v_{1,p}, v_{1,c}\}$ can be reconstructed from $v_{1}$,  and vice versa. 
Then, 
  \begin{align}
 \Hen(v_{1}) =   \Hen(v_{1,c}, v_{1,p}) &= \Hen(v_{1,c})+  \Hen(v_{1,p})   \non\\
 &=    \log (2 \cdot P^{ \frac{ 3\alpha -2 - \epsilon}{2}} +1)  +  \log (2 \cdot P^{ \frac{ 1 - \alpha - \epsilon}{2}} +1)   \non\\
 &=   \frac{2 \alpha -1 - 2\epsilon}{2} \log P + o(\log P) . \label{eq:rateAWGNIC333}                                  
 \end{align}
At this point, $\Imu(v_1; y_1)$ can be lower bounded by
  \begin{align}
  \Imu(v_1; y_1)   &\geq   \Imu(v_1; \hat{v}_{1,c}, \hat{v}_{1,p})  \label{eq:rateAWGNIC1202}     \\
  &=   \Hen(v_1) -   \Hen(v_1  |  \hat{v}_{1,c}, \hat{v}_{1,p})    \non    \\
    &\geq    \Hen(v_1) -       \bigl( 1+    \text{Pr} [  \{ v_{1,c} \neq \hat{v}_{1,c} \} \cup  \{ v_{1,p} \neq \hat{v}_{1,p} \}  ] \cdot \Hen(v_{1}) \bigr)  \label{eq:rateAWGNIC1404}     \\
        & =     \bigl( 1 -   \text{Pr} [  \{ v_{1,c} \neq \hat{v}_{1,c} \} \cup  \{ v_{1,p} \neq \hat{v}_{1,p} \}  ] \bigr)   \cdot \Hen(v_{1})  - 1  \label{eq:rateAWGNIC1405}     
 \end{align}
where \eqref{eq:rateAWGNIC1202} results from the Markov chain $v_1 \to y_1 \to  \{\hat{v}_{1,c}, \hat{v}_{1,p} \}$;
\eqref{eq:rateAWGNIC1404} results from Fano's inequality.
By combining  \eqref{eq:rateAWGNIC333},  \eqref{eq:rateAWGNIC1405} and Lemma~\ref{lm:rateerror2334}, we have 
  \begin{align}
  \Imu(v_1; y_1)   &\geq    \frac{2 \alpha -1 - 2\epsilon}{2} \log P + o(\log P) .   \label{eq:rateAWGNIC17762}  
 \end{align}

For the term $\Imu(v_1; y_2 | v_2 )$ in \eqref{eq:lboundit1}, we can bound it as
  \begin{align}
  &\Imu(v_1; y_2 | v_2 )     \non\\
\leq &  \Imu(v_1; y_2, u_1,  v_{1,c} + u_2  | v_2 )     \label{eq:rateAWGNIC18374}   \\
=  &    \Imu(v_1;    v_{1,c} + u_2  | u_1,  v_2 )   + \Imu(v_1;    y_2  | u_1,  v_2, v_{1,c} + u_2 )   + \underbrace{ \Imu(v_1; u_1  | v_2 )}_{=0}    \non  \\
=   & \Imu(v_1;    v_{1,c} + u_2  )   +  \Imu(v_1;     h_{21} h_{22}  v_{1,p}   +  z_{2}   | u_1,  v_2 , v_{1,c} + u_2 )     \label{eq:rateAWGNIC18735}   \\
=   & \Hen( v_{1,c} + u_2  )  - \Hen(u_2  )     +  \hen( h_{21} h_{22}  v_{1,p}   +  z_{2}   | u_1,  v_2 , v_{1,c} + u_2 )   -  \hen(   z_{2} )    \label{eq:rateAWGNIC77364}    \\
\leq &  \underbrace{ \log (4  \cdot P^{ \frac{ 3\alpha -2 - \epsilon}{2}} +1)    -   \log (2  \cdot P^{ \frac{ 3\alpha -2 - \epsilon}{2}} +1)}_{\leq  \log 2 }    +   \underbrace{\hen( h_{21} h_{22}  v_{1,p}   +  z_{2} )}_{\leq  \frac{1}{2} \log ( 2 \pi e ( | h_{21}|^2 \cdot | h_{22}|^2\cdot \E |v_{1,p}|^2     +  \E|  z_{2}|^2 )) }   -  \frac{1}{2}\log (2 \pi e)    \label{eq:rateAWGNIC1955}   \\
\leq &  1   +  \frac{1}{2}\log (34\pi e )       -  \frac{1}{2} \log (2\pi e)    \label{eq:rateAWGNIC8356}   \\
= & \log (2\sqrt{17})    \label{eq:rateAWGNIC9562}  
 \end{align}
where \eqref{eq:rateAWGNIC18374}    uses the fact that adding information does not reduce the mutual information;
 \eqref{eq:rateAWGNIC18735} follows from the fact that $v_1, v_2, u_1, u_2 $ are mutually independent;
\eqref{eq:rateAWGNIC77364}   uses the fact that $\{ v_{1,p}, v_{1,c}\}$ can be reconstructed from $v_{1}$;
\eqref{eq:rateAWGNIC1955} holds true because $ \hen(   z_{2} )  =  \frac{1}{2}\log (2\pi e)$,  $\Hen(u_2  ) =   \log (2  \cdot P^{ \frac{ 3\alpha -2 - \epsilon}{2}} +1) $,  and $\Hen( v_{1,c} + u_2  ) \leq  \log (4  \cdot P^{ \frac{ 3\alpha -2 - \epsilon}{2}} +1)$;  note that $u_2$  is uniformly drawn from $ \Omega (\xi   = \gamma_{u_{2}} \cdot \frac{ 1}{Q},   \   Q =  P^{ \frac{ 3\alpha -2 - \epsilon}{2}} )$,  and  $v_{1,c} + u_2 \in  2 \cdot \Omega (\xi   = \gamma_{u_{2}} \cdot \frac{ 1}{Q},   \   Q =  P^{ \frac{ 3\alpha -2 - \epsilon}{2}} ) $, where $2\cdot \Omega (\xi,  Q)  \defeq   \{ \xi \cdot a :   \    a \in  \Zc  \cap [-2Q,   2Q]   \}$;
\eqref{eq:rateAWGNIC8356} follows from the identity that $\log (4a +1) - \log (2 a +1) \leq  \log 2 $ for any $a \geq 0$ and the fact that Gaussian input maximizes the differential entropy; note that $\E|   h_{21} h_{22}  v_{1,p}   +  z_{2}   |^2 =  | h_{21}|^2 \cdot | h_{22}|^2\cdot \E |v_{1,p}|^2    +  \E|  z_{2}|^2  \leq 16  + 1$  (see \eqref{eq:xpower} and \eqref{eq:xpowerbeta}), where $v_{1,p}$ is uniformly drawn from $\Omega (\xi   = \gamma_{v_{1,p}} \cdot \frac{ 1}{Q},    \   Q = P^{ \frac{ 1 - \alpha - \epsilon}{2}} )$.

Finally,  by incorporating \eqref{eq:rateAWGNIC17762}  and \eqref{eq:rateAWGNIC9562} into \eqref{eq:lboundit1}, the secure rate $R_1$ is lower bounded by 
\begin{align}
R_1  & =    \Imu(v_1; y_1) - \Imu(v_1; y_2 | v_2 )      \non \\
& \geq   \frac{2 \alpha -1 - 2\epsilon}{2} \log P + o(\log P).  \label{eq:lbounditfinal1}
\end{align}
Due to the symmetry, by interchanging the roles of users, it gives the lower bound on the secure rate $R_2$: 
\begin{align}
R_2   & =    \Imu(v_2; y_2) - \Imu(v_2; y_1 | v_1 )  \non\\
& \geq     \frac{2 \alpha -1 - 2\epsilon}{2} \log P + o(\log P).  \label{eq:lbounditfinal2}
\end{align}
 By setting $\epsilon \to 0$, then the secure GDoF pair $(d_1 =  2 \alpha -1,  d_2 =  2 \alpha -1 )$ is  achievable by the proposed cooperative jamming scheme when  $2/3 <  \alpha \leq  3/4$.

 \subsection{Rate analysis when $3/2 \leq  \alpha \leq  2$   \label{sec:CJscheme322}}

 For the case with $3/2 \leq   \alpha \leq 2$, we design the parameters  such that
 \begin{align}
\beta_{v_{1,c}}&= \beta_{v_{2,c}} =\alpha -1,        \quad\quad \quad    \quad\quad       \beta_{u_1}=  \beta_{u_2}= 0,   \quad  \quad  \quad \quad\quad \beta_{v_{1,p}}=\beta_{v_{2,p}}= \infty  \label{eq:para333} \\
 \lambda_{v_{1,c}} &= \lambda_{v_{2,c}}= 2 - \alpha - \epsilon,\quad  \quad    \lambda_{u_1}=  \lambda_{u_2}= 2 - \alpha  - \epsilon,   \quad \quad  \quad  \lambda_{v_{1,p}}= \lambda_{v_{2,p}}  = 0  \label{eq:para444}
 \end{align}
 where  $\epsilon>0$ can be set arbitrarily small.
 In this case,  $\beta_{v_{k,p}}  =\infty$ and $\lambda_{v_{k,p}}=0$, it implies that the random variable $v_{k,p}$ can be treated as an empty term  in the transmitted signal $x_k$ for $k=1,2$.  Then,  the transmitted signal at transmitter~$k$ takes the following form
 \begin{align}
  x_k  =   \sqrt{P^{ -(\alpha -1)}}  h_{\ell\ell} \cdot v_{k,c}    +     h_{k\ell}  \cdot u_{k}    \label{eq:xvkkk322}
 \end{align}
  for $\ell \neq k$,  $k, \ell =1,2$.   
Then the received signals at the receivers take the following forms 
\begin{align}
y_{1} &=    \sqrt{P^{ 2- \alpha}} h_{11}h_{22} v_{1,c}   +    \sqrt{P}  h_{12} h_{11}(   v_{2,c}  +  u_1)    +   \sqrt{P^{ \alpha }}  h_{12} h_{21} u_{2}   +  z_{1}    \label{eq:yvk1322}  \\
y_{2} &=      \sqrt{P^{ 2- \alpha}} h_{22}h_{11} v_{2,c}      +     \sqrt{P} h_{21} h_{22}(  v_{1,c}  +    u_2)       +  \sqrt{P^{ \alpha  }}   h_{21} h_{12} u_{1}    +  z_{2} .   \label{eq:yvk2322}  
\end{align}

From \eqref{eq:Rk623J} and \eqref{eq:Rk623bJ},  the proposed scheme achieves  the following  secure rates: $R_1  =    \Imu(v_1; y_1) - \Imu(v_1; y_2 | v_2 ) $ and $R_2  =    \Imu(v_2; y_2) - \Imu(v_2; y_1 | v_1 ) $.
For this case, $v_k$ takes the form \[v_k =  \sqrt{P^{ -(\alpha -1)}}  v_{k,c} , \quad k=1,2 \] (see \eqref{eq:xvk} and \eqref{eq:para333}).   
From $y_k$ expressed in \eqref{eq:yvk1322} and \eqref{eq:yvk2322},  $v_{k,c}$  can be  estimated by using a successive decoding method, $k=1,2$. 
The following lemma provides a result on the error probability for this estimation.

 \begin{lemma}  \label{lm:rateerror322}
When $3/2 \leq \alpha \leq 2$, given the signal design in \eqref{eq:constellationGsym1}-\eqref{eq:constellationGsym1u} and \eqref{eq:para333}-\eqref{eq:xvkkk322},  the error probability of  estimating  $v_{k,c}$ from $y_k$  is
 \begin{align}
 \text{Pr} [ v_{k,c} \neq \hat{v}_{k,c}  ]  \to 0         \quad \text {as}\quad  P\to \infty    \label{eq:error1c1p322}
 \end{align}
for $k=1,2$, where  $ \hat{v}_{k,c}$  is  the corresponding estimate for  $v_{k,c}$  based on the observation $y_k$ expressed in \eqref{eq:yvk1322} and \eqref{eq:yvk2322}. 
 \end{lemma}

\begin{proof}
See Appendix~\ref{sec:rateerror322}.
\end{proof}

We will proceed with the lower bound on the secure rate $R_1$. 
In this case, given  that $v_{1,c}$  is  \emph{uniformly}  drawn from  $\Omega ( \xi =  \gamma_{v_{1,c}} \cdot \frac{ 1}{Q} ,   \   Q =  P^{ \frac{2 - \alpha  - \epsilon }{2}} )$ and that $v_1 =  \sqrt{P^{ -(\alpha -1)}}  v_{1,c}$,  the rate of $v_{1} $ is  
  \begin{align}
 \Hen(v_{1}) =  \Hen(v_{1,c}) &=  \log (2 \cdot P^{ \frac{2 - \alpha  - \epsilon }{2}} +1)   \label{eq:rateAWGNIC1c322}     
 \end{align}
 where $\epsilon$ can be set as $\epsilon \to 0$. 
Then, $\Imu(v_1; y_1)$ can be lower bounded by
  \begin{align}
  \Imu(v_1; y_1)      &=   \Hen(v_1) -   \Hen(v_1  |  y_1)    \non    \\
    &\geq    \Hen(v_1) -       \bigl( 1+    \text{Pr} [  v_{1,c} \neq \hat{v}_{1,c}  \}  ] \cdot \Hen(v_{1}) \bigr)  \label{eq:rateAWGNIC1404322}     \\
        & =     \bigl( 1 -   \text{Pr} [   v_{1,c} \neq \hat{v}_{1,c}  ] \bigr)   \cdot \Hen(v_{1})  - 1  \label{eq:rateAWGNIC1405322}       \\
        & =  \frac{ 2 - \alpha  - \epsilon}{2} \log P + o(\log P)   \label{eq:rateAWGNIC17762322}  
 \end{align}
where
\eqref{eq:rateAWGNIC1404322} results from Fano's inequality;
and \eqref{eq:rateAWGNIC17762322}   follows from  \eqref{eq:rateAWGNIC1c322} and Lemma~\ref{lm:rateerror322}.
On the other hand, by following the steps \eqref{eq:rateAWGNIC18374}-\eqref{eq:rateAWGNIC9562}, the term $\Imu(v_1; y_2 | v_2 )$ can be bounded as 
  \begin{align}
&\Imu(v_1; y_2 | v_2 )     \non\\
\leq &  \Imu(v_1; y_2, u_1,  v_{1,c} + u_2  | v_2 )     \non \\  
=  &    \Imu(v_1;    v_{1,c} + u_2  | u_1,  v_2 )   + \Imu(v_1;    y_2  | u_1,  v_2, v_{1,c} + u_2 )   + \Imu(v_1; u_1  | v_2 )    \non  \\
=   & \Imu(v_1;    v_{1,c} + u_2  )   +  \underbrace{ \Imu(v_1;      z_{2}   | u_1,  v_2 , v_{1,c} + u_2 ) }_{=0}    \label{eq:rateAWGNIC18735322}   \\
=   & \Hen( v_{1,c} + u_2  )  - \Hen(u_2  )     \label{eq:rateAWGNIC77364322}    \\
\leq &  \log (4  \cdot P^{ \frac{ 2 - \alpha - \epsilon}{2}} +1)    -   \log (2  \cdot P^{ \frac{ 2 - \alpha - \epsilon}{2}} +1)  \label{eq:rateAWGNIC1955322}   \\
\leq &  1      \label{eq:rateAWGNIC9562322}  
 \end{align}
where   \eqref{eq:rateAWGNIC18735322} and \eqref{eq:rateAWGNIC77364322} follow from the fact that $v_1, v_2, u_1, u_2 $ and $z_2$ are mutually independent;
 note that $v_1 =  \sqrt{P^{ -(\alpha -1)}}  v_{1,c}$ for this case;
 \eqref{eq:rateAWGNIC1955322} results from the facts that  $\Hen(u_2  ) =   \log (2  \cdot P^{ \frac{ 2 - \alpha - \epsilon}{2}} +1) $,  and  that $\Hen( v_{1,c} + u_2  ) \leq  \log (4  \cdot P^{ \frac{ 2 - \alpha - \epsilon}{2}} +1)$.

Finally,  by combining the results in \eqref{eq:rateAWGNIC17762322}  and \eqref{eq:rateAWGNIC9562322}, the secure rate $R_1$ is lower bounded by 
\begin{align}
R_1  & =    \Imu(v_1; y_1) - \Imu(v_1; y_2 | v_2 )      \non \\
& \geq  \frac{ 2 - \alpha  - \epsilon}{2} \log P + o(\log P).  \label{eq:lbounditfinal1322}
\end{align}
Due to the symmetry, by interchanging the roles of users, it gives the lower bound on the secure rate $R_2$: 
\begin{align}
R_2   & =    \Imu(v_2; y_2) - \Imu(v_2; y_1 | v_1 )  \non\\
& \geq     \frac{ 2 - \alpha  - \epsilon}{2} \log P  + o(\log P).  \label{eq:lbounditfinal2322}
\end{align}
By setting $\epsilon \to 0$, then the secure GDoF pair $(d_1 =  2 - \alpha ,  d_2 = 2 - \alpha  )$ is  achievable by the proposed cooperative jamming scheme when  $3/2 \leq  \alpha \leq 2$.

\subsection{Rate analysis when $3/4 \leq  \alpha \leq  1$   \label{sec:CJscheme341}}
 
The rate analysis for this case with $3/4 \leq  \alpha \leq  1$, as well as for the case with $1 \leq  \alpha \leq  3/2$ described in the next subsection, is different from that for the previous two cases with $2/3 <  \alpha \leq  3/4$ and $3/2 \leq  \alpha \leq 2$.  In the previous two cases, the information signal $v_k$ can be  estimated from $y_k$  by using a successive decoding method and the error probability of this estimation vanishes as $P \to \infty$  (see Lemma~\ref{lm:rateerror2334} and Lemma~\ref{lm:rateerror322}),  which can be used to bound the secure rates $R_k$ for $k=1,2$ (see \eqref{eq:rateAWGNIC1202}-\eqref{eq:lbounditfinal2} and \eqref{eq:rateAWGNIC1404322}- \eqref{eq:lbounditfinal2322}). 
However, in this case, the previous successive decoding method cannot be used in the rate analysis. Instead, we will use the approaches of noise removal and signal separation.

For the case with $3/4 \leq  \alpha \leq  1$, we design the parameters  such that
 \begin{align}
\beta_{v_{1,c}}&= \beta_{v_{2,c}} =0,        \quad\quad \quad    \quad\quad       \beta_{u_1}=  \beta_{u_2}= 1- \alpha,   \quad  \quad  \quad \quad\quad \beta_{v_{1,p}}=\beta_{v_{2,p}}= \alpha  \label{eq:para34111} \\
 \lambda_{v_{1,c}} &= \lambda_{v_{2,c}}= \alpha/3 - \epsilon,\quad   \quad  \  \lambda_{u_1}=  \lambda_{u_2}= \alpha/3 - \epsilon,   \quad \quad  \quad  \lambda_{v_{1,p}}= \lambda_{v_{2,p}}  = 1 - \alpha - \epsilon  \label{eq:para34122}
 \end{align}
 where  $\epsilon>0$ can be set arbitrarily small.
 In this case, the transmitted signal at transmitter~$k$ takes the following form
 \begin{align}
  x_k  =    h_{\ell\ell} \cdot v_{k,c} +   \sqrt{P^{ -\alpha}}    \cdot h_{\ell\ell} \cdot v_{k,p} +     \sqrt{P^{ -(1- \alpha)}}  \cdot  h_{k\ell}  \cdot u_{k}    \label{eq:xvkkk1341}
 \end{align}
  for $\ell \neq k$,  $k, \ell =1,2$.   
Then the received signals at the receivers become
\begin{align}
y_{1} &=    \sqrt{P} h_{11}h_{22} v_{1,c}  +    \sqrt{P^{ 1 - \alpha}} h_{11}h_{22} v_{1,p}    +    \sqrt{P^{ \alpha}}  h_{12} h_{11}(   v_{2,c}  +  u_1)    +   \sqrt{P^{ 2\alpha - 1 }}  h_{12} h_{21} u_{2}  +    h_{12} h_{11}  v_{2,p}   +  z_{1}    \label{eq:yvk1341}  \\
y_{2} &=     \sqrt{P} h_{22}h_{11} v_{2,c}    +    \sqrt{P^{ 1 - \alpha}} h_{22}h_{11} v_{2,p}    +     \sqrt{P^{ \alpha}} h_{21} h_{22}(  v_{1,c}  +    u_2)       +  \sqrt{P^{ 2\alpha - 1 }}   h_{21} h_{12} u_{1}  +    h_{21} h_{22}  v_{1,p}   +  z_{2} .   \label{eq:yvk2341}  
\end{align}

From \eqref{eq:Rk623J} and \eqref{eq:Rk623bJ},  the proposed scheme achieves  the following  secure rates: $R_1  =    \Imu(v_1; y_1) - \Imu(v_1; y_2 | v_2 ) $ and $R_2  =    \Imu(v_2; y_2) - \Imu(v_2; y_1 | v_1 ) $.
Let us first focus on the lower bound of $R_1$. 
For this case, $v_1$ takes the form $v_1 = v_{1,c} +   \sqrt{P^{ -\alpha}}   \cdot v_{1,p}$.   
From $y_1$ expressed in \eqref{eq:yvk1341},  we will show that $\{v_{1,c}, v_{1,p}\}$  can be  estimated with vanishing error probability when $P$ is large, given the  design in \eqref{eq:constellationGsym1}-\eqref{eq:constellationGsym1u} and \eqref{eq:para34111}-\eqref{eq:xvkkk1341}. 
The following lemma provides a result on the error probability for this estimation.

 \begin{lemma}  \label{lm:rateerror341}
 When $3/4 \leq \alpha \leq 1$, given the signal design in \eqref{eq:constellationGsym1}-\eqref{eq:constellationGsym1u} and \eqref{eq:para34111}-\eqref{eq:xvkkk1341}, then for almost all channel coefficients  $\{h_{k\ell}\} \in (1, 2]^{2\times 2}$,  the error probability of estimating  $\{v_{k,c}, v_{k,p} \}$ from $y_k$  is
 \begin{align}
 \text{Pr} [  \{ v_{k,c} \neq \hat{v}_{k,c} \} \cup  \{ v_{k,p} \neq \hat{v}_{k,p} \}  ]  \to 0         \quad \text {as}\quad  P\to \infty    \label{eq:error1c1p341}
 \end{align}
for $k=1,2$, where  $ \hat{v}_{k,c}$ and  $\hat{v}_{k,p}  $ are the corresponding estimates for  $v_{k,c}$ and $v_{k,p}$, respectively, based on the observation $y_k$ expressed in \eqref{eq:yvk1341} and \eqref{eq:yvk2341}. 
 \end{lemma}
 \begin{proof}
In this proof  we use the approaches of noise removal and signal separation. The full details are described in  Section~\ref{sec:rateerror341}.
  \end{proof}
 
With the result of Lemma~\ref{lm:rateerror341}, we  proceed to bound the secure rate $R_1$ by following the steps in \eqref{eq:rateAWGNIC111}-\eqref{eq:lbounditfinal2} that were used for the previous case of $2/3 <  \alpha \leq  3/4$. 
For this case of  $3/4 \leq \alpha \leq 1$, since   $v_{1,c} \in    \Omega (\xi   =  \gamma_{v_{1,c}} \cdot \frac{ 1}{Q},   \   Q =  P^{ \frac{\alpha/3 - \epsilon}{2}} ) $ and $v_{1,p}  \in    \Omega (\xi   =\gamma_{v_{1,p}} \cdot \frac{ 1}{Q},   \   Q = P^{ \frac{ 1 - \alpha - \epsilon}{2}} ) $, we have
  \begin{align}
  \Hen(v_{1,c}) &=  \log (2 \cdot P^{ \frac{ \alpha/3  - \epsilon}{2}} +1)   \label{eq:rateAWGNIC111341}     \\
  \Hen (v_{1,p}) &=  \log (2 \cdot P^{ \frac{ 1 - \alpha - \epsilon}{2}} +1)  \label{eq:rateAWGNIC222341}                              
 \end{align}
 where $\epsilon$ can be set as $\epsilon \to 0$. 
 Due to the signal design $v_{1} \defeq   v_{1,c} +   \sqrt{P^{ - \alpha}}    \cdot  v_{1,p} $ (see \eqref{eq:xvk} and \eqref{eq:para34111}), it is true that  $\{ v_{1,p}, v_{1,c}\}$ can be reconstructed from $v_{1}$,  and vice versa. 
Then, 
  \begin{align}
 \Hen(v_{1}) =   \Hen(v_{1,c}, v_{1,p}) &= \Hen(v_{1,c})+  \Hen(v_{1,p}) 
 =   \frac{1- 2\alpha/3  - 2\epsilon}{2} \log P + o(\log P) . \label{eq:rateAWGNIC333341}                                  
 \end{align}
At this point,   $\Imu(v_1; y_1)$ can be lower bounded by
  \begin{align}
  \Imu(v_1; y_1)   &\geq   \Imu(v_1; \hat{v}_{1,c}, \hat{v}_{1,p})  \label{eq:rateAWGNIC1202341}     \\
  &=   \Hen(v_1) -   \Hen(v_1  |  \hat{v}_{1,c}, \hat{v}_{1,p})    \non    \\
    &\geq    \Hen(v_1) -       \bigl( 1+    \text{Pr} [  \{ v_{1,c} \neq \hat{v}_{1,c} \} \cup  \{ v_{1,p} \neq \hat{v}_{1,p} \}  ] \cdot \Hen(v_{1}) \bigr)  \label{eq:rateAWGNIC1404341}     \\
        & =     \bigl( 1 -   \text{Pr} [  \{ v_{1,c} \neq \hat{v}_{1,c} \} \cup  \{ v_{1,p} \neq \hat{v}_{1,p} \}  ] \bigr)   \cdot \Hen(v_{1})  - 1  \label{eq:rateAWGNIC1405341}    \\
       & =   \frac{1- 2\alpha/3  - 2\epsilon}{2} \log P + o(\log P)      \label{eq:rateAWGNIC17762341}  
 \end{align}
 for  almost all channel coefficients  $\{h_{k\ell}\} \in (1, 2]^{2\times 2}$, where \eqref{eq:rateAWGNIC1202341} results from the Markov chain $v_1 \to y_1 \to  \{\hat{v}_{1,c}, \hat{v}_{1,p} \}$;
\eqref{eq:rateAWGNIC1404341} stems from Fano's inequality;
and \eqref{eq:rateAWGNIC17762341}  follows from  \eqref{eq:rateAWGNIC333341} and  Lemma~\ref{lm:rateerror341}. 

On the other hand,  $\Imu(v_1; y_2 | v_2 )$ can be bounded  as
  \begin{align}
  &\Imu(v_1; y_2 | v_2 )     \non\\
\leq   & \Hen( v_{1,c} + u_2  )  - \Hen(u_2  )     +  \hen( h_{21} h_{22}  v_{1,p}   +  z_{2}   | u_1,  v_2 , v_{1,c} + u_2 )   -  \hen(   z_{2} )    \label{eq:rateAWGNIC77364341}    \\
\leq &  \underbrace{ \log (4  \cdot P^{ \frac{  \alpha/3 - \epsilon}{2}} +1)    -   \log (2  \cdot P^{ \frac{  \alpha/3 - \epsilon}{2}} +1)}_{\leq  \log 2 }    +   \underbrace{\hen( h_{21} h_{22}  v_{1,p}   +  z_{2} )}_{\leq  \frac{1}{2} \log ( 2 \pi e \times 17) }   -  \frac{1}{2}\log (2 \pi e)    \label{eq:rateAWGNIC1955341}   \\
\leq  & \log (2\sqrt{17})    \label{eq:rateAWGNIC9562341}  
 \end{align}
where \eqref{eq:rateAWGNIC77364341} follows from the steps in \eqref{eq:rateAWGNIC18374}-\eqref{eq:rateAWGNIC77364};    \eqref{eq:rateAWGNIC1955341} holds true because $ \hen(   z_{2} )  =  \frac{1}{2}\log (2\pi e)$,  $\Hen(u_2  ) =   \log (2  \cdot P^{ \frac{  \alpha/3 - \epsilon}{2}} +1) $,  and $\Hen( v_{1,c} + u_2  ) \leq  \log (4  \cdot P^{ \frac{  \alpha/3 - \epsilon}{2}} +1)$; 
\eqref{eq:rateAWGNIC9562341} follows from the identity that $\log (4a +1) - \log (2 a +1) \leq  \log 2 $ for any $a \geq 0$ and  the fact that $\hen( h_{21} h_{22}  v_{1,p}   +  z_{2} ) \leq   \frac{1}{2} \log ( 2 \pi e ( | h_{21}|^2 \cdot | h_{22}|^2\cdot \E |v_{1,p}|^2     +  \E|  z_{2}|^2 ))  \leq   \frac{1}{2} \log ( 2 \pi e \times 17)  $.

Finally,  with the results in \eqref{eq:rateAWGNIC17762341}  and \eqref{eq:rateAWGNIC9562341}, the secure rate $R_1$ is lower bounded by 
\begin{align}
R_1  & =    \Imu(v_1; y_1) - \Imu(v_1; y_2 | v_2 )   \geq   \frac{1- 2\alpha/3 - 2\epsilon}{2} \log P + o(\log P)  \label{eq:lbounditfinal1001}
\end{align}
for  almost all the channel coefficients.
Due to the symmetry, by interchanging the roles of users, it gives the lower bound on the secure rate $R_2 \geq     \frac{1- 2\alpha/3 - 2\epsilon}{2} \log P + o(\log P)$.
 By setting $\epsilon \to 0$, then the secure GDoF pair $(d_1 =  1- 2\alpha/3,  d_2 =  1- 2\alpha/3 )$ is  achievable by the proposed cooperative jamming scheme for  almost all the channel coefficients  $\{h_{k\ell}\} \in (1, 2]^{2\times 2}$,  for the case of  $3/4 \leq \alpha \leq  1$.

\subsection{Rate analysis when $1 \leq  \alpha \leq  3/2$   \label{sec:CJscheme132}}
 
The rate analysis for this case with $1 \leq  \alpha \leq  3/2$  also uses  the approaches of noise removal and signal separation.  For this case, we design the parameters  such that
 \begin{align}
\beta_{v_{1,c}}&= \beta_{v_{2,c}} =\alpha -1 ,        \quad\quad \quad    \quad\quad       \beta_{u_1}=  \beta_{u_2}= 0 ,   \quad  \quad  \quad \quad \beta_{v_{1,p}}=\beta_{v_{2,p}}= \infty  \label{eq:para13211} \\
 \lambda_{v_{1,c}} &= \lambda_{v_{2,c}}= \alpha/3 - \epsilon,\quad   \quad  \  \lambda_{u_1}=  \lambda_{u_2}= \alpha/3 - \epsilon,   \quad \quad  \quad  \lambda_{v_{1,p}}= \lambda_{v_{2,p}}  =0   \label{eq:para13222}
 \end{align}
 where  $\epsilon>0$ can be set arbitrarily small.
 In this case, the transmitted signal at transmitter~$k$ consists of only two symbols:
  \begin{align}
  x_k  =   \sqrt{P^{ -(\alpha -1)}}   h_{\ell\ell} \cdot v_{k,c}   +     h_{k\ell}  \cdot u_{k}    \label{eq:xvkkk1132}
 \end{align}
  for $\ell \neq k$,  $k, \ell =1,2$.   
Then the received signals at the receivers take the following forms
\begin{align}
y_{1} &=    \sqrt{P^{ 2- \alpha}} h_{11}h_{22} v_{1,c}   +    \sqrt{P}  h_{12} h_{11}(   v_{2,c}  +  u_1)    +   \sqrt{P^{ \alpha }}  h_{12} h_{21} u_{2}   +  z_{1}    \label{eq:yvk1322132}  \\
y_{2} &=      \sqrt{P^{ 2- \alpha}} h_{22}h_{11} v_{2,c}      +     \sqrt{P} h_{21} h_{22}(  v_{1,c}  +    u_2)       +  \sqrt{P^{ \alpha  }}   h_{21} h_{12} u_{1}    +  z_{2} .   \label{eq:yvk2322132}  
\end{align}

In the following we will bound the  secure rates expressed in \eqref{eq:Rk623J} and \eqref{eq:Rk623bJ}. 
For this case, $v_k$ takes the form \[v_k =  \sqrt{P^{ -(\alpha -1)}}  v_{k,c} , \quad k=1,2 \] (see \eqref{eq:xvk} and \eqref{eq:para13211}).   
From $y_k$ expressed in \eqref{eq:yvk1322132} and \eqref{eq:yvk2322132},  $v_{k,c}$  can be  estimated by using the approaches of noise removal and signal separation, $k=1,2$.  The following lemma provides a result on the error probability for this estimation.

  \begin{lemma}  \label{lm:rateerror132}
 When $1 \leq \alpha \leq 3/2$, given the signal design in \eqref{eq:constellationGsym1}-\eqref{eq:constellationGsym1u} and \eqref{eq:para13211}-\eqref{eq:xvkkk1132}, then for almost all channel coefficients  $\{h_{k\ell}\} \in (1, 2]^{2\times 2}$,  the error probability of estimating  $v_{k,c}$ from $y_k$  is
 \begin{align}
 \text{Pr} [  v_{k,c} \neq \hat{v}_{k,c}  ]  \to 0         \quad \text {as}\quad  P\to \infty    \label{eq:error1c1p132}
 \end{align}
for $k=1,2$, where  $ \hat{v}_{k,c}$  is the corresponding estimate for  $v_{k,c}$ based on the observation $y_k$ expressed in \eqref{eq:yvk1322132} and \eqref{eq:yvk2322132}. 
 \end{lemma}
 \begin{proof}
In this proof  we use the approaches of noise removal and signal separation. The full details are described in  Appendix~\ref{sec:rateerror132}.
  \end{proof}

We will proceed with the lower bound on the secure rate $R_1$ expressed in \eqref{eq:Rk623J}.
In this case, given  that $v_{1,c} \in \Omega ( \xi = 2 \gamma \cdot \frac{ 1}{Q} ,   \   Q =  P^{ \frac{\alpha/3 - \epsilon }{2}} )$ and that $v_1 =  \sqrt{P^{ -(\alpha -1)}}  v_{1,c}$,  the rate of $v_{1} $ is  
  \begin{align}
 \Hen(v_{1}) =  \Hen(v_{1,c}) &=  \log (2 \cdot P^{ \frac{\alpha/3 - \epsilon }{2}} +1)   \label{eq:rateAWGNIC1c132}     
 \end{align}
 where $\epsilon$ can be set as $\epsilon \to 0$. 
Then, $\Imu(v_1; y_1)$ can be lower bounded by
  \begin{align}
  \Imu(v_1; y_1)      &=   \Hen(v_1) -   \Hen(v_1  |  y_1)    \non    \\
    &\geq    \Hen(v_1) -       \bigl( 1+    \text{Pr} [  v_{1,c} \neq \hat{v}_{1,c}  \}  ] \cdot \Hen(v_{1}) \bigr)   \non    \\
        & =  \frac{ \alpha/3 - \epsilon}{2} \log P + o(\log P)    \label{eq:rateAWGNIC17762132}  
 \end{align}
 for  almost all the channel coefficients  $\{h_{k\ell}\} \in (1, 2]^{2\times 2}$, where  \eqref{eq:rateAWGNIC17762132}   follows from  \eqref{eq:rateAWGNIC1c132} and Lemma~\ref{lm:rateerror132}.
On the other hand, by following the steps related to \eqref{eq:rateAWGNIC18735322}-\eqref{eq:rateAWGNIC9562322},  $\Imu(v_1; y_2 | v_2 )$  can be bounded as
  \begin{align}
\Imu(v_1; y_2 | v_2 )   \leq   1 .     \label{eq:rateAWGNIC9562132}  
 \end{align}

Finally,  by incorporating \eqref{eq:rateAWGNIC17762132}  and \eqref{eq:rateAWGNIC9562132} into \eqref{eq:Rk623J}, the secure rate $R_1$ is lower bounded by 
\begin{align}
R_1  & =    \Imu(v_1; y_1) - \Imu(v_1; y_2 | v_2 )      \non \\
& \geq  \frac{ \alpha/3  - \epsilon}{2} \log P + o(\log P)  \label{eq:lbounditfinal1132}
\end{align}
for  almost all the channel coefficients.
Due to the symmetry, by interchanging the roles of users, it gives the lower bound on the second secure rate:  $R_2  =    \Imu(v_2; y_2) - \Imu(v_2; y_1 | v_1 )  \geq     \frac{ \alpha/3 - \epsilon}{2} \log P  + o(\log P)$.
 By setting $\epsilon \to 0$, then the secure GDoF pair $(d_1 =  \alpha/3 ,  d_2 = \alpha/3  )$ is  achievable by the proposed cooperative jamming scheme for  almost all the channel coefficients  $\{h_{k\ell}\} \in (1, 2]^{2\times 2}$,  for the case of  $1 \leq \alpha \leq  3/2$.

\section{Proof of Lemma~\ref{lm:rateerror341}  \label{sec:rateerror341} }

We will prove that  when $3/4 \leq \alpha \leq 1$, given the signal design in \eqref{eq:constellationGsym1}-\eqref{eq:constellationGsym1u} and \eqref{eq:para34111}-\eqref{eq:xvkkk1341}, then for almost all the channel coefficients  $\{h_{k\ell}\} \in (1, 2]^{2\times 2}$,  the error probability of estimating  $\{v_{k,c}, v_{k,p} \}$ from $y_k$  is
 \begin{align}
 \text{Pr} [  \{ v_{k,c} \neq \hat{v}_{k,c} \} \cup  \{ v_{k,p} \neq \hat{v}_{k,p} \}  ]  \to 0         \quad \text {as}\quad  P\to \infty    \label{eq:error1c1p341a}
 \end{align}
for $k=1,2$, where  $ \hat{v}_{k,c}$ and  $\hat{v}_{k,p}  $ are the corresponding estimates for  $v_{k,c}$ and $v_{k,p}$, respectively, based on the observation $y_k$ expressed in \eqref{eq:yvk1341} and \eqref{eq:yvk2341}. 
In this case, we will use the approaches of noise removal and signal separation that will be discussed below. 

Due to the symmetry we will focus on the proof for the first user ($k=1$). 
At first,  we will estimate three symbols
  \begin{align}
v_{1,c} &\in \Omega ( \xi =  \gamma_{v_{1,c}} \cdot \frac{ 1}{Q} ,   \   Q =  P^{ \frac{ \alpha/3 - \epsilon }{2}} )    \non \\
v_{2,c}  +  u_1  &\in 2 \cdot \Omega ( \xi =  \gamma_{v_{2,c}} \cdot \frac{ 1}{Q} ,   \   Q =  P^{ \frac{ \alpha/3 - \epsilon }{2}} ) \non\\
u_{2}  &\in \Omega ( \xi =  \gamma_{u_{2}} \cdot \frac{ 1}{Q} ,   \   Q =  P^{ \frac{ \alpha/3 - \epsilon }{2}} )  \non
 \end{align}
 simultaneously from $y_1$  by treating the other signals as noise, where $y_1$ is expressed in \eqref{eq:yvk1341}. Note that we will estimate the sum $v_{2,c}  +  u_1$ but not the individual symbols $v_{2,c}$ and $u_1$.  From  $y_1$ expressed in \eqref{eq:yvk1341}, it  can be rewritten as 
\begin{align}
 y_1  &=    \sqrt{P} h_{11}h_{22} v_{1,c}  +    \sqrt{P^{ 1 - \alpha}} h_{11}h_{22} v_{1,p}    +    \sqrt{P^{ \alpha}}  h_{12} h_{11}(   v_{2,c}  +  u_1)    +   \sqrt{P^{ 2\alpha - 1 }}  h_{12} h_{21} u_{2}  +    h_{12} h_{11}  v_{2,p}   +  z_{1}   \non\\
 &=   \sqrt{P^{ 2\alpha - 1 }} (h_{12} h_{21}  u_{2}  +  \sqrt{P^{1-  \alpha }}  h_{12} h_{11} (   v_{2,c}  +  u_1)  +  \sqrt{P^{2- 2\alpha}}  h_{11}h_{22} v_{1,c})   +  \tilde{z}_{1}     \non \\
         &=    \sqrt{P^{ 5\alpha/3 - 1  + \epsilon}} \cdot 2\gamma \cdot ( g_0 q_0 + \sqrt{P^{1-  \alpha }} g_1 q_1 + \sqrt{P^{2- 2\alpha}} g_2 q_2 )   +  \tilde{z}_{1}    \label{eq:yvk4835}
\end{align}
where   $\tilde{z}_{1}   \defeq     \sqrt{P^{ 1 - \alpha}} h_{11}h_{22} v_{1,p}  +    h_{12} h_{11}  v_{2,p}   +  z_{1}$ and 
 \[ g_0\defeq  h_{12} h_{21},    \quad  g_1\defeq   h_{12} h_{11}  , \quad   g_2 \defeq h_{11}h_{22} \] 
    \[ q_0  \defeq  \frac{Q_{\max}}{2\gamma} \cdot   u_{2}  ,    \quad  q_1  \defeq  \frac{Q_{\max}}{2\gamma} \cdot   (   v_{2,c}  +  u_1), \quad   q_2  \defeq  \frac{ Q_{\max}}{2\gamma} \cdot  v_{1,c}, \quad Q_{\max} \defeq P^{ \frac{ \alpha/3 - \epsilon }{2}} \]
   for a given constant $\gamma  \in \bigl(0, \frac{1}{4\sqrt{2}}\bigr]$ (see \eqref{eq:gammadef}).
Based on our definitions, it holds true that $q_0, q_1, q_2 \in \Zc$ and $|q_0| \leq Q_{\max}$, $|q_1| \leq 2 Q_{\max}$,  $|q_2| \leq Q_{\max}$.
From the definition in \eqref{eq:alpha11},  i.e., $P \defeq \max_{k}\{ 2^{2 m_{kk}} \}$ and $\sqrt{P^{\alpha_{k\ell}}} =   2^{m_{k\ell}} ,  k, \ell =1,2$, it implies that $ \sqrt{P^{1-  \alpha }} \in  \Zc^+$ and $\sqrt{P^{2- 2\alpha}} \in \Zc^+$ for this case with $3/4 \leq \alpha \leq 1$.

From $y_1$ expressed in \eqref{eq:yvk4835}, $q_0, q_1, q_3$ can be estimated by using a demodulator,  which  searches for the corresponding estimates $\hat{q}_0, \hat{q}_1, \hat{q}_3$  by minimizing 
\[ |y_1 -     \sqrt{P^{ 5\alpha/3 - 1  + \epsilon}} \cdot 2\gamma \cdot ( g_0 \hat{q}_0 + \sqrt{P^{1-  \alpha }} g_1 \hat{q}_1 + \sqrt{P^{2- 2\alpha}} g_2 \hat{q}_2 ) | .   \]
Let us consider the minimum distance between the  signals generated by  $(q_0, q_1, q_2)$ and $(q_0', q_1', q_2') $, defined as
  \begin{align}
d_{\min}  (g_0, g_1, g_2)   \defeq    \min_{\substack{ q_0, q_2, q_0', q_2' \in \Zc  \cap [- Q_{\max},    Q_{\max}]  \\  q_1, q_1' \in \Zc  \cap [- 2Q_{\max},   2 Q_{\max}]   \\  (q_0, q_1, q_2) \neq  (q_0', q_1', q_2')  }}  | g_0  (q_0 - q_0') + \sqrt{P^{1-  \alpha }} g_1 (q_1 - q_1') + \sqrt{P^{2- 2\alpha}} g_2 (q_2 - q_2')  | .    \label{eq:minidis111}
 \end{align}
The following Lemma~\ref{lm:distance3432} shows that, given the signal design in \eqref{eq:constellationGsym1}-\eqref{eq:constellationGsym1u}
and \eqref{eq:para34111}-\eqref{eq:xvkkk1341},  the minimum distance $d_{\min}$ defined in \eqref{eq:minidis111} is sufficiently large for almost all the channel  coefficients $\{h_{k\ell}\} \in (1, 2]^{2\times 2} $ when $P$ is large. Before providing Lemma~\ref{lm:distance3432}, we will present an existing result from \cite{NM:13}, which will be used in the proof of Lemma~\ref{lm:distance3432}.

\begin{lemma} \cite[Lemma~14]{NM:13} \label{lm:NMb}
Let $\beta \in (0,1]$, $A_1, A_2 \in \Zc^+$, and $Q_0, Q_1, Q_2 \in \Zc^+$. Define the event
\[  B'(q_0, q_1, q_2)  \defeq \{ (g_0, g_1, g_2)  \in (1,4]^3 :  |  g_0q_0 + A_1 g_1 q_1 + A_2 g_2 q_2    | < \beta \}    \]
and set 
\begin{align*}
B'  \defeq   \bigcup_{\substack{ q_0, q_1, q_2 \in \Zc:  \\  (q_0, q_1, q_2) \neq  0,  \\  |q_k| \leq Q_k  \ \forall k }}  B'(q_0, q_1, q_2) . 
\end{align*}
Then the Lebesgue measure of $B' $, denoted by $\Lc (B' )$,  is bounded by
\begin{align*}
\Lc (B' )  \leq   504 \beta \Bigl( & 2 \min \bigl\{ Q_2,  \frac{Q_0}{A_2} \bigr\}  + \min \bigl\{ Q_1 \tilde{Q}_2,  \frac{Q_0 \tilde{Q}_2}{A_1},  \frac{A_2 \tilde{Q}_2 \tilde{Q}_2}{A_1}   \bigr\} \\  + &2 \min \bigl\{ Q_1,  \frac{Q_0}{A_1} \bigr\}  + \min \bigl\{ Q_2 \tilde{Q}_1,  \frac{Q_0 \tilde{Q}_1}{A_2},  \frac{A_1 \tilde{Q}_1 \tilde{Q}_1}{A_2}   \bigr\}        \Bigr)
\end{align*}
where 
\begin{align*}
\tilde{Q}_1 &\defeq \min\Bigl\{Q_1,  8\cdot \frac{\max\{Q_0, A_2 Q_2\}}{A_1}\Bigr\} \\
\tilde{Q}_2 &\defeq \min\Bigl\{Q_2,  8\cdot \frac{\max\{Q_0, A_1 Q_1\}}{A_2}\Bigr\} .
\end{align*}
\end{lemma}

\vspace{10pt}

\begin{lemma}  \label{lm:distance3432}
Consider the case $\alpha \in [3/4, 1]$, and  consider some constants $\delta \in (0, 1]$ and  $\epsilon >0$.   Given the signal design in \eqref{eq:constellationGsym1}-\eqref{eq:constellationGsym1u} and \eqref{eq:para34111}-\eqref{eq:xvkkk1341},  then the minimum distance $d_{\min}$ defined in \eqref{eq:minidis111} is bounded by
 \begin{align}
d_{\min}    \geq   \delta P^{- \frac{8\alpha/3 -2 }{2}}    \label{eq:distancegeq}
 \end{align}
for all  the channel  coefficients $\{h_{k\ell}\} \in (1, 2]^{2\times 2} \setminus \Ho$, and the Lebesgue measure of the outage set  $\Ho \subseteq (1,2]^{2\times 2}$ , denoted by $\mathcal{L}(\Ho)$,  satisfies  
 \begin{align}
\mathcal{L}(\Ho) \leq 258048 \delta   \cdot     P^{ - \frac{ \epsilon  }{2}}.    \label{eq:LebmeaB}
 \end{align}
 \end{lemma}
 \begin{proof}
We consider the case of $3/4 \leq \alpha \leq 1$. Let 
\[ \beta \defeq  \delta  P^{- \frac{8\alpha/3 -2 }{2}} , \quad A_1 \defeq   P^{\frac{1-  \alpha}{2} } ,  \quad  A_2 \defeq   P^{1-  \alpha } \]
 \[ g_0\defeq  h_{12} h_{21},    \quad  g_1\defeq   h_{12} h_{11}  , \quad   g_2 \defeq h_{11}h_{22} \] 
\[ Q_0 \defeq 2 Q_{\max},    \quad Q_1\defeq 4 Q_{\max} , \quad Q_2 \defeq 2 Q_{\max},  \quad  Q_{\max} \defeq P^{ \frac{ \alpha/3 - \epsilon }{2}}       \]
for some $\epsilon >0$ and $\delta \in (0, 1]$.
Let us define the event
\begin{align}
B(q_0, q_1, q_2)  \defeq \{ (g_0, g_1, g_2)  \in (1,4]^3 :  |  g_0q_0 + A_1 g_1 q_1 + A_2 g_2 q_2    | < \beta \}       \label{eq:Boutage11}
\end{align}
and set 
\begin{align}
B  \defeq   \bigcup_{\substack{ q_0, q_1, q_2 \in \Zc:  \\  (q_0, q_1, q_2) \neq  0,  \\  |q_k| \leq Q_k  \ \forall k }}  B(q_0, q_1, q_2) .   \label{eq:Boutage22}
\end{align}
From  Lemma~\ref{lm:NMb},  the Lebesgue measure of $B$, denoted by $\Lc (B)$,  is bounded by
\begin{align}
\Lc (B ) & \leq   504 \beta \Bigl(  2 \min \bigl\{ 2 Q_{\max},  \frac{ 2 Q_{\max}}{ P^{1-  \alpha }} \bigr\}  +  \tilde{Q}_2 \cdot \min \bigl\{ 4 Q_{\max} ,  \frac{2 Q_{\max} }{P^{\frac{1-  \alpha}{2} }},  \frac{P^{1-  \alpha } \tilde{Q}_2}{P^{\frac{1-  \alpha}{2} }}   \bigr\} \non\\ & \quad\quad\quad + 2 \min \bigl\{ 4 Q_{\max},  \frac{ 2 Q_{\max}}{ P^{\frac{1-  \alpha}{2} }} \bigr\}  + \tilde{Q}_1 \cdot \min \bigl\{ 2 Q_{\max} ,  \frac{2 Q_{\max} }{P^{1-  \alpha }},  \frac{P^{\frac{1-  \alpha}{2} }  \cdot 4 Q_{\max}}{P^{1-  \alpha }}   \bigr\}        \Bigr)   \non\\
 &= 504 \beta \Bigl(  \frac{ 4 Q_{\max}}{ P^{1-  \alpha }}  +  \tilde{Q}_2 \cdot   \frac{2 Q_{\max} }{P^{\frac{1-  \alpha}{2} }}  + \frac{ 4 Q_{\max}}{ P^{\frac{1-  \alpha}{2} }}   + \tilde{Q}_1 \cdot \frac{2 Q_{\max} }{P^{1-  \alpha }}       \Bigr)  \non \\
  &= 504 \beta \Bigl(  \frac{ 4 Q_{\max}}{ P^{1-  \alpha }}  + 2 Q_{\max} \cdot \min\{1, 16 P^{-\frac{1-  \alpha}{2} }\}  \cdot   \frac{2 Q_{\max} }{P^{\frac{1-  \alpha}{2} }}  + \frac{ 4 Q_{\max}}{ P^{\frac{1-  \alpha}{2} }}   + 4 Q_{\max} \cdot \frac{2 Q_{\max} }{P^{1-  \alpha }}       \Bigr)  \non\\
   &\leq  504 \beta \Bigl(  \frac{ 4 Q_{\max}}{ P^{1-  \alpha }}  + 2 Q_{\max} \cdot  16 \cdot   \frac{2 Q_{\max} }{P^{1-  \alpha }}  + \frac{ 4 Q_{\max}}{ P^{\frac{1-  \alpha}{2} }}   + 4 Q_{\max} \cdot \frac{2 Q_{\max} }{P^{1-  \alpha }}       \Bigr)   \non\\
   & \leq  504 \beta \cdot    \frac{ 16Q_{\max}}{ P^{\frac{1-  \alpha}{2} }}  \cdot \max \{  16  \frac{Q_{\max}}{ P^{\frac{1-  \alpha}{2} }} , 1 \}  \non\\ 
      & =  504 \beta \cdot    16 P^{ \frac{ 4\alpha/3 -1 - \epsilon  }{2}} \cdot \max \{  16 P^{ \frac{ 4\alpha/3 -1 - \epsilon  }{2}} , 1 \}  \non\\
            & \leq  504 \beta \cdot    16 P^{ \frac{ 4\alpha/3 -1 - \epsilon  }{2}} \cdot   16 P^{ \frac{ 4\alpha/3 -1   }{2}}   \non\\ 
             & =  504 \delta P^{- \frac{8\alpha/3 -2 }{2}}  \cdot    16 P^{ \frac{ 4\alpha/3 -1 - \epsilon  }{2}} \cdot   16 P^{ \frac{ 4\alpha/3 -1   }{2}}   \non\\   
                & =  129024 \delta   \cdot     P^{ - \frac{ \epsilon  }{2}}          \label{eq:LeBmeasure}
\end{align}
for a constant $\delta \in (0, 1]$, where
\begin{align*}
\tilde{Q}_1 &= \min\Bigl\{4 Q_{\max}, \   8 \cdot \frac{\max\{2 Q_{\max}, \  P^{1-  \alpha }\cdot 2 Q_{\max} \}}{P^{\frac{1-  \alpha}{2} }}\Bigr\} = \min\Bigl\{4 Q_{\max},   \  8 P^{\frac{1-  \alpha}{2} }\cdot 2 Q_{\max} \Bigr\} =4 Q_{\max}\\
\tilde{Q}_2 &\defeq \min\Bigl\{2 Q_{\max}, \   8\cdot \frac{\max\{2 Q_{\max},  \  P^{\frac{1-  \alpha}{2} } \cdot4 Q_{\max}\}}{ P^{1-  \alpha }}\Bigr\}  = 2 Q_{\max} \cdot \min\{1,  \ 16 P^{-\frac{1-  \alpha}{2} }\}.
\end{align*}
The set  $B$ defined in \eqref{eq:Boutage22} is the collection of  $(g_0, g_1, g_2) \in (1,4]^3$,  and  for any $(g_0, g_1, g_2) \in B$ there exists at least one triple  $(q_0, q_1, q_2) \in \{ q_0, q_1, q_2:  q_0, q_1, q_2 \in \Zc, (q_0, q_1, q_2) \neq  0,  |q_k| \leq Q_k  \ \forall k \}$, such that $|  g_0q_0 + A_1 g_1 q_1 + A_2 g_2 q_2    | <  \delta P^{- \frac{8\alpha/3 -2 }{2}}$.
Therefore,  $B$ can be considered as an outage set. 
For any  triple $(g_0, g_1, g_2)$  outside the outage set $B$, i.e., $(g_0, g_1, g_2)\notin B$, it is apparent that $d_{\min} (g_0, g_1, g_2)   \geq   \delta P^{- \frac{8\alpha/3 -2 }{2}}$. 

In our setting $g_0\defeq  h_{12} h_{21},      g_1\defeq   h_{12} h_{11},    g_2 \defeq h_{11}h_{22}$. 
Let us define  $\Ho$  as the collection of the quadruples $(h_{11}, h_{12}, h_{22}, h_{21} ) \in (1, 2]^{2\times 2}$   such that the corresponding triples $(g_0, g_1, g_2)$ are in the outage set $B$, that is,
\[ \Ho \defeq \{  (h_{11}, h_{12}, h_{22}, h_{21} ) \in (1, 2]^{2\times 2} :      (g_0, g_1, g_2) \in B  \} . \]
In the second step, the goal is to bound the Lebesgue measure of $\Ho$: 
\begin{align}
\Lc (\Ho ) & =  \int_{h_{11}=1}^2  \int_{h_{12}=1}^2 \int_{h_{21}=1}^2 \int_{h_{22}=1}^2    \mathbbm{1}_{\Ho}  ( h_{11}, h_{12}, h_{22}, h_{21}) d h_{22} d h_{21} d h_{12}  d h_{11}           \label{eq:LeBmeasure882441}\\  
& =  \int_{h_{11}=1}^2  \int_{h_{12}=1}^2 \int_{h_{21}=1}^2 \int_{h_{22}=1}^2    \mathbbm{1}_{B}  (h_{12} h_{21}, h_{12} h_{11}, h_{11}h_{22}) d h_{22} d h_{21} d h_{12}  d h_{11}          \non\\ 
& \leq    \int_{h_{11}=1}^2  \int_{g_{2}=1}^{4} \int_{g_{1} =1}^4 \int_{g_{0}=1}^4    \mathbbm{1}_{B}  (g_0, g_{1}, g_{2})  g_1^{-1} h_{11}^{-1} d g_{0} d g_{1} d g_{2}  d h_{11}          \non\\ 
& \leq    \int_{h_{11}=1}^2  \int_{g_{2}=1}^{4} \int_{g_{1} =1}^4 \int_{g_{0}=1}^4    \mathbbm{1}_{B}  (g_0, g_{1}, g_{2})  d g_{0} d g_{1} d g_{2}  d h_{11}          \non\\ 
& =    \int_{h_{11}=1}^2    \mathcal{L}(B)  d h_{11}          \non\\ 
& \leq     \int_{h_{11}=1}^2     129024 \delta   \cdot     P^{ - \frac{ \epsilon  }{2}}    d h_{11}           \label{eq:LeBmeasure0993}  \\ 
& =     258048 \delta   \cdot     P^{ - \frac{ \epsilon  }{2}}              \label{eq:LeBmeasure9999}  
\end{align}
where  $ \mathbbm{1}_{\Ho} (h_{11}, h_{12}, h_{22}, h_{21} )= 1 $ if  $(h_{11}, h_{12}, h_{22}, h_{21} ) \in  \Ho$, else   $\mathbbm{1}_{\Ho} (h_{11}, h_{12}, h_{22}, h_{21} )= 0$;  similarly, $ \mathbbm{1}_{B}  (g_0, g_1, g_2) =1$ if $(g_0, g_1, g_2) \in  B$, else $ \mathbbm{1}_{B}  (g_0, g_1, g_2) =0$;
and \eqref{eq:LeBmeasure0993} follows from the result in \eqref{eq:LeBmeasure}. Note that the approach on bounding $\Lc (\Ho )$ in the second step has been previously used in \cite{NM:13}.
At this point we complete the proof of this lemma.
 \end{proof}

 Lemma~\ref{lm:distance3432} shows that, when  $3/4 \leq \alpha \leq 1$, and given the signal design in \eqref{eq:constellationGsym1}-\eqref{eq:constellationGsym1u}
and \eqref{eq:para34111}-\eqref{eq:xvkkk1341},  the minimum distance $d_{\min}$ defined in \eqref{eq:minidis111} is sufficiently large
for all  the channel  coefficients $\{h_{k\ell}\} \in (1, 2]^{2\times 2} $ except for an outage set   $\Ho \subseteq (1,2]^{2\times 2}$ with Lebesgue measure  $\mathcal{L}(\Ho)$ satisfying  \[ \mathcal{L}(\Ho)  \to 0, \quad \text {as}\quad  P\to \infty \] 
(see \eqref{eq:LeBmeasure9999}).
In what follows, we will consider the channel  coefficients $\{h_{k\ell}\} \in (1, 2]^{2\times 2} $ that are not in the an outage set   $\Ho$. Under this channel condition, the minimum distance $d_{\min}$ defined in \eqref{eq:minidis111}  is bounded by  $d_{\min}    \geq   \delta P^{- \frac{8\alpha/3 -2 }{2}} $   (see \eqref{eq:distancegeq} in Lemma~\ref{lm:distance3432}) for a given constant $\delta \in (0, 1]$.

At this point, we go back to the expression of $y_1$ in \eqref{eq:yvk4835} and decode the sum \[x_{s}  \defeq  g_0 q_0 + \sqrt{P^{1-  \alpha }} g_1 q_1 + \sqrt{P^{2- 2\alpha}} g_2 q_2 \] by treating other signals as noise (noise removal).  Once  $x_{s} $ is decoded correctly, then $q_0, q_1, q_2$ can be recovered due to the fact that $g_0, g_1, g_2$ are rationally independent (signal separation, cf.~\cite{MGMK:14}).
Note that $y_1$ expressed in \eqref{eq:yvk4835} can also be rewritten as
\begin{align}
 y_1    &=    \sqrt{P^{ 5\alpha/3 - 1  + \epsilon}} \cdot 2\gamma  \underbrace{( g_0 q_0 + \sqrt{P^{1-  \alpha }} g_1 q_1 + \sqrt{P^{2- 2\alpha}} g_2 q_2 )}_{ \defeq x_{s}  }   +  \sqrt{P^{ 1 - \alpha}}  \bigl(\underbrace{h_{11}h_{22} v_{1,p}  +    \frac{1}{\sqrt{P^{ 1 - \alpha}}}h_{12} h_{11}  v_{2,p}\bigr)}_{\defeq  \tilde{g} }   +  z_{1}   \non\\
 &=    \sqrt{P^{ 5\alpha/3 - 1  + \epsilon}} \cdot 2\gamma \cdot x_{s}   +  \sqrt{P^{ 1 - \alpha}} \tilde{g}  +  z_{1}     \label{eq:yk182375}  
\end{align}
where $\tilde{g} \defeq h_{11}h_{22} v_{1,p}  +    \frac{1}{\sqrt{P^{ 1 - \alpha}}}h_{12} h_{11}  v_{2,p}$ and 
\[ |\tilde{g} | \leq  \tilde{g}_{\max} \defeq \sqrt{2}  \quad  \forall   \tilde{g} \]
for this case of  $3/4 \leq \alpha \leq 1$, and given the signal design in \eqref{eq:constellationGsym1}-\eqref{eq:constellationGsym1u}
and \eqref{eq:para34111}-\eqref{eq:xvkkk1341}.   Based on our definition, the minimum distance for $x_{s}$ is   $d_{\min}$ defined in \eqref{eq:minidis111}.  Lemma~\ref{lm:distance3432} reveals that, under the channel condition $\{h_{k\ell}\} \notin \Ho$, the  minimum distance for $x_{s}$ is bounded by  $d_{\min}    \geq   \delta P^{- \frac{8\alpha/3 -2 }{2}} $.
For this channel model in \eqref{eq:yk182375}, the probability of error for decoding $x_{s}$ from $y_1$ is 
 \begin{align}
 &\quad \  \text{Pr} [ x_s \neq \hat{x}_s ]     \non\\
  &\leq     \text{Pr} \Bigl[   | z_1  + P^{ \frac{1 - \alpha}{2}} \tilde{g}|  >    P^{ \frac{5\alpha/3 - 1  + \epsilon}{2}} \cdot 2\gamma   \cdot \frac{d_{\text{min}} }{2}  \Bigr]   \non \\
    &=   \text{Pr} [   z_1   >       P^{ \frac{5\alpha/3 - 1  + \epsilon}{2}} \cdot 2\gamma   \cdot \frac{d_{\text{min}} }{2} -  P^{ \frac{1 - \alpha}{2}} \tilde{g} ]  +  \text{Pr} [     z_1   >    P^{ \frac{5\alpha/3 - 1  + \epsilon}{2}} \cdot 2\gamma   \cdot \frac{d_{\text{min}} }{2}  +  P^{ \frac{1 - \alpha}{2}} \tilde{g}    ]     \label{eq:error2244cc}  \\
  & \leq       \text{Pr} [  z_1   >  P^{ \frac{5\alpha/3 - 1  + \epsilon}{2}} \cdot 2\gamma   \cdot \frac{d_{\text{min}} }{2}   -  P^{ \frac{1 - \alpha}{2}} \tilde{g}_{\max}   ]  +  \text{Pr} [  z_1   >  P^{ \frac{5\alpha/3 - 1  + \epsilon}{2}} \cdot 2\gamma   \cdot \frac{d_{\text{min}} }{2}   -  P^{ \frac{1 - \alpha}{2}} \tilde{g}_{\max}   ]  \label{eq:error2256cc} \\
    & =   2  \cdot     {\bf{Q}} \bigl(  P^{ \frac{5\alpha/3 - 1  + \epsilon}{2}} \cdot 2\gamma   \cdot \frac{d_{\text{min}} }{2}   -  P^{ \frac{1 - \alpha}{2}} \tilde{g}_{\max}  \bigr)    \non   \\ 
    & \leq    2  \cdot     {\bf{Q}} \bigl(  P^{ \frac{1 - \alpha}{2}} ( \gamma \delta  P^{ \frac{ \epsilon}{2}}   -  \sqrt{2})  \bigr)     \label{eq:error9982cc} 
  \end{align}
  where $\hat{x}_s$ is the estimate for $ x_s$ by choosing the point close to $ x_s$, based on the observation $y$;
  $d_{\min}$ is defined in \eqref{eq:minidis111}; 
  $\delta \in (0, 1]$ and $\gamma  \in \bigl(0, \frac{1}{4\sqrt{2}}\bigr]$ are two constants;
  \eqref{eq:error2244cc} follows from the fact that  $1- {\bf{Q}} (a)  = {\bf{Q}} (- a )$ for any $a \in \Rc$,  where the ${\bf{Q}}$-function is defined as ${\bf{Q}}(a )  \defeq  \frac{1}{\sqrt{2\pi}} \int_{a}^{\infty}  \exp( -\frac{ s^2}{2} ) d s$;
\eqref{eq:error2256cc} uses the fact that  $ |\tilde{g} | \leq  \tilde{g}_{\max} \defeq \sqrt{2},    \forall   \tilde{g}$;
\eqref{eq:error9982cc} follows from the result that $d_{\min}    \geq   \delta P^{- \frac{8\alpha/3 -2 }{2}} $ (see Lemma~\ref{lm:distance3432}).
When $   \gamma \delta  P^{ \frac{ \epsilon}{2}}   -  \sqrt{2} \geq 0 $,  the error probability $\text{Pr} [ x_s \neq \hat{x}_s ]$ can be further bounded as
  \begin{align}
  \text{Pr} [ x_s \neq \hat{x}_s ]   &\leq    \exp \Bigl(    -  \frac{ P^{1 - \alpha}  \bigl(  \gamma \delta  P^{ \frac{ \epsilon}{2}}   -  \sqrt{2} \bigr)^2 }{2} \Bigr)  \non
   \end{align}
 by using the identity that  $ {\bf{Q}} (a ) \leq   \frac{1}{2}\exp ( -  a^2 /2 ),  \    \forall a \geq 0.$ 
At this point, we can conclude that  the error probability for decoding $x_{s}$ from $y_1$ is 
 \begin{align}
 \text{Pr} [ x_s \neq \hat{x}_s ] \to 0  \quad  \text{as} \quad   P\to \infty.    \label{eq:error885256}                           
 \end{align}
Once  $x_{s} $ is decoded correctly, then the three symbols $q_0  =  \frac{Q_{\max}}{2\gamma} \cdot   u_{2}$,  $q_1  =  \frac{Q_{\max}}{2\gamma} \cdot   (   v_{2,c}  +  u_1)$ and  $ q_2  =  \frac{ Q_{\max}}{2\gamma} \cdot  v_{1,c}$ can be recovered from $x_{s}  =  g_0 q_0 + \sqrt{P^{1-  \alpha }} g_1 q_1 + \sqrt{P^{2- 2\alpha}} g_2 q_2 $ due to the fact that $g_0, g_1, g_2$ are rationally independent.

In the next step, we remove the decoded $x_{s}$ from $y_1$  (see \eqref{eq:yk182375}) and then decode $v_{1,p}$ from the following observation
\begin{align}
 y_1  -  \sqrt{P^{ 5\alpha/3 - 1  + \epsilon}} \cdot 2\gamma \cdot x_{s} &=  \sqrt{P^{ 1 - \alpha}}  h_{11}h_{22} v_{1,p}  + h_{12} h_{11}  v_{2,p}    +  z_{1} .    \non
\end{align}
Note that $v_{1,p}, v_{2,p} \in    \Omega (\xi   =\gamma \cdot \frac{ 1}{Q},   \   Q = P^{ \frac{ 1 - \alpha - \epsilon}{2}} )$  and $h_{12} h_{11}  v_{2,p} \leq  \frac{1}{\sqrt{2}}$. Then, from Lemma~\ref{lm:AWGNic} (see Section~\ref{sec:CJGau}) we conclude that the error probability for decoding $v_{1,p}$  is 
 \begin{align}
  \text{Pr} [ v_{1,p} \neq \hat{v}_{1,p} ]     \to 0  \quad  \text{as} \quad   P\to \infty . \label{eq:error155525}                           
 \end{align}
 
By combining the results in \eqref{eq:error885256} and \eqref{eq:error155525}, as well as the fact that $v_{1,c}$ can be recovered from $x_s$, we conclude that  the error probability of estimating  $\{v_{k,c}, v_{k,p} \}$ from $y_k$  is
 \begin{align}
 \text{Pr} [  \{ v_{k,c} \neq \hat{v}_{k,c} \} \cup  \{ v_{k,p} \neq \hat{v}_{k,p} \}  ]  \to 0         \quad \text {as}\quad  P\to \infty    \label{eq:error1c1p0595}
 \end{align}
for $k=1$ and $3/4 \leq \alpha \leq 1$,  given the signal design in \eqref{eq:constellationGsym1}-\eqref{eq:constellationGsym1u} and \eqref{eq:para34111}-\eqref{eq:xvkkk1341}.  This result holds true  for  all the channel coefficients  $\{h_{k\ell}\} \in (1, 2]^{2\times 2}$  except  for an outage set   $\Ho \subseteq (1,2]^{2\times 2}$ with Lebesgue measure  $\mathcal{L}(\Ho)$ satisfying  \[ \mathcal{L}(\Ho)  \to 0, \quad \text {as}\quad  P\to \infty \] 
(see \eqref{eq:LeBmeasure9999}).
Due to the symmetry, the result in \eqref{eq:error1c1p0595} also holds true for the case of $k=2$. At this point, we complete the proof.

\section{Converse \label{sec:converse}}

For the Gaussian interference channel defined in Section~\ref{sec:systemGaussian}, we  provide an outer bound on  the secure   capacity region, which is stated in the following  lemma.

\begin{lemma}  \label{lm:gupper}
For the two-user Gaussian  interference channel defined in Section~\ref{sec:systemGaussian},  the  secure  capacity region is  bounded by
\begin{align}
 R_1 + R_2  & \leq  
 \frac{1}{2} \log \Bigl(1+ P^{\alpha_{22}-\alpha_{12}} \cdot \frac{|h_{22}|^2}{|h_{12}|^2} +  P^{\alpha_{21} -(\alpha_{11} - \alpha_{12} )^+} \cdot \frac{|h_{21}|^2}{|h_{11}|^2}  \Bigr)      \non \\
  & \quad +  \frac{1}{2} \log \Bigl(1+ P^{\alpha_{11}-\alpha_{21}} \cdot \frac{|h_{11}|^2}{|h_{21}|^2} +  P^{\alpha_{12} -(\alpha_{22} - \alpha_{21} )^+} \cdot \frac{|h_{12}|^2}{|h_{22}|^2}  \Bigr) +   \log 10  \label{eq:gaussianupbound}   \\
 2R_1 + R_2  & \leq   \frac{1}{2} \log \bigl(1\!+\!   \frac{P^{ (\alpha_{11} - \alpha_{21} )^+}}{|h_{21}|^2}   \bigr)    +   \frac{1}{2} \log \bigl(1\!+\!   \frac{P^{ (\alpha_{22} - \alpha_{12} )^+}}{|h_{12}|^2}   \bigr)   + \frac{1}{2} \log \bigl(1\!+\! P^{\alpha_{11}}   |h_{11}|^2  \!+\! P^{\alpha_{12}}   |h_{12}|^2  \bigr)        \!+\! \log 9    \label{eq:detup4} \\ 
2R_2 + R_1  & \leq    \frac{1}{2} \log \bigl(1\!+\!   \frac{P^{ (\alpha_{22} - \alpha_{12} )^+}}{|h_{12}|^2}   \bigr)    +   \frac{1}{2} \log \bigl(1\!+\!   \frac{P^{ (\alpha_{11} - \alpha_{21} )^+}}{|h_{12}|^2}   \bigr)   + \frac{1}{2} \log \bigl(1\!+\! P^{\alpha_{22}}   |h_{22}|^2  \!+\! P^{\alpha_{21}}   |h_{21}|^2  \bigr)        \!+\!  \log 9    \label{eq:detup5}   \\
 R_1  & \leq      \frac{1}{2} \log \Bigl(1+ P^{\alpha_{11}-\alpha_{21}} \cdot \frac{|h_{11}|^2}{|h_{21}|^2} +  P^{\alpha_{22} + \alpha_{11} - \alpha_{21} } \cdot \frac{|h_{11}|^2 |h_{22}|^2}{|h_{21}|^2}  \Bigr)     \label{eq:detup1} \\ 
R_2 & \leq      \frac{1}{2} \log \Bigl(1+ P^{\alpha_{22}-\alpha_{12}} \cdot \frac{|h_{22}|^2}{|h_{12}|^2} +  P^{\alpha_{11} + \alpha_{22} - \alpha_{12} } \cdot \frac{|h_{22}|^2 |h_{11}|^2}{|h_{12}|^2}  \Bigr)        \label{eq:detup2} \\ 
  R_1  & \leq      \frac{1}{2} \log \bigl(1+ P^{\alpha_{11}} \cdot |h_{11}|^2 \bigr)     \label{eq:detup11} \\ 
    R_2  & \leq      \frac{1}{2} \log \bigl(1+ P^{\alpha_{22}} \cdot |h_{22}|^2 \bigr) .    \label{eq:detup22} 
\end{align}
\end{lemma}

 \vspace{5pt}
Before providing the proof of  Lemma~\ref{lm:gupper}, let us provide a result on the secure sum GDoF derived from Lemma~\ref{lm:gupper}. 

\begin{corollary} [Secure sum GDoF] \label{cor:symSGDoF}
For the two-user   Gaussian  interference channel,   the secure sum GDoF is upper  bounded by 
\begin{align}
d_{\text{sum}}   & \leq      ( \alpha_{11}  -  ( \alpha_{21} - \alpha_{22} )^+ )^+   + ( \alpha_{22}  -  ( \alpha_{12} - \alpha_{11} )^+ )^+  \label{eq:Gup12}  \\ 
d_{\text{sum}}   & \leq    \max\{  \alpha_{21}- (\alpha_{11 } - \alpha_{12})^+ ,  \  \alpha_{22}- \alpha_{12}, \ 0  \}      + \max\{  \alpha_{12}- (\alpha_{22 } - \alpha_{21})^+ ,  \  \alpha_{11}- \alpha_{21},  \ 0 \}     \label{eq:Gup3} \\ 
d_{\text{sum}}   & \leq \frac{1}{3}  \Bigl( \max\{ \alpha_{11}, \alpha_{12} \}  +  (  \alpha_{11}  - \alpha_{21})^+   +    (  \alpha_{22}  - \alpha_{12})^+   \non \\ 
    & \quad \  +    \max\{ \alpha_{22}, \alpha_{21} \}  +  (  \alpha_{11}  - \alpha_{21})^+   +    (  \alpha_{22}  - \alpha_{12})^+  \Bigr).   \label{eq:Gup45}  
\end{align}
\end{corollary}
\begin{proof}
The proof is based on the result of Lemma~\ref{lm:gupper}, by focusing on the secure sum GDoF measurement. 
Specifically, \eqref{eq:Gup12} results from \eqref{eq:detup1}-\eqref{eq:detup22}; \eqref{eq:Gup3} stems from \eqref{eq:gaussianupbound}; while \eqref{eq:Gup45} follows from the combination of \eqref{eq:detup4} and \eqref{eq:detup5}.
\end{proof}

Note that, for the symmetric case with $\alpha_{11}= \alpha_{22}=1$ and $\alpha_{21}= \alpha_{12}=\alpha$,  the secure sum GDoF bound described in Corollary~\ref{cor:symSGDoF} is simplified as 
\[ d_{\text{sum}} \leq 2\cdot  \min \Bigl\{ \max\{  \alpha- (1 - \alpha)^+ ,  \  1- \alpha  \}  , \   (1 -  ( \alpha - 1 )^+ )^+, \    \frac{\max\{ 1, \alpha \}  +  2(  1  - \alpha)^+   }{3}  \Bigr\} \]
which serves as the converse of Theorem~\ref{thm:GDoF}. 
In the following we will provide the proof of Lemma~\ref{lm:gupper}.

\subsection{Proof of  bound \eqref{eq:gaussianupbound}  \label{sec:conBgaussianupbound}}

Let us now prove  bound \eqref{eq:gaussianupbound} in Lemma~\ref{lm:gupper}.  
At first we define that 
\begin{align}
  s_{kk}(t)  & \defeq   \sqrt{P^{(\alpha_{kk}-\alpha_{k\ell})^+}} h_{kk} x_{k}(t) +  \tilde{z}_{k}(t)    \label{eq:defs11} \\
 s_{\ell k}(t)  & \defeq   \sqrt{P^{\alpha_{\ell k}}} h_{\ell k} x_{k}(t)  +z_{\ell}(t),  \label{eq:defs12}
\end{align}
 for $k,\ell   \in \{1,2\}, k \neq \ell$,  where  $\tilde{z}_{k}(t) \sim \mathcal{N}(0, 1)$ is a virtual noise that is independent of the other noise and transmitted signals.  Let $s^{\bln}_{kk} \defeq \{ s_{kk}(t)\}_{t=1}^{\bln}$ and $s^{\bln}_{\ell k} \defeq \{ s_{\ell k}(t)\}_{t=1}^{\bln}$.
We begin with the rate of user~1:
\begin{align}
  \bln R_1 &= \Hen(w_1)    \nonumber \\
  &= \Imu(w_1; y^{\bln}_1) + \Hen(w_1| y^{\bln}_1)  \nonumber \\
  &\leq  \Imu(w_1; y^{\bln}_1) + \bln \epsilon_{1,n}    \label{eq:Fano}  \\
  &\leq  \Imu(w_1;  y^{\bln}_1) - \Imu(w_1; y^{\bln}_2) +  \bln \epsilon_{1,n} + \bln \epsilon   \label{eq:secrecy} 
\end{align}
where \eqref{eq:Fano} follows from  Fano's inequality,  $\lim_{n\to\infty} \epsilon_{1,n} = 0$;
\eqref{eq:secrecy} results from the secrecy constraint, i.e.,  $ \Imu(w_1; y^{\bln}_2)  \leq \bln \epsilon$ for an  arbitrary small  $\epsilon$. Similarly, for the rate of user~2 we have 
\begin{align}
  \bln R_2 & \leq  \Imu(w_2; y^{\bln}_2) - \Imu(w_2 ; y^{\bln}_1) +  \bln \epsilon_{2,n}  + \bln \epsilon    \label{eq:secrecy2} 
\end{align}
which, together with \eqref{eq:secrecy}, gives the following bound on the sum rate:
\begin{align}
 & \bln R_1 + \bln R_2 -\bln \epsilon_{1,n} -\bln \epsilon_{2,n} -2 \bln \epsilon     \nonumber \\
   \leq &\Imu(w_1; y^{\bln}_1) - \Imu(w_1;  y^{\bln}_2)  +  \Imu(w_2 ;  y^{\bln}_2) - \Imu(w_2 ;  y^{\bln}_1)     \nonumber \\
   = & \hen(y^{\bln}_1) - \hen(y^{\bln}_1| w_1)  -    \hen(y^{\bln}_2) +  \hen(y^{\bln}_2| w_1)   + \hen(y^{\bln}_2) - \hen(y^{\bln}_2| w_2)  -    \hen(y^{\bln}_1) +  \hen(y^{\bln}_1| w_2)       \nonumber \\
   = & \hen(y^{\bln}_2| w_1)  - \hen(y^{\bln}_1| w_1)   +  \hen(y^{\bln}_1| w_2)    - \hen(y^{\bln}_2| w_2).     \label{eq:sumrate1} 
\end{align}
To bound the right-hand side of \eqref{eq:sumrate1}, we provide the following lemma. 

 \vspace{10pt}
\begin{lemma}  \label{lm:differencebG}   
For $y_{1}(t)$ and $y_{2}(t)$ expressed in \eqref{eq:channelG}, we have
\begin{align} 
   \hen(y^{\bln}_2| w_1)  - \hen(y^{\bln}_1| w_1)   \leq &  \frac{n}{2} \log \Bigl(1+ P^{\alpha_{22}-\alpha_{12}} \cdot \frac{|h_{22}|^2}{|h_{12}|^2} +  P^{\alpha_{21} -(\alpha_{11} - \alpha_{12} )^+} \cdot \frac{|h_{21}|^2}{|h_{11}|^2}  \Bigr) +    \frac{n}{2} \log 10   \label{eq:differenceb11G}  \\  
 \hen(y^{\bln}_1| w_2)  - \hen(y^{\bln}_2| w_2)  \leq  &   \frac{n}{2} \log \Bigl(1+ P^{\alpha_{11}-\alpha_{21}} \cdot \frac{|h_{11}|^2}{|h_{21}|^2} +  P^{\alpha_{12} -(\alpha_{22} - \alpha_{21} )^+} \cdot \frac{|h_{12}|^2}{|h_{22}|^2}  \Bigr) +    \frac{n}{2} \log 10. \label{eq:differenceb22G}      
\end{align}
\end{lemma}
\begin{proof}
Let us first prove the bound in \eqref{eq:differenceb11G}:
\begin{align}
  &  \hen(y^{\bln}_2| w_1)  - \hen(y^{\bln}_1| w_1)       \nonumber \\
  =&  \hen(s_{11}^{\bln}, y^{\bln}_2| w_1)  -   \underbrace{\hen(s_{11}^{\bln}| y^{\bln}_2, w_1)}_{\defeq J_1}    - \hen(s_{11}^{\bln}, y^{\bln}_1| w_1)    + \underbrace{\hen(s_{11}^{\bln} | y^{\bln}_1, w_1) }_{\defeq J_2}  \label{eq:bd3141}   \\ 
 = &   \hen( y^{\bln}_2| s_{11}^{\bln}, w_1) - \hen( y^{\bln}_1|s_{11}^{\bln},  w_1)   - J_1  + J_2      \nonumber   \\
= &  \hen(s_{12}^{\bln},  y^{\bln}_2| s_{11}^{\bln}, w_1) -  \underbrace{\hen(s_{12}^{\bln} | y^{\bln}_2,  s_{11}^{\bln}, w_1)}_{\defeq J_3}    - \hen( y^{\bln}_1|s_{11}^{\bln},  w_1)    - J_1  + J_2 \non \\
 =&    \hen(s_{12}^{\bln}| s_{11}^{\bln}, w_1)  +   \underbrace{\hen( y^{\bln}_2| s_{12}^{\bln}, s_{11}^{\bln}, w_1)}_{\defeq J_4}   - \hen( y^{\bln}_1|s_{11}^{\bln},  w_1)  - J_1  + J_2 -J_3 \label{eq:bd8225}      \\
  =&   \hen(y^{\bln}_1| x_1^{\bln}, s_{11}^{\bln}, w_1)     -  \hen( y^{\bln}_1|s_{11}^{\bln},  w_1)    -  J_1   +  J_2  - J_3  + J_4 \label{eq:bd2535}      \\ 
  =&   -  \Imu (y^{\bln}_1;  x_1^{\bln}| s_{11}^{\bln}, w_1)  - J_1  + J_2 -J_3 +J_4 \label{eq:bd992211}  
 \end{align}
where  $J_1\defeq \hen(s_{11}^{\bln}| y^{\bln}_2, w_1)$, $ J_2 \defeq \hen(s_{11}^{\bln} | y^{\bln}_1, w_1) $, $J_3 \defeq \hen(s_{12}^{\bln} | y^{\bln}_2, s_{11}^{\bln}, w_1)$ and $J_4\defeq \hen( y^{\bln}_2| s_{12}^{\bln}, s_{11}^{\bln}, w_1)$;
  the steps from \eqref{eq:bd3141} to \eqref{eq:bd8225}   follow from   chain rule;
  \eqref{eq:bd2535}  stems from that 
\begin{align}
   & \hen(s_{12}^{\bln}| s_{11}^{\bln}, w_1)      \non\\
 = &  \hen(s_{12}^{\bln})   \label{eq:bd8325}  \\
 = &  \hen(s_{12}^{\bln}| x^{\bln}_1, s_{11}^{\bln}, w_1)   \label{eq:bd1794}      \\
 = &  \hen( \{ \sqrt{P^{\alpha_{12}}} h_{12} x_{2}(t)  +z_{1}(t)  \}_{t=1}^{\bln} | x^{\bln}_1, s_{11}^{\bln}, w_1)  \nonumber   \\
 = &  \hen( \{    \sqrt{P^{\alpha_{11}}} h_{11} x_{1}(t)  + \sqrt{P^{\alpha_{12}}} h_{12} x_{2}(t)    +z_{1}(t)  \}_{t=1}^{\bln} | x^{\bln}_1, s_{11}^{\bln}, w_1)   \label{eq:bd7420}       \\
 = &  \hen(y^{\bln}_1| x_1^{\bln}, s_{11}^{\bln}, w_1)    \nonumber
 \end{align}
where \eqref{eq:bd8325} and \eqref{eq:bd1794} use the fact that $s_{12}^{\bln}$ is independent of $s_{11}^{\bln}, w_1$ and $ x^{\bln}_1$;
\eqref{eq:bd7420} follows from that $\hen(a|b)=\hen(a+b |b)$ for any continuous random variables $a$ and $b$.    
Going back to \eqref{eq:bd992211}, we further have 
\begin{align}
   &\hen(y^{\bln}_2| w_1)  - \hen(y^{\bln}_1| w_1)   \non\\
   & =  -  \Imu (y^{\bln}_1;  x_1^{\bln}| s_{11}^{\bln}, w_1)  - J_1  + J_2 -J_3 +J_4     \label{eq:bd887733} \\
   & \leq  - J_1  + J_2 -J_3 +J_4     \label{eq:bd33405} \\
   &\leq   \frac{n}{2} \log \Bigl(1+ P^{\alpha_{22}-\alpha_{12}} \cdot \frac{|h_{22}|^2}{|h_{12}|^2} +  P^{\alpha_{21} -(\alpha_{11} - \alpha_{12} )^+} \cdot \frac{|h_{21}|^2}{|h_{11}|^2}  \Bigr) +    \frac{n}{2} \log 10  \label{eq:bd92374}  
 \end{align}
 where \eqref{eq:bd887733} is from \eqref{eq:bd992211};  
 \eqref{eq:bd33405}   stems from  the nonnegativity of mutual information; 
 \eqref{eq:bd92374} follows from Lemma~\ref{lm:Jbounds} (see below).
Similarly, by interchanging the roles of user~1 and user~2, we also have 
\begin{align}
  & \hen(y^{\bln}_1| w_2)  - \hen(y^{\bln}_2| w_2) 
  \leq     \frac{n}{2} \log \Bigl(1+ P^{\alpha_{11}-\alpha_{21}} \cdot \frac{|h_{11}|^2}{|h_{21}|^2} +  P^{\alpha_{12} -(\alpha_{22} - \alpha_{21} )^+} \cdot \frac{|h_{12}|^2}{|h_{22}|^2}  \Bigr) +    \frac{n}{2} \log 10   \label{eq:bd9256}.     
 \end{align}
 \end{proof}
 
 \vspace{10pt}
\begin{lemma}  \label{lm:Jbounds}
For $J_1 = \hen(s_{11}^{\bln}| y^{\bln}_2, w_1)$, $ J_2 = \hen(s_{11}^{\bln} | y^{\bln}_1, w_1) $, $J_3 = \hen(s_{12}^{\bln} | y^{\bln}_2, s_{11}^{\bln}, w_1)$ and $J_4= \hen( y^{\bln}_2| s_{12}^{\bln}, s_{11}^{\bln}, w_1)$,  we have
\begin{align}
J_1  & \geq   \frac{n}{2}\log (2\pi e)   \label{eq:Jb1} \\
J_3  & \geq   \frac{n}{2}\log (2\pi e)   \label{eq:Jb3}  \\
J_2  & \leq   \frac{n}{2} \log (20\pi e )    \label{eq:Jb2}  \\
J_4  & \leq  \frac{n}{2} \log \Bigl(2\pi e\Bigl(1+ P^{\alpha_{22}-\alpha_{12}} \cdot \frac{|h_{22}|^2}{|h_{12}|^2} +  P^{\alpha_{21} -(\alpha_{11} - \alpha_{12} )^+} \cdot \frac{|h_{21}|^2}{|h_{11}|^2}  \Bigr)\Bigr).    \label{eq:Jb4}  
\end{align}
\end{lemma}
\begin{proof}
See Appendix~\ref{app:Jbounds}. 
\end{proof}

  Finally, by incorporating the results \eqref{eq:differenceb11G} and \eqref{eq:differenceb22G} of Lemma~\ref{lm:differencebG}  into  \eqref{eq:sumrate1}, it gives the following bound on the sum rate
 \begin{align}
 &R_1 + R_2  - \epsilon_{1,n} - \epsilon_{2,n}  - 2  \epsilon \non\\  
 \leq&  \frac{1}{2} \log \Bigl(1+ P^{\alpha_{22}-\alpha_{12}} \cdot \frac{|h_{22}|^2}{|h_{12}|^2} +  P^{\alpha_{21} -(\alpha_{11} - \alpha_{12} )^+} \cdot \frac{|h_{21}|^2}{|h_{11}|^2}  \Bigr)      \non \\
  + & \frac{1}{2} \log \Bigl(1+ P^{\alpha_{11}-\alpha_{21}} \cdot \frac{|h_{11}|^2}{|h_{21}|^2} +  P^{\alpha_{12} -(\alpha_{22} - \alpha_{21} )^+} \cdot \frac{|h_{12}|^2}{|h_{22}|^2}  \Bigr) +   \log 10    .
\end{align}
 By setting  $n\to \infty, \  \epsilon_{1,n} \to 0, \ \epsilon_{2,n} \to 0$ and  $\epsilon \to 0$, we  get the desired bound \eqref{eq:gaussianupbound}.

\subsection{Proof of bounds \eqref{eq:detup1} and \eqref{eq:detup2} \label{sec:detup12} }

Let us now prove bound \eqref{eq:detup1}. 
Beginning with Fano's inequality, we can bound the rate of user~1 as:
\begin{align}
  \bln R_1   
  &\leq  \Imu(w_1; y^{\bln}_1) + \bln \epsilon_{1,n}  \non \\
  &\leq  \Imu(w_1;  \{  \sqrt{P^{\alpha_{11}}} h_{11} x_{1}(t) +z_{1}(t) \}_{t=1}^{\bln}) + \bln \epsilon_{1,n}    \label{eq:2remove}  \\
    &\leq  \Imu(w_1;  \{  \sqrt{P^{\alpha_{11}}} h_{11} x_{1}(t) +z_{1}(t) \}_{t=1}^{\bln}  ,  y^{\bln}_2) + \bln \epsilon_{1,n}    \label{eq:2genie}  \\
    &\leq   \Imu(w_1;  \{  \sqrt{P^{\alpha_{11}}} h_{11} x_{1}(t) +z_{1}(t) \}_{t=1}^{\bln},  y^{\bln}_2) - \Imu(w_1; y^{\bln}_2)  +  \bln \epsilon  + \bln \epsilon_{1,n}  \label{eq:2R1sec88}    \\
&=  \Imu(w_1; \{  \sqrt{P^{\alpha_{11}}} h_{11} x_{1}(t) +z_{1}(t) \}_{t=1}^{\bln}  |  y^{\bln}_2)    +  \bln \epsilon  + \bln \epsilon_{1,n}   \non \\ 
&=  \Imu(w_1; \{  \sqrt{P^{\alpha_{11}}} h_{11} x_{1}(t) +z_{1}(t) \}_{t=1}^{\bln} |  \{    \bar{y}_2(t)    \}_{t=1}^{\bln} )    +  \bln \epsilon  + \bln \epsilon_{1,n}   \label{eq:s12y2444} \\ 
&=  \Imu(w_1; \{   \bar{y}_2(t)   - \sqrt{P^{\alpha_{11}}} h_{11} x_{1}(t) - z_{1}(t) \}_{t=1}^{\bln} |  \{    \bar{y}_2(t)    \}_{t=1}^{\bln} )    +  \bln \epsilon  + \bln \epsilon_{1,n}   \non \\ 
&=  \Imu(w_1; \{  \sqrt{P^{\alpha_{22}+ \alpha_{11} -  \alpha_{21}}}  \frac{h_{11}h_{22}}{h_{21}} x_{2}(t) +       \sqrt{P^{\alpha_{11} -  \alpha_{21}}}  \frac{h_{11}}{h_{21}} z_{2}(t)  - z_{1}(t) \}_{t=1}^{\bln} |  \{    \bar{y}_2(t)    \}_{t=1}^{\bln} )    +  \bln \epsilon  + \bln \epsilon_{1,n}   \non \\ 
& \leq    \frac{n}{2} \log \Bigl(1+ P^{\alpha_{11}-\alpha_{21}} \cdot \frac{|h_{11}|^2}{|h_{21}|^2} +  P^{\alpha_{22} + \alpha_{11} - \alpha_{21} } \cdot \frac{|h_{11}|^2 |h_{22}|^2}{|h_{21}|^2}  \Bigr)          + \bln \epsilon  + \bln \epsilon_{1,n}   \label{eq:s12y8375}
\end{align}
where   $\lim_{n\to\infty} \epsilon_{1,n} = 0$;
\eqref{eq:2remove} stems from the Markov chain  of $w_1 \to \{  \sqrt{P^{\alpha_{11}}} h_{11} x_{1}(t) +z_{1}(t) \}_{t=1}^{\bln} \to  y^{\bln}_1$;  
\eqref{eq:2genie} results from the fact that adding information does not decrease the mutual information;
 \eqref{eq:2R1sec88}  results  from  the secrecy constraint, i.e., $ \Imu(w_1; y^{\bln}_2)  \leq \bln \epsilon$ for an  arbitrary small  $\epsilon$ (cf.~\eqref{eq:defsecrecy1});
  \eqref{eq:s12y2444} uses the definition that \[ \bar{y}_2(t) \defeq  \sqrt{P^{\alpha_{11} -  \alpha_{21}}}  \frac{h_{11}}{h_{21}}y_2(t) =      \sqrt{P^{\alpha_{11}}} h_{11} x_{1}(t) +   \sqrt{P^{\alpha_{22}+ \alpha_{11} -  \alpha_{21}}}  \frac{h_{11}h_{22}}{h_{21}} x_{2}(t) +       \sqrt{P^{\alpha_{11} -  \alpha_{21}}}  \frac{h_{11}}{h_{21}} z_{2}(t) ; \]
\eqref{eq:s12y8375} follows from  the fact that $ \Imu(w_1; \{  \sqrt{P^{\alpha_{22}+ \alpha_{11} -  \alpha_{21}}}  \frac{h_{11}h_{22}}{h_{21}} x_{2}(t) +       \sqrt{P^{\alpha_{11} -  \alpha_{21}}}  \frac{h_{11}}{h_{21}} z_{2}(t)  - z_{1}(t) \}_{t=1}^{\bln} |  \{    \bar{y}_2(t)    \}_{t=1}^{\bln} )  
\leq   \sum_t  \hen(   \sqrt{P^{\alpha_{22}+ \alpha_{11} -  \alpha_{21}}}  \frac{h_{11}h_{22}}{h_{21}} x_{2}(t) +       \sqrt{P^{\alpha_{11} -  \alpha_{21}}}  \frac{h_{11}}{h_{21}} z_{2}(t)  - z_{1}(t) )  -  \sum_t  \hen(z_{1}(t)) $ and the fact that Gaussian input maximizes the differential entropy. 
Letting  $n\to \infty, \  \epsilon_{1,n} \to 0 $ and  $\epsilon \to 0$, it gives bound \eqref{eq:detup1}.  By interchanging the roles of user~1 and user~2, bound \eqref{eq:detup2} can be proved in a similar way.

\subsection{Proof of bounds \eqref{eq:detup11} and \eqref{eq:detup22} \label{sec:detup1122} }

By following from \eqref{eq:2remove}, we have
\begin{align}
  \bln R_1   -  \bln \epsilon_{1,n}     &\leq  \Imu(w_1;  \{  \sqrt{P^{\alpha_{11}}} h_{11} x_{1}(t) +z_{1}(t) \}_{t=1}^{\bln}) + \bln \epsilon_{1,n}   \non \\
     & \leq    \frac{n}{2} \log \bigl(1+ P^{\alpha_{11}} \cdot  |h_{11}|^2   \bigr)   \label{eq:s12y83255}
\end{align}
which gives the bound in \eqref{eq:detup11}.  The bound in \eqref{eq:detup22} can be  proved in a similar way by interchanging the roles of user~1 and user~2.

\subsection{Proof of bounds \eqref{eq:detup4} and \eqref{eq:detup5} \label{sec:detup45} }

Let us now prove bound \eqref{eq:detup4}.  
Let \[\tilde{x}_{k}(t) \defeq \sqrt{P^{\max\{\alpha_{k k}, \alpha_{\ell k}\}}}x_{k}(t) +  \tilde{z}_{k}(t) \]  and  $ \tilde{x}^{\bln}_{k} \defeq \{ \tilde{x}_{k}(t)  \}_{t=1}^{\bln} $ for $k,\ell   \in \{1,2\}, k \neq \ell$,  where  $\tilde{z}_{k}(t) \sim \mathcal{N}(0, 1)$ is a virtual noise that is independent of the other noise and transmitted signals. Recall that \[s_{\ell k}(t)  \defeq   \sqrt{P^{\alpha_{\ell k}}} h_{\ell k} x_{k}(t)  +z_{\ell}(t)\] for $k,\ell   \in \{1,2\}, k \neq \ell$ (cf.~\eqref{eq:defs12})
Beginning with Fano's inequality, the secure rate of user~1 is bounded as:
\begin{align}
 & \bln R_1 -  \bln \epsilon_{1,n}    \non\\
&\leq  \Imu(w_1; y^{\bln}_1)   \non \\
  &\leq  \Imu(w_1; y^{\bln}_1)    - \Imu(w_1; y^{\bln}_2) + \bln \epsilon    \label{eq:2R1sec325}  \\
  &\leq  \Imu(w_1; y^{\bln}_1, \tilde{x}^{\bln}_{1}, \tilde{x}^{\bln}_{2},    y^{\bln}_2 )   - \Imu(w_1; y^{\bln}_2) + \bln \epsilon    \label{eq:2R1add223}  \\
  &=   \Imu(w_1; y^{\bln}_1, \tilde{x}^{\bln}_{1}, \tilde{x}^{\bln}_{2} |   y^{\bln}_2 )    + \bln \epsilon   \non  \\
    &=   \hen(y^{\bln}_1, \tilde{x}^{\bln}_{1}, \tilde{x}^{\bln}_{2} |   y^{\bln}_2 )  - \hen(y^{\bln}_1, \tilde{x}^{\bln}_{1}, \tilde{x}^{\bln}_{2} |   y^{\bln}_2, w_1 )   + \bln \epsilon   \non  \\
        &=   \hen(y^{\bln}_1,y^{\bln}_2 , \tilde{x}^{\bln}_{1}, \tilde{x}^{\bln}_{2} ) -  \hen(y^{\bln}_2 )  - \hen(y^{\bln}_1, \tilde{x}^{\bln}_{1}, \tilde{x}^{\bln}_{2} |   y^{\bln}_2, w_1 )   + \bln \epsilon   \non  \\
        &=   \hen(\tilde{x}^{\bln}_{1}, \tilde{x}^{\bln}_{2} ) -  \hen(y^{\bln}_2 )  + \hen(y^{\bln}_1,y^{\bln}_2 | \tilde{x}^{\bln}_{1}, \tilde{x}^{\bln}_{2} )  - \hen(y^{\bln}_1, \tilde{x}^{\bln}_{1}, \tilde{x}^{\bln}_{2} |   y^{\bln}_2, w_1 )   + \bln \epsilon     \label{eq:2R1Markove763}   
 \end{align}
where \eqref{eq:2R1sec325}  results  from a secrecy constraint (cf.~\eqref{eq:defsecrecy1});
\eqref{eq:2R1add223}  stems from the fact that adding information does not decrease the mutual information;  
 On the other hand, we have 
\begin{align}
   \bln R_1    
&\leq  \Imu(w_1; y^{\bln}_1)  + \bln \epsilon_{1,n} \non \\
  &\leq  \Imu(x^{\bln}_{1}; y^{\bln}_1 )  + \bln \epsilon_{1,n}   \label{eq:2R1Markove82435}  \\
  &=    \hen(y^{\bln}_1 )  -  \hen( y^{\bln}_1 | x^{\bln}_{1})   +  \bln \epsilon_{1,n}    \non   \\
  &=    \hen(y^{\bln}_1 )  -  \hen( s^{\bln}_{12} | x^{\bln}_{1})   +  \bln \epsilon_{1,n}     \label{eq:2R1y122}  \\
  &=    \hen(y^{\bln}_1 )  -  \hen( s^{\bln}_{12})   + \bln \epsilon_{1,n}     \label{eq:2R1ind254}  
 \end{align}
where  
\eqref{eq:2R1Markove82435}  follows from  the Markov chain  of $w_1 \to  x^{\bln}_{1}  \to  y^{\bln}_1$;  
\eqref{eq:2R1y122} results from the fact that $y_1(t) =    \sqrt{P^{\alpha_{11}}} h_{11} x_{1}(t) +   s_{12}(t) $;
\eqref{eq:2R1ind254}  follows from the independence between $x^{\bln}_{1}$ and $s^{\bln}_{12}$.
In a similar way, we have 
\begin{align}
   \bln R_2    
&\leq    \hen(y^{\bln}_2 )  -  \hen( s^{\bln}_{21})   + \bln \epsilon_{2,n} .    \label{eq:2R2ind888}  
 \end{align}
Finally, by combining \eqref{eq:2R1Markove763}, \eqref{eq:2R1ind254} and \eqref{eq:2R2ind888}, it gives
\begin{align}
 &  2\bln R_1 + \bln R_2 -  2\bln \epsilon_{1,n} - \bln \epsilon_{2,n} - \bln \epsilon    \non  \\
 &\leq    \hen(\tilde{x}^{\bln}_{1}) -  \hen( s^{\bln}_{21})  +   \hen( \tilde{x}^{\bln}_{2} ) -  \hen( s^{\bln}_{12})  + \hen(y^{\bln}_1 )     + \hen(y^{\bln}_1,y^{\bln}_2 | \tilde{x}^{\bln}_{1}, \tilde{x}^{\bln}_{2} )  - \hen(y^{\bln}_1, \tilde{x}^{\bln}_{1}, \tilde{x}^{\bln}_{2} |   y^{\bln}_2, w_1 )  \non\\
  &\leq     \hen(\tilde{x}^{\bln}_{1}) -  \hen( s^{\bln}_{21})  +   \hen( \tilde{x}^{\bln}_{2} ) -  \hen( s^{\bln}_{12})  +  \hen(y^{\bln}_1 )       + \hen(y^{\bln}_1,y^{\bln}_2 | \tilde{x}^{\bln}_{1}, \tilde{x}^{\bln}_{2} )  -  \frac{3n}{2}\log (2\pi e)  \label{eq:2R1sum7145}\\
   &\leq    \hen(\tilde{x}^{\bln}_{1}) -  \hen( s^{\bln}_{21})  +   \hen( \tilde{x}^{\bln}_{2} ) -  \hen( s^{\bln}_{12}) +  \frac{n}{2} \log \bigl(1+ P^{\alpha_{11}}   |h_{11}|^2  + P^{\alpha_{12}}   |h_{12}|^2  \bigr)    \non\\&\quad + \hen(y^{\bln}_1,y^{\bln}_2 | \tilde{x}^{\bln}_{1}, \tilde{x}^{\bln}_{2} )  -  \frac{2n}{2}\log (2\pi e)  \label{eq:2R1sum56773}\\ 
      &\leq    \hen(\tilde{x}^{\bln}_{1}) -  \hen( s^{\bln}_{21})  +   \hen( \tilde{x}^{\bln}_{2} ) -  \hen( s^{\bln}_{12}) + \frac{n}{2} \log \bigl(1+ P^{\alpha_{11}}   |h_{11}|^2  + P^{\alpha_{12}}   |h_{12}|^2  \bigr)        +   n \log 9   \label{eq:2R1sum46727}
 \end{align}
where \eqref{eq:2R1sum7145} follows from the derivation that  $\hen(y^{\bln}_1, \tilde{x}^{\bln}_{1}, \tilde{x}^{\bln}_{2} |   y^{\bln}_2, w_1 ) \geq  \hen(y^{\bln}_1, \tilde{x}^{\bln}_{1}, \tilde{x}^{\bln}_{2} |   y^{\bln}_2, w_1, x^{\bln}_{1}, x^{\bln}_{2} )    = \hen(z^{\bln}_1, \tilde{z}^{\bln}_{1}, \tilde{z}^{\bln}_{2} )  = \frac{3n}{2}\log (2\pi e)$;  
\eqref{eq:2R1sum56773} holds true because $ \hen(y^{\bln}_1 ) \leq \frac{n}{2} \log \bigl(2\pi e( 1+ P^{\alpha_{11}}   |h_{11}|^2  + P^{\alpha_{12}}   |h_{12}|^2 ) \bigr) $;
\eqref{eq:2R1sum46727} uses the fact that 
 \begin{align}
 & \hen(y^{\bln}_1,y^{\bln}_2 | \tilde{x}^{\bln}_{1}, \tilde{x}^{\bln}_{2} ) \non\\
  =& \hen\bigl( \{ y_1 (t) -    P^{\frac{\alpha_{11}  -  \max\{\alpha_{11}, \alpha_{21}\}}{2}} h_{11}  \tilde{x}_{1}(t)   -    P^{\frac{\alpha_{12}  -  \max\{\alpha_{22}, \alpha_{12}\}}{2}} h_{12}  \tilde{x}_{2}(t)  \}_{t=1}^{\bln}, \non\\ & \quad  \{ y_2 (t)  -    P^{\frac{\alpha_{22}  -  \max\{\alpha_{22}, \alpha_{12}\}}{2}} h_{22}  \tilde{x}_{2}(t)   -    P^{\frac{\alpha_{21}  -  \max\{\alpha_{11}, \alpha_{21}\}}{2}} h_{21}  \tilde{x}_{1}(t)   \}_{t=1}^{\bln}   | \tilde{x}^{\bln}_{1}, \tilde{x}^{\bln}_{2} \bigr)   \non\\
  \leq  & \sum_{t=1}^{\bln} \hen \bigl( z_1(t)    -   P^{\frac{\alpha_{11}  -  \max\{\alpha_{11}, \alpha_{21}\}}{2}} h_{11}  \tilde{z}_{1}(t)    -    P^{\frac{\alpha_{12}  -  \max\{\alpha_{22}, \alpha_{12}\}}{2}} h_{12}  \tilde{z}_{2}(t)    \bigr)  \non\\   + & \sum_{t=1}^{\bln} \hen \bigl(  z_2 (t) -    P^{\frac{\alpha_{22}  -  \max\{\alpha_{22}, \alpha_{12}\}}{2}} h_{22}  \tilde{z}_{2}(t)   -    P^{\frac{\alpha_{21}  -  \max\{\alpha_{11}, \alpha_{21}\}}{2}} h_{21}  \tilde{z}_{1}(t)  \bigr)  \non\\
    \leq  &    \frac{n}{2} \log \bigl( 2\pi \bigl(1+ P^{\alpha_{11}  -  \max\{\alpha_{11}, \alpha_{21}\}} |h_{11}|^2  + P^{\alpha_{12}  -  \max\{\alpha_{22}, \alpha_{12}\}}   |h_{12}|^2   \bigr)\bigr)     \non\\   +&  \frac{n}{2} \log \bigl( 2\pi \bigl(1+ P^{\alpha_{22}  -  \max\{\alpha_{22}, \alpha_{12}\}} |h_{22}|^2  + P^{\alpha_{21}  -  \max\{\alpha_{11}, \alpha_{21}\}}   |h_{21}|^2   \bigr)\bigr)  \non\\
     \leq  &    \frac{n}{2} \log ( 2\pi  \times 9  )    +  \frac{n}{2} \log  ( 2\pi  \times 9 ) . 
 \end{align}
 Let us focus on $\hen(\tilde{x}^{\bln}_{1}) -  \hen( s^{\bln}_{21})$ in the right-hand side of \eqref{eq:2R1sum46727}:
  \begin{align}
  \hen(\tilde{x}^{\bln}_{1}) -  \hen( s^{\bln}_{21})
  =&     \hen(\tilde{x}^{\bln}_{1}) -  \hen(\tilde{z}^{\bln}_{1})   + \hen(z^{\bln}_{2}) -   \hen( s^{\bln}_{21})    \non\\
  =&     \Imu(      \tilde{x}^{\bln}_{1} ;  x^{\bln}_{1})   -    \Imu( s^{\bln}_{21} ;  x^{\bln}_{1})    \non\\
 \leq  &     \Imu(      \tilde{x}^{\bln}_{1} ;  x^{\bln}_{1} | s^{\bln}_{21})      \non\\
    = &     \hen(      \tilde{x}^{\bln}_{1}  | s^{\bln}_{21})   -  \underbrace{\hen(      \tilde{x}^{\bln}_{1} | x^{\bln}_{1}, s^{\bln}_{21}) }_{  = \hen( \tilde{z}^{\bln}_{1})}     \non\\
   = &     \hen(      \tilde{x}^{\bln}_{1}  | s^{\bln}_{21})   -   \frac{n}{2}\log (2\pi e)    \non\\
   \leq  &   \sum_{t=1}^{\bln}  \hen(      \tilde{x}_{1}(t)  | s_{21} (t))   -   \frac{n}{2}\log (2\pi e)    \non\\
    =  &   \sum_{t=1}^{\bln}  \hen\Bigl(      \tilde{x}_{1}(t)  -  P^{\frac{ \max\{\alpha_{11}, \alpha_{21}\} - \alpha_{21} }{2}}  \frac{s_{21} (t) }{h_{21} } \big| s_{21} (t)\Bigr)   -   \frac{n}{2}\log (2\pi e)    \non\\
  \leq   &   \sum_{t=1}^{\bln}  \hen\Bigl(      \tilde{z}_{1}(t)  -  P^{\frac{ \max\{\alpha_{11}, \alpha_{21}\} - \alpha_{21} }{2}}  \frac{z_{2} (t) }{h_{21} } \Bigr)   -   \frac{n}{2}\log (2\pi e)    \non\\
    =   &     \frac{n}{2} \log \bigl(1+ P^{ (\alpha_{11} - \alpha_{21} )^+} \cdot  \frac{1}{|h_{21}|^2}   \bigr).   \label{eq:2R1sum54356}
 \end{align}
  Similarly, we have 
  \begin{align}
  \hen( \tilde{x}^{\bln}_{2} ) -  \hen( s^{\bln}_{12})  \leq     \frac{n}{2} \log \bigl(1+ P^{ (\alpha_{22} - \alpha_{12} )^+} \cdot  \frac{1}{|h_{12}|^2}   \bigr).    \label{eq:2R1sum9285}
 \end{align}
At this point, by incorporating  \eqref{eq:2R1sum54356} and \eqref{eq:2R1sum9285} into \eqref{eq:2R1sum46727}, it gives
\begin{align}
 &  2 R_1 +  R_2 -  2 \epsilon_{1,n} -  \epsilon_{2,n} -  \epsilon    \non  \\
      &\leq   \frac{1}{2} \log \bigl(1+   \frac{P^{ (\alpha_{11} - \alpha_{21} )^+}}{|h_{21}|^2}   \bigr)    +   \frac{1}{2} \log \bigl(1+   \frac{P^{ (\alpha_{22} - \alpha_{12} )^+}}{|h_{12}|^2}   \bigr)  + \frac{1}{2} \log \bigl(1+ P^{\alpha_{11}}   |h_{11}|^2  + P^{\alpha_{12}}   |h_{12}|^2  \bigr)        +  \log 9 .  \non 
 \end{align}
By setting  $n\to \infty, \  \epsilon_{1,n}, \epsilon_{2,n} \to 0 $ and  $\epsilon \to 0$, it gives bound \eqref{eq:detup4}.
By interchanging the roles of user~1 and user~2, bound \eqref{eq:detup5} can be proved in a similar way.

\section{Conclusion}   \label{sec:conclusion}

This work considered the two-user Gaussian interference channel with confidential messages. For the symmetric setting,  this work completed the optimal secure sum GDoF characterization for all the interference regimes.
For the general setting, this work showed that a simple scheme without cooperative jamming (i.e., GWC-TIN scheme) can achieve  the secure sum capacity to within a constant gap, when the conditions of \eqref{eq:capGaussian1} and \eqref{eq:capGaussian2} are satisfied. 
In this GWC-TIN scheme,  each transmitter uses a Gaussian wiretap codebook, while each receiver treats interference as noise when decoding the desired message. 
For the symmetric case, this simple scheme is optimal when the interference-to-signal ratio $\alpha$ is no more than $2/3$. However, when the ratio $\alpha$ is more than $2/3$, we showed that this simple scheme is not optimal anymore and a scheme with cooperative jamming is proposed to achieve the optimal secure sum GDoF.  
In the future work we will try to understand when it is  necessary to use cooperative jamming and when it is not,  for the secure communication over the other  networks.

\appendices

\section{Proof of Lemma~\ref{lm:Jbounds}} \label{app:Jbounds}

Remind that  $J_1= \hen(s_{11}^{\bln}| y^{\bln}_2, w_1)$, $ J_2 = \hen(s_{11}^{\bln} | y^{\bln}_1, w_1) $, $J_3 = \hen(s_{12}^{\bln} | y^{\bln}_2, s_{11}^{\bln}, w_1)$, $J_4= \hen( y^{\bln}_2| s_{12}^{\bln}, s_{11}^{\bln}, w_1)$,  $s_{11}(t)  =  \sqrt{P^{(\alpha_{11}-\alpha_{12})^+}} h_{11} x_{1}(t) +  \tilde{z}_{1}(t)$, and $ s_{12}(t) =  \sqrt{P^{\alpha_{12}}} h_{12} x_{2}(t)  +z_{1}(t)$.  
At first we focus on the lower bound of $J_1$:
\begin{align}
J_1  & = \hen(s_{11}^{\bln}| y^{\bln}_2, w_1)  \nonumber \\
       & \geq \hen(s_{11}^{\bln}|  x_1^{\bln}, y^{\bln}_2, w_1)     \label{eq:Jb4256}  \\
      & = \hen(   \{   \sqrt{P^{(\alpha_{11}-\alpha_{12})^+}} h_{11} x_{1}(t) +  \tilde{z}_{1}(t) \}_{t=1}^{\bln}  |  x_1^{\bln}, y^{\bln}_2, w_1)     \nonumber  \\
       & = \hen(   \{ \tilde{z}_{1}(t) \}_{t=1}^{\bln}  |  x_1^{\bln}, y^{\bln}_2, w_1)     \label{eq:Jb8488}  \\
       & = \hen(   \{ \tilde{z}_{1}(t) \}_{t=1}^{\bln} )    \nonumber  \\
           & =  \frac{n}{2}\log (2\pi e) \nonumber   
\end{align}
where \eqref{eq:Jb4256} follows from the fact that conditioning reduces differential entropy;   \eqref{eq:Jb8488} follows from the fact that $\hen(a|b)=\hen(a-b |b)$ for any continuous random variables $a$ and $b$; the last equality holds true because $\hen( \tilde{z}_{1}(t)) = \frac{1}{2}\log (2\pi e) $.   Similarly, we have 
\begin{align}
J_3  & = \hen(s_{12}^{\bln} | y^{\bln}_2, s_{11}^{\bln}, w_1)  \nonumber \\
       & \geq  \hen(s_{12}^{\bln} |  x_2^{\bln}, y^{\bln}_2, s_{11}^{\bln}, w_1)    \nonumber \\
       & = \hen(   \{ z_{1}(t) \}_{t=1}^{\bln} )    \nonumber  \\
       & =  \frac{n}{2}\log (2\pi e) . \nonumber 
\end{align}

Now we focus on the  upper bound of $J_2$:
\begin{align}
J_2   
 = &\hen(s_{11}^{\bln} | y^{\bln}_1, w_1)  \nonumber \\
  = &\sum_{t=1}^{\bln}  \hen(s_{11}(t) | s_{11}^{t-1}, y^{\bln}_1, w_1)  \label{eq:Jb0263}  \\
  \leq & \sum_{t=1}^{\bln}  \hen(s_{11}(t) |  y_1(t))  \label{eq:Jb9255}  \\
    = & \sum_{t=1}^{\bln}  \hen\bigl(s_{11}(t) -  \sqrt{P^{-\alpha_{12} }}y_1(t) |  y_1(t)\bigr)  \label{eq:Jb2595}  \\
     = & \sum_{t=1}^{\bln}  \hen \Bigl(  \tilde{z}_{1}(t)  + (  \sqrt{P^{(\alpha_{11}-\alpha_{12})^+}} - \sqrt{P^{\alpha_{11}-\alpha_{12} }}  )  h_{11} x_{1}(t) 
     -   h_{12} x_{2}(t)  - \sqrt{P^{-\alpha_{12} }} z_{1}(t)    \big|  y_1(t) \Bigr)  \nonumber  \\
     \leq & \sum_{t=1}^{\bln}  \hen \Bigl(    \tilde{z}_{1}(t)  + (  \sqrt{P^{(\alpha_{11}-\alpha_{12})^+}} - \sqrt{P^{\alpha_{11}-\alpha_{12} }}  )  h_{11} x_{1}(t)  
      -    h_{12} x_{2}(t)  - \sqrt{P^{-\alpha_{12} }} z_{1}(t)     \Bigr)   \label{eq:Jb5367}  \\
       \leq&   \frac{n}{2} \log (2\pi e(1+  \underbrace{(  \sqrt{P^{(\alpha_{11}-\alpha_{12})^+}} - \sqrt{P^{\alpha_{11}-\alpha_{12} }}  )^2 }_{\leq 1}  \cdot \underbrace{|h_{11}|^2}_{\leq 4}+    \underbrace{|h_{12}|^2}_{\leq 4} + \underbrace{P^{-\alpha_{12} } }_{\leq 1} ))  \label{eq:Jb44112} \\
       \leq &   \frac{n}{2} \log (20 \pi e )    \label{eq:Jb00998} 
\end{align}
where \eqref{eq:Jb0263} results from chain rule;
 \eqref{eq:Jb9255}  and \eqref{eq:Jb5367} follow from the fact that conditioning reduces differential entropy;
  \eqref{eq:Jb2595} uses the fact that $\hen(a|b)=\hen(a-\beta b |b)$ for a constant $\beta$; 
  \eqref{eq:Jb44112}  follows from  the fact that  $ \hen \bigl(   \tilde{z}_{1}(t)  +  \beta_0  x_{1}(t)  - \beta_1  x_{2}(t)  -  \beta_2 z_{1}(t)   \bigr)  \leq  \frac{1}{2} \log (2\pi e(1+ \beta_0^2 + \beta_1 ^2 +   \beta_2^2  ))$ for constants $\beta_0$, $\beta_1$ and $\beta_2$;
  \eqref{eq:Jb00998}  uses the identities  $0\leq \sqrt{P^{(\alpha_{11}-\alpha_{12})^+}} - \sqrt{P^{\alpha_{11}-\alpha_{12} }}   \leq  1$ and  $ P^{-\alpha_{12} } \leq 1 $. 
  Recall that  $P\geq 1$ and $\alpha_{k\ell} \geq 0,   h_{k\ell} \in (1,2],  \forall k, \ell  \in \{1,2\}$.
Similarly, we have the following bound on $J_4$:
\begin{align}
J_4  & = \hen( y^{\bln}_2| s_{12}^{\bln}, s_{11}^{\bln}, w_1)  \nonumber \\
        & = \sum_{t=1}^{\bln}  \hen(y_{2}(t) | y_{2}^{t-1}, s_{12}^{\bln}, s_{11}^{\bln}, w_1)  \label{eq:Jb02631}  \\
        & \leq  \sum_{t=1}^{\bln}  \hen(y_{2}(t) | s_{12}(t), s_{11}(t) )  \label{eq:Jb92551}  \\
        & =  \sum_{t=1}^{\bln}  \hen\bigl(y_{2}(t) -  \sqrt{P^{\alpha_{22}-\alpha_{12}}}  \frac{h_{22}}{h_{12}} s_{12}(t)  
        -  \sqrt{P^{\alpha_{21} -(\alpha_{11} - \alpha_{12} )^+}}  \frac{h_{21}}{h_{11}}  s_{11}(t)   | s_{12}(t), s_{11}(t) \bigr)  \label{eq:Jb25951}  \\
        & =  \sum_{t=1}^{\bln}  \hen \Bigl( z_{2}(t)   -  \sqrt{P^{\alpha_{22}-\alpha_{12}}}  \frac{h_{22}}{h_{12}}  z_{1}(t) 
        -  \sqrt{P^{\alpha_{21} -(\alpha_{11} - \alpha_{12} )^+}}  \frac{h_{21}}{h_{11}}   \tilde{z}_{1}(t)       \big| s_{12}(t), s_{11}(t)  \Bigr)  \nonumber  \\
        & \leq  \sum_{t=1}^{\bln}   \hen \Bigl( z_{2}(t)   -  \sqrt{P^{\alpha_{22}-\alpha_{12}}} \frac{h_{22}}{h_{12}}  z_{1}(t)  
        -  \sqrt{P^{\alpha_{21} -(\alpha_{11} - \alpha_{12} )^+}}  \frac{h_{21}}{h_{11}}   \tilde{z}_{1}(t)   \Bigr)   \label{eq:Jb53671}  \\
       & \leq    \frac{n}{2} \log \Bigl(2\pi e\Bigl(1+ P^{\alpha_{22}-\alpha_{12}} \cdot \frac{|h_{22}|^2}{|h_{12}|^2} +  P^{\alpha_{21} -(\alpha_{11} - \alpha_{12} )^+} \cdot \frac{|h_{21}|^2}{|h_{11}|^2}  \Bigr)\Bigr)  \nonumber 
\end{align}
where \eqref{eq:Jb02631} results from chain rule;
 \eqref{eq:Jb92551}  and \eqref{eq:Jb53671} follow from the fact that conditioning reduces differential entropy;
     \eqref{eq:Jb25951} uses the fact that $\hen(a|b,c)=\hen(a-\beta_1b -\beta_2 c |b, c)$ for  constants $\beta_1$ and  $\beta_2$;
 the last inequality  stems from the fact that $\hen \bigl( z_{2}(t)   -  \beta_3   z_{1}(t)    -  \beta_4   \tilde{z}_{1}(t)     \bigr)   \leq \frac{1}{2} \log (2\pi e(1+ \beta_3^2 + \beta_4^2))$ for  constants $\beta_3$ and  $\beta_4$. At this point we complete the proof.

\section{The gap between secure sum capacity  upper  and lower bounds for  Theorem~\ref{thm:GaussianNCJ} } \label{sec:gap}

As discussed in Section~\ref{sec:noCJGau}, the GWC-TIN scheme achieves the  secure sum capacity lower bound  
\begin{align*}
  &C^{lb}_{\text{sum}}  \defeq \frac{1}{2} \log \bigl(  1+    \frac{|h_{11}|^2 P^{\alpha_{11} - \alpha_{21} }}{1+ |h_{12}|^2 }  \bigr)  -  \frac{1}{2} \log (1+ |h_{21}|^2)  +  \frac{1}{2} \log \bigl(  1+     \frac{|h_{22}|^2 P^{\alpha_{22} -  \alpha_{12} }}{ 1+ |h_{21}|^2 }\bigr)  -   \frac{1}{2} \log (1+ |h_{12}|^2) 
  \end{align*}  (see \eqref{eq:Rk111} and \eqref{eq:Rk432}).
In Lemma~\ref{lm:gupper} (see  Section~\ref{sec:converse}), we provide a secure sum capacity upper bound in  \eqref{eq:gaussianupbound}, that is,
\begin{align*}
 C_{\text{sum}}  \leq  C^{ub}_{\text{sum}}  \defeq  &  
 \frac{1}{2} \log \Bigl(1+ P^{\alpha_{22}-\alpha_{12}} \cdot \frac{|h_{22}|^2}{|h_{12}|^2} +  P^{\alpha_{21} -(\alpha_{11} - \alpha_{12} )^+} \cdot \frac{|h_{21}|^2}{|h_{11}|^2}  \Bigr)      \non \\
  + & \frac{1}{2} \log \Bigl(1+ P^{\alpha_{11}-\alpha_{21}} \cdot \frac{|h_{11}|^2}{|h_{21}|^2} +  P^{\alpha_{12} -(\alpha_{22} - \alpha_{21} )^+} \cdot \frac{|h_{12}|^2}{|h_{22}|^2}  \Bigr) +   \log 10 . 
\end{align*}
 If the following two conditions are satisfied, 
 \begin{align}
\alpha_{22} + (\alpha_{11 } - \alpha_{12})^+  &\geq   \alpha_{21}+  \alpha_{12}   \label{eq:capGaussian111} \\
\alpha_{11} + (\alpha_{22 } - \alpha_{21})^+  &\geq   \alpha_{21}+  \alpha_{12} \label{eq:capGaussian222} 
\end{align}
(see \eqref{eq:capGaussian1} and \eqref{eq:capGaussian2}), then  the gap between $C^{ub}_{\text{sum}} $ and $C^{lb}_{\text{sum}} $ is bounded by    
\begin{align}
&C^{ub}_{\text{sum}} -  C^{lb}_{\text{sum}}   \non\\
 \leq   &    \frac{1}{2} \log \bigl(1+ 4 P^{\alpha_{22}-\alpha_{12}} +  4 P^{\alpha_{21} -(\alpha_{11} - \alpha_{12} )^+}    \bigr) 
   +  \frac{1}{2} \log \bigl(1+ 4 P^{\alpha_{11}-\alpha_{21}} +  4 P^{\alpha_{12} -(\alpha_{22} - \alpha_{21} )^+}   \bigr) +   \log 10     \non\\
   & -  \frac{1}{2} \log \bigl(  1+    \frac{ P^{\alpha_{11} - \alpha_{21} }}{5 }  \bigr)   -  \frac{1}{2} \log \bigl(  1+     \frac{ P^{\alpha_{22} -  \alpha_{12} }}{ 5 }\bigr)  +    \log 5     \label{eq:gap5525}  \\
\leq   &    \frac{1}{2} \log \bigl(\frac{1+ 8 P^{\alpha_{22}-\alpha_{12}}  }{  1+ \frac{1}{5} P^{\alpha_{22}-\alpha_{12}} }   \bigr)  +  \frac{1}{2} \log \bigl(\frac{1+ 8 P^{\alpha_{11}-\alpha_{21}}  }{  1+ \frac{1}{5} P^{\alpha_{11}-\alpha_{21}} }   \bigr)   +   \log 50       \label{eq:gap7725}  \\
\leq   &     \log 40   +   \log 50       \label{eq:gap2946}  \\
\leq   &   11       \label{eq:gap0985}  
\end{align}
where  \eqref{eq:gap5525} uses the fact that $h_{k\ell} \in (1,2],  \forall k, \ell  \in \{1,2\}$;
\eqref{eq:gap7725} follows from the conditions in \eqref{eq:capGaussian111} and \eqref{eq:capGaussian222};
\eqref{eq:gap2946} results  from the  identity that $ \frac{1+ a_1 b}{1+ a_2 b} \leq  \frac{a_1 }{a_2}$ for any positive numbers $a_1, a_2$ and $b$ such that $a_1 \geq 1  \geq a_2 > 0$.
Recall that  $P\geq 1$ and $\alpha_{k\ell} \geq 0,  \forall k, \ell  \in \{1,2\}$.
Note that, the gap can be further reduced by optimizing the computations in the converse and achievability.

\section{Proof of Lemma~\ref{lm:AWGNic}  \label{sec:AWGNic} }

This section  provides  the proof of Lemma~\ref{lm:AWGNic} (see Section~\ref{sec:CJGau}).
Recall that we consider the communication of $x$  over  a  channel model given in \eqref{eq:aqmy1}, that is, 
  \begin{align}
   y=  \sqrt{P^{\alpha_1}}  h x + \sqrt{P^{\alpha_2}} g + z   \non
 \end{align}
where  $x \in \Omega (\xi,  Q)$ and $g  \in  \Sc_{g}$ is a discrete random variable such that $ |g | \leq  g_{\max}, \ \forall g \in  \Sc_{g}$,  where $g_{\max}$ is a positive and finite constant  independent of $P$.   
The  minimum distance of the constellation  for $\sqrt{P^{\alpha_1}}  h x$ is   $d_{\text{min}} (\sqrt{P^{\alpha_1}} hx) =\sqrt{P^{\alpha_1}}h \cdot \xi$.
For this channel model, the probability of error for decoding $x$ from $y$ is 
 \begin{align}
  \text{Pr} (e) &=\text{Pr} [ x \neq \hat{x} ]     \non\\
 &=  \sum_{ i = -Q }^{Q} \text{Pr} [  x = \xi \cdot  i ]  \cdot  \text{Pr} [ x \neq \hat{x} | x = \xi \cdot  i ]   \non \\
  &\leq  \sum_{ i = -Q }^{Q} \text{Pr} [  x = \xi \cdot i ]  \cdot  \Bigl(   \text{Pr} [   z   <  - P^{\frac{\alpha_2}{2}} g  -   d_{\text{min}} /2  ]  +  \text{Pr} [     z   >  - P^{\frac{\alpha_2}{2}} g   +  d_{\text{min}} /2  ]   \Bigr)  \non \\
    &=   \text{Pr} [   z   >   P^{\frac{\alpha_2}{2}} g  +   d_{\text{min}} /2  ]  +  \text{Pr} [     z   >  - P^{\frac{\alpha_2}{2}} g   +  d_{\text{min}} /2  ]     \label{eq:error2244}  \\
  & \leq       \text{Pr} [   z   >   P^{\frac{\alpha_2}{2}} g  +   d_{\text{min}} /2    \  |  \  g =  - g_{\max}  ]  +  \text{Pr} [     z   >  -P^{\frac{\alpha_2}{2}}  g   +  d_{\text{min}} /2  \  |  \  g = g_{\max} ]     \label{eq:error2256} \\
    & =   2  \cdot     {\bf{Q}} \bigl(  \frac{  d_{\text{min}} /2 -  P^{\frac{\alpha_2}{2}} g_{\max}  }{ \sigma} \bigr)     \label{eq:error2277}   
  \end{align}
  where $\hat{x}$ is the estimate for $x$ by choosing the closest point in $\Omega (\xi,  Q)$, based on the observation $y$;
  \eqref{eq:error2244} follows from the fact that  $1- {\bf{Q}} (a)  = {\bf{Q}} (- a )$ for any $a \in \Rc$,  where the ${\bf{Q}}$-function is defined as ${\bf{Q}}(a )  \defeq  \frac{1}{\sqrt{2\pi}} \int_{a}^{\infty}  \exp( -\frac{ s^2}{2} ) d s$.      
When $   d_{\text{min}} /2  \geq  P^{\frac{\alpha_2}{2}} g_{\max}   $,  the error probability can be further bounded as
  \begin{align}
   \text{Pr} (e)  &\leq    \exp \Bigl(    -  \frac{  (d_{\text{min}} /2 - P^{\frac{\alpha_2}{2}}  g_{\max}  )^2 }{2 \sigma^2} \Bigr),  \non
    \end{align}
 by using the identity that  $ {\bf{Q}} (a ) \leq   \frac{1}{2}\exp ( -  a^2 /2 ),  \    \forall a \geq 0.$ 
At this point, for the case of $\alpha_1 - \alpha_2 >0$, by setting  $Q$ and $\xi$ such that 
  \begin{align*}
   Q =  \frac{P^{\frac{\bar{\alpha}}{2}} \cdot h \gamma }{2 g_{\max} },    \quad  \quad    \xi =  \gamma \cdot \frac{ 1}{Q},     \quad   \forall \bar{\alpha} \in (0, \alpha_1 -\alpha_2)
    \end{align*}
where $\gamma \in (0, 1/\sqrt{2}]$ is a constant independent of $P$, then it holds true that $d_{\text{min}} /2 >    P^{\frac{\alpha_2}{2}} g_{\max}$ and the probability of error for decoding a symbol $x$ from $y$ is \[ \text{Pr} (e) \to 0  \quad  \text{as} \quad   P\to \infty.\]

\section{Proof of Lemma~\ref{lm:rateerror2334}  \label{sec:rateerror2334} }

We will prove that  when $2/3 < \alpha \leq 3/4$, given the signal design in \eqref{eq:constellationGsym1}-\eqref{eq:constellationGsym1u} and \eqref{eq:para111}-\eqref{eq:xvkkk12334},  the error probability of estimating  $\{v_{k,c}, v_{k,p} \}$ from $y_k$  is
 \begin{align}
 \text{Pr} [  \{ v_{k,c} \neq \hat{v}_{k,c} \} \cup  \{ v_{k,p} \neq \hat{v}_{k,p} \}  ]  \to 0         \quad \text {as}\quad  P\to \infty    \non
 \end{align}
for $k=1,2$, where  $ \hat{v}_{k,c}$ and  $\hat{v}_{k,p}  $ are the corresponding estimates for  $v_{k,c}$ and $v_{k,p}$, respectively, based on the observation $y_k$ expressed in \eqref{eq:yvk12334} and \eqref{eq:yvk22334}. 

Due to the symmetry we will focus on the proof for the first user ($k=1$). 
In the first step,  we estimate $v_{1,c} \in \Omega ( \xi =  \gamma_{v_{1,c}} \cdot \frac{ 1}{Q} ,   \   Q =  P^{ \frac{ 3\alpha -2 - \epsilon }{2}} )$ from $y_1$  by treating the other signals as noise, where $y_1$ is expressed in \eqref{eq:yvk12334}. Note that $y_1$ in \eqref{eq:yvk12334} can be rewritten as 
\begin{align}
 y_1  =    \sqrt{P} h_{11}h_{22}  v_{1,c}      +  \sqrt{P^{ \alpha}} g   +  z_{1}    \label{eq:yvk101}
\end{align}
where 
  \begin{align}
 g  &  \defeq      h_{12} h_{11}(   v_{2,c}  +  u_1)    +   \sqrt{P^{ -(1-\alpha) }}  h_{12} h_{21} u_{2}  + \sqrt{P^{ -( 2\alpha -1)}} h_{11}h_{22} v_{1,p}    +   \sqrt{P^{ - \alpha}}  h_{12} h_{11}  v_{2,p} . \non
 \end{align}
It holds true that \[ |g | \leq  5 \sqrt{2} \] for any realizations of $g$,  given the signal design in  \eqref{eq:constellationGsym1}-\eqref{eq:constellationGsym1u} and \eqref{eq:para111}-\eqref{eq:xvkkk12334},  under the regime of $2/3 <  \alpha \leq  3/4$.
Then, from Lemma~\ref{lm:AWGNic} we can conclude that  the error probability of estimating $v_{1,c}$ from $y_1$  is 
 \begin{align}
\text{Pr} [v_{1,c} \neq \hat{v}_{1,c}] \to 0, \quad \text {as}\quad  P\to \infty.   \label{eq:error776}
\end{align}

In the second step, we  remove the decoded $v_{1,c}$ from $y_1$ and then estimate $v_{2,c} + u_1 \in   2\cdot \Omega (\xi   =  \gamma_{v_{2,c}} \cdot \frac{ 1}{Q},   \   Q =  P^{ \frac{ 3\alpha -2 - \epsilon}{2}} )$  from  the following observation 
 \begin{align}
y_1-  \sqrt{P} h_{11} v_{1,c}   =       \sqrt{P^{ \alpha}}  h_{12} h_{11}(   v_{2,c}  +  u_1)    +   \sqrt{P^{ 2\alpha - 1 }}  g'  +  z_{1}       \label{eq:yvk102}
\end{align}
 where  $2\cdot \Omega (\xi,  Q)  \defeq   \{ \xi \cdot a :   \    a \in  \Zc  \cap [-2Q,   2Q]   \}$,  and   $ g'  \defeq   h_{12} h_{21} u_{2}  +  \sqrt{P^{ 2 - 3\alpha }} h_{11}h_{22} v_{1,p}    +     \sqrt{P^{ 1-2\alpha  }}   h_{12} h_{11}  v_{2,p}$.  It holds true that $ |g' | \leq  3  \sqrt{2}$ for any realizations of $g'$ in this case with $2/3 <  \alpha \leq  3/4$.
Let $s_{vu} \defeq v_{2,c} + u_1$  and let  $\hat{s}_{vu}$ be the estimate of $s_{vu}$. Then from Lemma~\ref{lm:AWGNic} we can conclude that  
 \[\text{Pr} [s_{vu}  \neq  \hat{s}_{vu} |  v_{1,c} = \hat{v}_{1,c}] \to 0, \quad \text {as}\quad  P\to \infty.\]
At this point, we have
  \begin{align}
 &\text{Pr} [s_{vu}  \neq  \hat{s}_{vu} ]  \non\\
  = & \text{Pr} [v_{1,c} = \hat{v}_{1,c}]  \cdot  \text{Pr} [s_{vu}  \neq  \hat{s}_{vu} |  v_{1,c} = \hat{v}_{1,c}]  +  
 \text{Pr} [v_{1,c} \neq \hat{v}_{1,c}]  \cdot  \text{Pr} [s_{vu}  \neq  \hat{s}_{vu} |  v_{1,c} \neq  \hat{v}_{1,c}]      \non\\
  \leq &\text{Pr} [s_{vu}  \neq  \hat{s}_{vu} |  v_{1,c} = \hat{v}_{1,c}] +  \text{Pr} [v_{1,c} \neq \hat{v}_{1,c}]   \to 0    \quad \text {as }\quad  P \to \infty.    \label{eq:error887}
\end{align}

In the third step, we remove the decoded $v_{2,c} + u_1$  from $y_1$ and then decode $  u_{2}  \in  \Omega (\xi   =  \gamma_{u_{2}} \cdot \frac{ 1}{Q},    \   Q =  P^{ \frac{ 3\alpha -2 - \epsilon}{2}} )$.
With the similar steps as before, from Lemma~\ref{lm:AWGNic} we can conclude  
 \begin{align}
 \text{Pr} [ u_{2}  \neq    \hat{u}_{2}] &   \to 0    \quad \text {as }\quad  P \to \infty.   \label{eq:error998}
\end{align}

In the final step, similarly, we remove the decoded $u_{2}$ from $y_1$ and then decode $v_{1,p} \in    \Omega (\xi   =\gamma_{v_{1,p}} \cdot \frac{ 1}{Q},   \   Q = P^{ \frac{ 1 - \alpha - \epsilon}{2}} )$, 
with error probability given as
  \begin{align}
 \text{Pr} [ v_{1,p} \neq  \hat{v}_{1,p}] &   \to 0    \quad \text {as }\quad  P \to \infty.   \label{eq:error2950}
\end{align}

By combining the results of \eqref{eq:error776} and \eqref{eq:error2950}, it holds true that 
the error probability of estimating $\{v_{1,c}, v_{1,p} \}$ from $y_1$ is
\begin{align}
 \text{Pr} [  \{ v_{1,c} \neq \hat{v}_{1,c} \} \cup  \{ v_{1,p} \neq \hat{v}_{1,p} \}  ]   \leq   \text{Pr} [  v_{1,c} \neq \hat{v}_{1,c} ]  +  \text{Pr} [  v_{1,p} \neq \hat{v}_{1,p}  ]  \to   0         \quad \text {as}\quad  P\to \infty   \non 
\end{align}
which completes the proof for the case of $k=1$. 
Due to the symmetry, the proof for the case of $k=2$ follows from the above steps, with the roles of users interchanged.

\section{Proof of Lemma~\ref{lm:rateerror322}  \label{sec:rateerror322} }

The proof of Lemma~\ref{lm:rateerror322} is very similar to the proof of Lemma~\ref{lm:rateerror2334}.
In this case we will prove that,  when $3/2 \leq \alpha \leq 2$, and given the signal design in \eqref{eq:constellationGsym1}-\eqref{eq:constellationGsym1u} and \eqref{eq:para333}-\eqref{eq:xvkkk322},  the error probability of estimating $v_{k,c}$ from $y_k$  is
 \begin{align}
\text{Pr} [ v_{k,c} \neq \hat{v}_{k,c}  ]  \to 0         \quad \text {as}\quad  P\to \infty     \non
 \end{align}
for $k=1,2$, where  $ \hat{v}_{k,c}$  is  the corresponding estimate for  $v_{k,c}$  based on the observation $y_k$ expressed in \eqref{eq:yvk1322} and \eqref{eq:yvk2322}. 

Due to the symmetry we will focus on the proof for the first user ($k=1$). 
In the first step,  we estimate $u_{2} \in \Omega ( \xi =  \gamma_{u_{2}} \cdot \frac{ 1}{Q} ,   \   Q =  P^{ \frac{2 - \alpha  - \epsilon }{2}} )$ from $y_1$  by treating the other signals as noise, where $y_1$ is expressed in \eqref{eq:yvk1322}. Note that $y_1$ in \eqref{eq:yvk1322} can be rewritten as 
\begin{align}
 y_1  =    \sqrt{P^{ \alpha }}  h_{12} h_{21} u_{2}      +  \sqrt{P} \tilde{g}   +  z_{1}    \label{eq:yvk101322}
\end{align}
where 
  \begin{align}
 \tilde{g}  &  \defeq      h_{12} h_{11}(   v_{2,c}  +  u_1)  +   \sqrt{P^{ -( \alpha-1)}} h_{11}h_{22} v_{1,c}    . \non
 \end{align}
It holds true that \[ |\tilde{g} | \leq  3 \sqrt{2} \] for any realizations of $\tilde{g}$, under the regime of $3/2 \leq   \alpha \leq  2$.
Then, from Lemma~\ref{lm:AWGNic} we can conclude that  the error probability of estimating $u_{2}$ from $y_1$  is 
 \begin{align}
\text{Pr} [u_{2} \neq \hat{u}_{2}] \to 0, \quad \text {as}\quad  P\to \infty.   \label{eq:error776322}
\end{align}

In the second step, we  remove the decoded $u_{2}$ from $y_1$ and then estimate $v_{2,c} + u_1 \in   2\cdot \Omega ( \xi =  \gamma_{u_{1}} \cdot \frac{ 1}{Q} ,   \   Q =  P^{ \frac{2 - \alpha  - \epsilon }{2}} )$ from  the following observation 
 \begin{align}
y_1-   \sqrt{P^{ \alpha }}  h_{12} h_{21} u_{2}  =      \sqrt{P}  h_{12} h_{11}(   v_{2,c}  +  u_1)    +   \sqrt{P^{ 2- \alpha}} h_{11}h_{22} v_{1,c}   +  z_{1}.       \label{eq:yvk102322}
\end{align}
 It holds true that $ |h_{11}h_{22} v_{1,c} | \leq   \sqrt{2}$ for any realizations of $v_{1,c}$.
Let $s_{vu} \defeq v_{2,c} + u_1$  and let  $\hat{s}_{vu}$ be the estimate of $s_{vu}$. Then, from Lemma~\ref{lm:AWGNic} we can conclude that  
 \[\text{Pr} [s_{vu}  \neq  \hat{s}_{vu} |  u_{2} = \hat{u}_{2}   ] \to 0, \quad \text {as}\quad  P\to \infty\]
and that
  \begin{align}
 \text{Pr} [s_{vu}  \neq  \hat{s}_{vu} ]   \leq &\text{Pr} [s_{vu}  \neq  \hat{s}_{vu} |  u_{2} = \hat{u}_{2}   ] +  \text{Pr} [u_{2} \neq  \hat{u}_{2} ]   \to 0    \quad \text {as }\quad  P \to \infty.    \label{eq:error887322}
\end{align}

In the final step,  we remove the decoded $v_{2,c} + u_1$ from $y_1$ and then decode $v_{1,c} \in   \Omega ( \xi =  \gamma_{v_{1,c}} \cdot \frac{ 1}{Q} ,   \   Q =  P^{ \frac{2 - \alpha  - \epsilon }{2}} )$  from  the following observation  \[y_1-   \sqrt{P^{ \alpha }}  h_{12} h_{21} u_{2}  - \sqrt{P}  h_{12} h_{11}(   v_{2,c}  +  u_1) =   \sqrt{P^{ 2- \alpha}} h_{11}h_{22} v_{1,c}    +  z_{1}.      \]  At this point one can easily conclude that the associated error probability is  
  \begin{align}
 \text{Pr} [ v_{1,c} \neq  \hat{v}_{1,c}] &   \to 0    \quad \text {as }\quad  P \to \infty   \label{eq:error2950322}
\end{align}
which completes the proof for the case $k=1$, as well as the proof for the case $k=2$ due to the symmetry.

\section{Proof of Lemma~\ref{lm:rateerror132}  \label{sec:rateerror132} }

We will prove Lemma~\ref{lm:rateerror132} in this section.  Specifically,  we will prove that, when  $1 \leq \alpha \leq 3/2$ and given the signal design in \eqref{eq:constellationGsym1}-\eqref{eq:constellationGsym1u} and \eqref{eq:para13211}-\eqref{eq:xvkkk1132}, then for almost all the channel coefficients  $\{h_{k\ell}\} \in (1, 2]^{2\times 2}$,  the error probability of estimating  $v_{k,c}$ from $y_k$  is   
\begin{align}
 \text{Pr} [  v_{k,c} \neq \hat{v}_{k,c}  ]  \to 0         \quad \text {as}\quad  P\to \infty   \non
 \end{align}
for $k=1,2$, where  $ \hat{v}_{k,c}$  is the corresponding estimate for  $v_{k,c}$ based on the observation $y_k$ expressed in \eqref{eq:yvk1322132} and \eqref{eq:yvk2322132}. 
Similar to the proof of Lemma~\ref{lm:rateerror341}, this proof will use the approaches of noise removal and signal separation. 

Due to the symmetry we will focus on the proof for the first user ($k=1$). 
Let us first rewrite  $y_1$ expressed in \eqref{eq:yvk1322132}  as 
\begin{align}
y_{1} &=    \sqrt{P^{ 2- \alpha}} h_{11}h_{22} v_{1,c}   +    \sqrt{P}  h_{12} h_{11}(   v_{2,c}  +  u_1)    +   \sqrt{P^{ \alpha }}  h_{12} h_{21} u_{2}   +  z_{1}    \non \\
         &=    \sqrt{P^{ 2- 4 \alpha/3 +  \epsilon}} \cdot 2\gamma \cdot ( \bar{g}_0 \bar{q}_0 + \sqrt{P^{ \alpha -1 }}\bar{g}_1 \bar{q}_1 + \sqrt{P^{2\alpha -2 }} \bar{g}_2 \bar{q}_2 )   +  z_{1}    \non\\
         &=    \sqrt{P^{ 2- 4 \alpha/3 +  \epsilon}} \cdot 2\gamma \cdot \bar{x}_s   +  z_{1}    \label{eq:yvk4835132}
\end{align}
where  $\bar{x}_s  \defeq ( \bar{g}_0 \bar{q}_0 + \sqrt{P^{ \alpha -1 }}\bar{g}_1 \bar{q}_1 + \sqrt{P^{2\alpha -2 }} \bar{g}_2 \bar{q}_2 )$ and 
 \[ \bar{g}_0\defeq  h_{11}h_{22},    \quad  \bar{g}_1\defeq   h_{12} h_{11}  , \quad   \bar{g}_2 \defeq h_{12} h_{21} \] 
    \[ \bar{q}_0  \defeq  \frac{Q_{\max}}{2\gamma} \cdot   v_{1,c}  ,    \quad  \bar{q}_1  \defeq  \frac{Q_{\max}}{2\gamma} \cdot   (   v_{2,c}  +  u_1), \quad   \bar{q}_2  \defeq  \frac{ Q_{\max}}{2\gamma} \cdot  u_{2},   \quad  Q_{\max} \defeq P^{ \frac{ \alpha/3 - \epsilon }{2}}\]
   for a given constant $\gamma  \in \bigl(0, \frac{1}{4\sqrt{2}}\bigr]$ (see \eqref{eq:gammadef}), $v_{1,c} \in \Omega ( \xi =  2\gamma \cdot \frac{ 1}{Q} ,   Q =  P^{ \frac{ \alpha/3 - \epsilon }{2}} )$,  $v_{2,c}  +  u_1 \in 2 \cdot \Omega ( \xi =  2\gamma \cdot \frac{ 1}{Q} ,  Q =  P^{ \frac{ \alpha/3 - \epsilon }{2}} )$  and $u_{2} \in \Omega ( \xi =  2\gamma \cdot \frac{ 1}{Q} ,   \   Q =  P^{ \frac{ \alpha/3 - \epsilon }{2}} ) $.
Based on our definitions, it holds true that $\bar{q}_0, \bar{q}_1, \bar{q}_2 \in \Zc$,   $|\bar{q}_0| \leq Q_{\max}$, $|\bar{q}_1| \leq 2 Q_{\max}$,  $|\bar{q}_2| \leq Q_{\max}$,   $ \sqrt{P^{ \alpha -1 }} \in  \Zc^+$ and $\sqrt{P^{ 2\alpha -2}} \in \Zc^+$ for this case with $1 \leq \alpha \leq 3/2$.
Let us consider the minimum distance for $\bar{x}_s$ defined as
  \begin{align}
\bar{d}_{\min}  (\bar{g}_0, \bar{g}_1, \bar{g}_2)   \defeq    \min_{\substack{ \bar{q}_0, \bar{q}_2, \bar{q}_0', \bar{q}_2' \in \Zc  \cap [- Q_{\max},    Q_{\max}]  \\  \bar{q}_1, \bar{q}_1' \in \Zc  \cap [- 2Q_{\max},   2 Q_{\max}]   \\  (\bar{q}_0, \bar{q}_1, \bar{q}_2) \neq  (\bar{q}_0', \bar{q}_1', \bar{q}_2')  }}  | \bar{g}_0  (\bar{q}_0 - \bar{q}_0') + \sqrt{P^{ \alpha - 1}} \bar{g}_1 (\bar{q}_1 - \bar{q}_1') + \sqrt{P^{2\alpha -2}} \bar{g}_2 (\bar{q}_2 - \bar{q}_2')  |.     \label{eq:minidis111132}
 \end{align}
The following Lemma~\ref{lm:distance132} provides  a result regarding the lower bound on the minimum distance $\bar{d}_{\min}$.

\begin{lemma}  \label{lm:distance132}
Consider the case $\alpha \in [1, 3/2]$, and  consider some constants $\delta \in (0, 1]$ and  $\epsilon >0$.   Given the signal design in \eqref{eq:constellationGsym1}-\eqref{eq:constellationGsym1u} and \eqref{eq:para13211}-\eqref{eq:xvkkk1132},  then the minimum distance $\bar{d}_{\min}$ defined in \eqref{eq:minidis111132} is bounded by
 \begin{align}
d_{\min}    \geq   \delta P^{- \frac{2- 4\alpha/3 }{2}}   \label{eq:distancegeq132}
 \end{align}
for all  the channel  coefficients $\{h_{k\ell}\} \in (1, 2]^{2\times 2} \setminus \Hob$, and the Lebesgue measure of the outage set  $\Hob \subseteq (1,2]^{2\times 2}$ , denoted by $\mathcal{L}(\Hob)$,  satisfies  
 \begin{align}
\mathcal{L}(\Hob) \leq 258048 \delta   \cdot     P^{ - \frac{ \epsilon  }{2}}.   \label{eq:LebmeaB132}
 \end{align}
 \end{lemma}
 \begin{proof}
This proof is very similar to the proof for Lemma~\ref{lm:distance3432}.
We will consider the case of $1 \leq \alpha \leq 3/2$. Let 
\[ \bar{\beta} \defeq  \delta  P^{- \frac{2- 4\alpha/3  }{2}} , \quad \bar{A}_1 \defeq   P^{\frac{\alpha -1}{2} } ,  \quad  \bar{A}_2 \defeq   P^{ \alpha -1} \]
 \[ \bar{g}_0\defeq  h_{11}h_{22},    \quad  \bar{g}_1\defeq   h_{12} h_{11}  , \quad   \bar{g}_2 \defeq h_{12} h_{21} \] 
\[ Q_0 \defeq 2 Q_{\max},    \quad Q_1\defeq 4 Q_{\max} , \quad Q_2 \defeq 2 Q_{\max},  \quad  Q_{\max} \defeq P^{ \frac{ \alpha/3 - \epsilon }{2}}       \]
for some $\epsilon >0$ and $\delta \in (0, 1]$.
Let us define the event
\begin{align}
\bar{B}(\bar{q}_0, \bar{q}_1, \bar{q}_2)  \defeq \{ (\bar{g}_0, \bar{g}_1, \bar{g}_2)  \in (1,4]^3 :  |  \bar{g}_0\bar{q}_0 + \bar{A}_1 \bar{g}_1 \bar{q}_1 + \bar{A}_2 \bar{g}_2 \bar{q}_2    | < \bar{ \beta} \}       \label{eq:Boutage11132}
\end{align}
and set 
\begin{align}
\bar{B}  \defeq   \bigcup_{\substack{ \bar{q}_0, \bar{q}_1, \bar{q}_2 \in \Zc:  \\  (\bar{q}_0, \bar{q}_1, \bar{q}_2) \neq  0,  \\  |\bar{q}_k| \leq Q_k  \ \forall k }}  B(\bar{q}_0, \bar{q}_1, \bar{q}_2) .   \label{eq:Boutage22132}
\end{align}
From  Lemma~\ref{lm:NMb},  the Lebesgue measure of $\bar{B}$, denoted by $\Lc (\bar{B})$,  is bounded by
\begin{align}
\Lc (\bar{B}) & \leq   504 \bar{\beta} \Bigl(  2 \min \bigl\{ 2 Q_{\max},  \frac{ 2 Q_{\max}}{ P^{ \alpha -1 }} \bigr\}  +  \tilde{Q}'_2 \cdot \min \bigl\{ 4 Q_{\max} ,  \frac{2 Q_{\max} }{P^{\frac{ \alpha -1}{2} }},  \frac{P^{ \alpha -1} \tilde{Q}'_2}{P^{\frac{\alpha -1}{2} }}   \bigr\} \non\\ & \quad\quad\quad + 2 \min \bigl\{ 4 Q_{\max},  \frac{ 2 Q_{\max}}{ P^{\frac{\alpha -1}{2} }} \bigr\}  + \tilde{Q}'_1 \cdot \min \bigl\{ 2 Q_{\max} ,  \frac{2 Q_{\max} }{P^{ \alpha -1}},  \frac{P^{\frac{  \alpha -1 }{2} }  \cdot 4 Q_{\max}}{P^{ \alpha -1}}   \bigr\}        \Bigr)   \non\\
 &= 504 \bar{\beta} \Bigl(  \frac{ 4 Q_{\max}}{ P^{ \alpha -1 }}  +  \tilde{Q}'_2 \cdot   \frac{2 Q_{\max} }{P^{\frac{ \alpha -1}{2} }}  + \frac{ 4 Q_{\max}}{ P^{\frac{ \alpha -1 }{2} }}   + \tilde{Q}'_1 \cdot \frac{2 Q_{\max} }{P^{ \alpha -1}}       \Bigr)  \non \\
  &= 504 \bar{\beta} \Bigl(  \frac{ 4 Q_{\max}}{ P^{  \alpha -1 }}  + 2 Q_{\max} \cdot \min\{1, 16 P^{-\frac{  \alpha -1}{2} }\}  \cdot   \frac{2 Q_{\max} }{P^{\frac{ \alpha -1}{2} }}  + \frac{ 4 Q_{\max}}{ P^{\frac{\alpha -1}{2} }}   + 4 Q_{\max} \cdot \frac{2 Q_{\max} }{P^{ \alpha -1}}       \Bigr)  \non\\
   &\leq  504 \bar{\beta} \Bigl(  \frac{ 4 Q_{\max}}{ P^{ \alpha -1 }}  + 2 Q_{\max} \cdot  16 \cdot   \frac{2 Q_{\max} }{P^{ \alpha -1 }}  + \frac{ 4 Q_{\max}}{ P^{\frac{ \alpha -1}{2} }}   + 4 Q_{\max} \cdot \frac{2 Q_{\max} }{P^{ \alpha -1 }}       \Bigr)   \non\\
   & \leq  504 \bar{\beta} \cdot    \frac{ 16Q_{\max}}{ P^{\frac{\alpha -1}{2} }}  \cdot \max \{  16  \frac{Q_{\max}}{ P^{\frac{ \alpha -1}{2} }} , 1 \}  \non\\ 
      & =  504 \bar{\beta} \cdot    16 P^{ \frac{ 1- 2\alpha/3  - \epsilon  }{2}} \cdot \max \{  16 P^{ \frac{ 1- 2\alpha/3  - \epsilon  }{2}} , 1 \}  \non\\
            & \leq  504 \bar{\beta} \cdot    16 P^{ \frac{1-  2\alpha/3  - \epsilon  }{2}} \cdot   16 P^{ \frac{1- 2\alpha/3    }{2}}   \non\\ 
             & =  504 \delta P^{- \frac{2- 4\alpha/3  }{2}}  \cdot    16 P^{ \frac{ 1- 2\alpha/3  - \epsilon  }{2}} \cdot   16 P^{ \frac{ 1- 2\alpha/3    }{2}}   \non\\   
                & =  129024 \delta   \cdot     P^{ - \frac{ \epsilon  }{2}}          \label{eq:LeBmeasure132}
\end{align}
for a constant $\delta \in (0, 1]$, where
\begin{align*}
\tilde{Q}'_1 &= \min\Bigl\{4 Q_{\max}, \   8 \cdot \frac{\max\{2 Q_{\max}, \  P^{ \alpha -1}\cdot 2 Q_{\max} \}}{P^{\frac{  \alpha -1 }{2} }}\Bigr\} = \min\Bigl\{4 Q_{\max},   \  8 P^{\frac{ \alpha -1}{2} }\cdot 2 Q_{\max} \Bigr\} =4 Q_{\max}\\
\tilde{Q}'_2 &\defeq \min\Bigl\{2 Q_{\max}, \   8\cdot \frac{\max\{2 Q_{\max},  \  P^{\frac{ \alpha -1}{2} } \cdot4 Q_{\max}\}}{ P^{ \alpha - 1}}\Bigr\}  = 2 Q_{\max} \cdot \min\{1,  \ 16 P^{-\frac{\alpha -1}{2} }\}.
\end{align*}
The set  $\bar{B}$ can be considered as an outage set. 
For any  triple $(\bar{g}_0, \bar{g}_1, \bar{g}_2)$  outside the outage set $\bar{B}$, i.e., $(\bar{g}_0, \bar{g}_1, \bar{g}_2)\notin \bar{B}$, it is apparent that $\bar{d}_{\min} (\bar{g}_0, \bar{g}_1, \bar{g}_2)   \geq   \delta P^{- \frac{2- 4\alpha/3 }{2}}$. 
In our setting $  \bar{g}_0\defeq  h_{11}h_{22}$,    $ \bar{g}_1\defeq   h_{12} h_{11}$ and  $\bar{g}_2 \defeq h_{12} h_{21} $. 
Let us define  $\Hob$  as the collection of the quadruples $(h_{11}, h_{12}, h_{22}, h_{21} ) \in (1, 2]^{2\times 2}$   such that the corresponding triples $(\bar{g}_0, \bar{g}_1, \bar{g}_2)$ are in the outage set $B$, that is,
\[ \Hob \defeq \{  (h_{11}, h_{12}, h_{22}, h_{21} ) \in (1, 2]^{2\times 2} :      (\bar{g}_0, \bar{g}_1, \bar{g}_2) \in \bar{B}  \} . \]
By following the similar steps in \eqref{eq:LeBmeasure882441}-\eqref{eq:LeBmeasure9999}, one can  bound the Lebesgue measure of $\Hob$ as: 
\begin{align}
\Lc (\Hob )  \leq  258048 \delta   \cdot     P^{ - \frac{ \epsilon  }{2}}.              \label{eq:LeBmeasure9999132}  
\end{align}
Now we complete the proof of this lemma.
 \end{proof}

At this point, we go back to the expression of $y_1$ in \eqref{eq:yvk4835132}, i.e., \[y_{1}   =    \sqrt{P^{ 2- 4 \alpha/3 +  \epsilon}} \cdot 2\gamma \cdot \bar{x}_s   +  z_{1}.\]
Based on our definition, the minimum distance for $\bar{x}_{s}$ is  $\bar{d}_{\min}$ defined in \eqref{eq:minidis111132}.  Lemma~\ref{lm:distance132} reveals that, under the channel condition $\{h_{k\ell}\} \notin \Hob$, the  minimum distance for $\bar{x}_{s}$ is bounded by  $\bar{d}_{\min}    \geq   \delta P^{- \frac{2- 4\alpha/3}{2}} $, which implies that  $\bar{x}_s$ can be estimated from $y_{1}$ and  the corresponding  error probability vanishes  as  $P\to \infty$.
Note that in this case, the minimum distance for $ \sqrt{P^{ 2- 4 \alpha/3 +  \epsilon}} \cdot 2\gamma \cdot \bar{x}_s$ is lower bounded by $\sqrt{P^{  \epsilon}} \cdot 2\gamma \delta$.
Once  $\bar{x}_{s} $ is decoded correctly, then the three symbols $ \bar{q}_0  \defeq  \frac{Q_{\max}}{2\gamma} \cdot   v_{1,c}$,   $\bar{q}_1  \defeq  \frac{Q_{\max}}{2\gamma} \cdot   (   v_{2,c}  +  u_1)$,   and  $\bar{q}_2  \defeq  \frac{ Q_{\max}}{2\gamma} \cdot  u_{2}$ can be recovered from $\bar{x}_s  = ( \bar{g}_0 \bar{q}_0 + \sqrt{P^{ \alpha -1 }}\bar{g}_1 \bar{q}_1 + \sqrt{P^{2\alpha -2 }} \bar{g}_2 \bar{q}_2 )$ due to the fact that $\bar{g}_0, \bar{g}_1, \bar{g}_2$ are rationally independent.
Then, we can conclude that the error probability for decoding $v_{k,c}$  is 
 \begin{align}
  \text{Pr} [ v_{k,c} \neq \hat{v}_{k,c} ]     \to 0  \quad  \text{as} \quad   P\to \infty  \label{eq:error155525132}                           
 \end{align}
 for $k=1$ and $1 \leq \alpha \leq 3/2$,  given the signal design in \eqref{eq:constellationGsym1}-\eqref{eq:constellationGsym1u} and  \eqref{eq:para13211}-\eqref{eq:xvkkk1132}.  This result holds true  for  all the channel coefficients  $\{h_{k\ell}\} \in (1, 2]^{2\times 2}$  except  for an outage set   $\Hob \subseteq (1,2]^{2\times 2}$ with Lebesgue measure  $\mathcal{L}(\Hob)$ satisfying  \[ \mathcal{L}(\Hob)  \to 0, \quad \text {as}\quad  P\to \infty \] 
(see \eqref{eq:LeBmeasure9999132}).
Due to the symmetry, the result in \eqref{eq:error155525132} also holds true for the case of $k=2$.

%


\end{document}